%% file: kr16.tex
\documentclass[letter]{article}

\usepackage{aaai16}
\usepackage{graphicx}
\usepackage[usenames,dvipsnames]{xcolor}
\usepackage{float}
\usepackage{amsthm,amstext,amssymb,amsmath}
\usepackage{mathtools} 
\usepackage{enumerate}
\usepackage{xspace}
\usepackage{pifont} 
\usepackage{misc}
\usepackage{thmtools,thm-restate}
\usepackage{boxedminipage}
\usepackage{tikz}
\usepackage{times}

\newcommand{\mc}[1]{\ensuremath{\mathcal{#1}}\xspace}

\renewcommand{\EL}{\mc{E\kern-0.1emL}}

\renewcommand{\phi}{\varphi}
\renewcommand{\epsilon}{\varepsilon}

\newcommand{\tiles}{\mathbb{T}}
\newcommand{\horiz}{\mathbb{H}}
\newcommand{\vertical}{\mathbb{V}}
\newcommand{\adom}{\ensuremath{\mathsf{adom}}}

\newcommand{\myeqnsep}{0.5mm}
\newcommand{\vect}[1]{\mathbf{#1}}

\declaretheorem{theorem}
\declaretheorem[sibling=theorem]{lemma}

\declaretheorem[sibling=theorem]{example}

\title{Containment in Monadic Disjunctive Datalog, MMSNP, and
  Expressive Description Logics}
  
\author{
      \hspace{18mm}Pierre Bourhis \\
      \hspace{18mm}CNRS CRIStAL, UMR 8189, Universit\'e Lille 1, \\
	      \hspace{18mm}INRIA Lille, France \\
      \hspace{18mm}pierre.bourhis@univ-lille1.fr
      \And 
      \hspace{18mm}Carsten Lutz \\
      \hspace{18mm}University of Bremen, Germany \\
      \hspace{18mm}clu@uni-bremen.de 
}

\begin{document}
\maketitle

\begin{abstract}
  We study query containment in three closely related formalisms:
  monadic disjunctive Datalog (MDDLog), MMSNP (a logical
  generalization of constraint satisfaction problems), and
  ontology-mediated queries (OMQs) based on expressive description
  logics and unions of conjunctive
  queries.  
  Containment in MMSNP was known to be decidable due to a result by
  Feder and Vardi, but its exact complexity has remained open. We
  prove {\sc 2NExpTime}-completeness and extend this result to monadic
  disjunctive Datalog and to OMQs. 
\end{abstract}

\section{Introduction}

In knowledge representation with ontologies, data centric applications
have become a significant subject of research.
In such applications, ontologies are used to address
incompleteness and heterogeneity of the data, and for enriching it
with background knowledge
\cite{DBLP:conf/rweb/CalvaneseGLLPRR09}. This trend has given rise to
the notion of an \emph{ontology-mediated query (OMQ)} which combines
a database query with an ontology, often formulated in a
description logic
(DL). 
From a data centric viewpoint, an OMQ can be viewed as a
normal database query 
that happens to consist of two components (the ontology and the
actual query).
 It is thus
natural to study OMQs in the same way as other query languages,
aiming to understand e.g.\ their expressive power and the complexity
of fundamental reasoning tasks such as query containment. In this
paper, we concentrate on the latter.

Containment of OMQs was first studied in
\cite{DBLP:conf/ecai/LevyR96,DBLP:conf/pods/CalvaneseGL98} and more
recently in
\cite{DBLP:conf/ijcai/CalvaneseOS11,DBLP:journals/tods/BienvenuCLW14,DBLP:conf/kr/BienvenuLW12}. To
appreciate the usefulness of this reasoning task, it is important to
recall that real-world ontologies can be very large and tend to change
frequently. As a user of OMQs, one might thus want to know whether the
ontology used in an OMQ can be replaced with a potentially much
smaller module extracted from a large ontology or with a newly released
version of the ontology, without compromising query answers. This
requires to decide equivalence of OMQs, which can be done by
answering two containment questions. Containment can also serve as a
central reasoning service when optimizing OMQs in static
analysis \cite{DBLP:conf/kr/BienvenuLW12}. 

In the most general form of OMQ containment, the two OMQs can
involve different ontologies and the data schema (ABox signature, in DL
terms) can be restricted to a subset of the signature of the
ontologies. While results for this form of containment have been
obtained for inexpressive DLs such as those of the DL-Lite and
$\mathcal{EL}$ families \cite{DBLP:conf/kr/BienvenuLW12},
containment of OMQs based on expressive DLs turned out to be a
technically challenging problem.
A step forward has been made in
\cite{DBLP:journals/tods/BienvenuCLW14} where it was observed that
there is a close relationship between three groups of formalisms:
(i)~OMQs based on expressive DLs, (ii)~monadic disjunctive Datalog
(MDDLog) programs, and (iii)~constraint satisfaction problems (CSPs)
as well as their logical generalization MMSNP. These observations have
given rise to first complexity results for containment of OMQs based
on expressive DLs, namely {\sc NExpTime}-completeness for several
cases where the actual query is an atomic query of the
form $A(x)$, with $A$ a monadic relation.

In this paper, we study containment in MDDLog, MMSNP, and OMQs that
are based on expressive DLs, conjunctive queries (CQ), and unions
thereof (UCQs).  A relevant result is due to Feder and Vardi
(\citeyear{DBLP:journals/siamcomp/FederV98}) who show that containment
of MMSNP sentences is decidable and that this gives rise to
decidability results for CSPs such as whether the complement of a CSP
is definable in monadic Datalog. As shown in
\cite{DBLP:journals/tods/BienvenuCLW14}, the complement of MMSNP is
equivalent to Boolean MDDLog programs and to Boolean OMQs with UCQs as
the actual query.
While these results can be used to infer decidability of containment
in the mentioned query languages, they do not immediately yield tight
complexity bounds. In particular, Feder and Vardi describe their
algorithm for containment in MMSNP only on a very high level of
abstraction, do not analyze its complexity, and do not attempt to
provide lower bounds.  Also a subsequent study of MMSNP containment
and related problems did not clarify the precise complexity
\cite{DBLP:conf/cp/Madelaine10}.  Other issues to be addressed are
that MMSNP containment corresponds only to the containment of
\emph{Boolean} queries and that the translation of OMQs into MMSNP
involves a double exponential blowup.

Our main contribution is to show that all of the mentioned containment
problems are 2{\sc NExpTime}-complete. In particular, this is the case
for MDDLog, MMSNP, OMQs whose ontology is formulated in a DL between
\ALC and $\mathcal{SHI}$ and where the actual queries are UCQs, and
OMQs whose ontology is formulated in a DL between \ALCI and
$\mathcal{SHI}$ and where the actual queries are CQs.  This closes
open problems from \cite{DBLP:conf/cp/Madelaine10} about MMSNP
containment and from \cite{DBLP:conf/kr/BienvenuLW12} about OMQ
containment. In addition, clarifying the complexity of MDDLog
containment is interesting from the perspective of database theory,
where Datalog containment has received a lot of attention. While being
undecidable in general \cite{DBLP:journals/jlp/Shmueli93}, containment
is known to be decidable for monadic Datalog
\cite{DBLP:conf/stoc/CosmadakisGKV88} and is in fact 2{\sc
  ExpTime}-complete \cite{DBLP:conf/icalp/BenediktBS12}.  Here, we
show that adding disjunction increases the complexity to 2{\sc
  NExpTime}. We refer to \cite{DBLP:conf/ijcai/BourhisKR15} for
another recent work that generalizes monadic Datalog containment, in
an orthogonal direction. It is interesting to note that all these
previous works rely on the existence of witness instances for
non-containment which have a tree-like shape. In MDDLog containment,
such witnesses are not guaranteed to exist which results in
significant technical challenges.

This paper is structured as follows. We first concentrate on MDDLog
containment, establishing that containment of a Boolean MDDLog program
in a Boolean conjunctive query (CQ) is 2{\sc NExpTime}-hard and that
containment in Boolean MDDLog is in 2{\sc NExpTime}. We then
generalize the upper bound to programs that are non-Boolean and admit
constant symbols. The lower bound uses a reduction of a tiling problem
and borrows 
queries from \cite{DBLP:conf/mfcs/BjorklundMS08}. For the upper bound,
we essentially follow the arguments of Feder and Vardi
(\citeyear{DBLP:journals/siamcomp/FederV98}), but provide full details
of all involved constructions and carefully analyze the involved
blowups. It turns out that an MDDLog containment question $\Pi_1
\subseteq \Pi_2$ can be decided non-deterministically in time single
exponential in the size of $\Pi_1$ and double exponential in the size
of $\Pi_2$. Together with some straightforward observations, this also
settles the complexity of MMSNP containment. We additionally observe
that FO- and Datalog-rewritability of MDDLog programs and (the
complements of) MMSNP sentences is 2{\sc NExpTime}-hard.

We then consider containment between OMQs, starting with the
observation that the 2{\sc NExpTime} lower bound for MDDLog also
yields that containment of an OMQ in a CQ is 2{\sc NExpTime}-hard when
ontologies are formulated in $\mathcal{ALC}$ and UCQs are used as
queries. The same is true for the containment of an OMQ in an OMQ, even
when their ontologies are identical. We then establish a matching upper
bound by translating OMQs to MDDLog and applying our results for
MDDLog containment. It is interesting that the complexity is double
exponential only in the size of the actual query (which tends to be
very small) and only single exponential in the size of the ontology.
We finally establish another 2{\sc NExpTime} lower bound which applies
to containment of OMQs whose ontologies are formulated in
$\mathcal{ALCI}$ and whose actual queries are CQs (instead of UCQs as
in the first lower bound).
This requires a different reduction strategy which borrows queries from
\cite{DBLP:conf/cade/Lutz08}.

Due to space limitations, we defer proof details to the appendix,
available at http://www.informatik.uni-bremen.de/tdki/research/papers.html.

\section{Preliminaries}

A \emph{schema} is a finite collection $\Sbf=(S_1,\dots,S_k)$ of
relation symbols with associated non-negative arity. An
\emph{\Sbf-fact}
is an expression of the form $S(a_1,\ldots, a_n)$ where $S\in\Sbf$ is
an $n$-ary relation symbol, and $a_1, \ldots, a_n$ are elements of
some fixed, countably infinite set $\mn{const}$ of \emph{constants}.
An \emph{\Sbf-instance} $I$ is a finite set of \Sbf-facts.  The
\emph{active domain} $\adom(I)$ of~$I$ is the set of all constants
that occur in the facts in~$I$.
%

An \emph{\Sbf-query} is semantically defined as a mapping $q$ that
associates with every \Sbf-instance $I$ a set of \emph{answers} $q(I)
\subseteq \adom(I)^{n}$, where $n\geq 0$ is the \emph{arity}
of~$q$. If $n=0$, then we say that $q$ is \emph{Boolean} and we write
$I \models q$ if $()\in q(I)$. We now introduce some concrete query
languages. A \emph{conjunctive query (CQ)} takes the form $\exists
\vect{y} \, \vp(\vect{x}, \vect{y})$ where $\vp$ is a conjunction of
relational atoms and \xbf, \ybf denote tuples of variables. The
variables in \xbf are called \emph{answer variables}.  Semantically,
$\exists
\vect{y} \, \vp(\vect{x}, \vect{y})$ denotes the query
$$
q(I) = \{ (a_{1},\ldots,a_{n})\in \adom(I)^{n} \mid I\models \varphi[a_{1},\ldots,a_{n}]\}.
$$
%
A \emph{union of conjunctive queries (UCQ)} is a disjunction of CQs
with the same free variables.  We now define disjunctive Datalog
programs,
see also \cite{DBLP:journals/tods/EiterGM97}.
A \emph{disjunctive Datalog rule} $\rho$ has
the form
\[S_1(\vect{x}_1) \vee \cdots \vee S_m(\vect{x}_m) \leftarrow
R_1(\vect{y}_1)\land \cdots\land R_n(\vect{y}_n)\]
where $n> 0$ and $m \geq 0$.\footnote{Empty rule heads (denoted
  $\bot$) are sometimes disallowed. We admit them only in our upper bound
  proofs, but do not use them for lower bounds, thus achieving maximum
  generality.  }
We refer to $S_1(\vect{x}_1) \vee \cdots \vee S_m(\vect{x}_m)$ as the
\emph{head} of $\rho$, and to $R_1(\vect{y}_1) \wedge \cdots \wedge
R_n(\vect{y}_n)$ as the \emph{body}. Every variable that occurs in the
head of a rule $\rho$ is required to also occur in the body of $\rho$.
A \emph{disjunctive Datalog (DDLog) program} $\Pi$ is a finite set of
disjunctive Datalog rules with a selected \emph{goal relation}
\mn{goal} that does not occur in rule bodies and appears only in
non-disjunctive \emph{goal rules} $\mn{goal}(\vect{x}) \leftarrow
R_1(\vect{x}_1) \wedge \cdots \wedge R_n(\vect{x}_n)$.  The
\emph{arity of $\Pi$} is the arity of the 
\mn{goal} relation.  Relation symbols that occur in the head of at
least one rule of $\Pi$ are \emph{intensional (IDB) relations}, and
all remaining relation symbols in $\Pi$ are \emph{extensional (EDB)
  relations}.  Note that, by definition, \mn{goal} is an IDB
relation. A DDLog program is called \emph{monadic} or an \emph{MDDLog
  program} if all its IDB relations except \mn{goal} have arity at most one.

An $\Sbf$-instance, with $\Sbf$ the set of all (IDB and EDB)
relations in $\Pi$, is a \emph{model} of $\Pi$ if it satisfies all
rules in~$\Pi$. We use $\mn{Mod}(\Pi)$ to denote the set of all models
of~$\Pi$.
%
%
Semantically, a DDLog program $\Pi$ of arity $n$ defines the
following query over the schema $\Sbf_E$ that consists of the EDB
relations of~$\Pi$: for every $\Sbf_E$-instance $I$,
$$
\begin{array}{r@{}l}
\Pi(I) =
\{\vect{a}\in \adom(I)^{n}\mid \; &\mn{goal}(\vect{a})\in
J  \text{ for all }  \\[\myeqnsep]
& J\in \mn{Mod}(\Pi) \text{ with } I\subseteq J \}.
\end{array}
$$ 
Let $\Pi_1,\Pi_2$ be DDLog programs over the same EDB schema $\Sbf_E$
and of the same arity.  We say that $\Pi_1$ \emph{is contained in}
$\Pi_2$, written $\Pi_1 \subseteq \Pi_2$, if for every $\Sbf_E$-instance
$I$, we have $\Pi_1(I) \subseteq \Pi_2(I)$.
\begin{example}
  Consider the following MDDLog program~$\Pi_1$ over EDB schema
  $\Sbf_E=\{A,B,r\}$:
  $$
  \begin{array}{rcl}
    A_1(x) \vee A_2(x) &\leftarrow& A(x)
    \\[\myeqnsep]
    \mn{goal}(x) &\leftarrow&
    A_1(x) \wedge r(x,y) \wedge A_1(y)  \\[\myeqnsep]
    \mn{goal}(x) &\leftarrow&
    A_2(x) \wedge r(x,y) \wedge A_2(y)  \\[\myeqnsep]
  \end{array}
  $$
  Let $\Pi_2$ consist of the single rule $\mn{goal}(x)
  \leftarrow B(x)$. Then $\Pi_1 \not\subseteq \Pi_2$ is witnessed,
  for example, by the $\Sbf_E$-instance $I=\{r(a,a),A(a)\}$. It is
  interesting to note that there is no tree-shaped $\Sbf_E$-instance
  that can serve as a witness although all rule bodies in $\Pi_1$ and
  $\Pi_2$ are tree-shaped. In fact, a tree-shaped instance does not
  admit any answers to $\Pi_1$ because we can alternate $A_1$
  and $A_2$ with the levels of the tree, avoiding to make $\mn{goal}$
  true anywhere.
\end{example}
%
%
An \emph{MMSNP sentence} over schema $\Sbf_E$ has the form
$
\exists X_1 \cdots \exists X_n \forall x_1 \cdots \forall x_m \vp
$
with $X_1,\dots,X_n$ monadic second-order variables,
$x_1,\dots,x_m$ first-order variables, and $\vp$ a conjunction of
formulas of the form
$$
\alpha_1 \wedge \cdots
\wedge \alpha_n \rightarrow \beta_1 \vee \cdots \vee \beta_m \mbox{ with $n,m \geq 0$},
$$
where each $\alpha_i$ takes the form $X_i(x_j)$ or $R(\vect{x})$ with
$R \in \Sbf_E$, and each $\beta_i$ takes the form $X_i(x_j)$. This
presentation is syntactically different from, but semantically
equivalent to the original definition from
\cite{DBLP:journals/siamcomp/FederV98}, which does not use the
implication symbol and instead restricts the allowed polarities of
atoms. An MMSNP sentence $\varphi$ can serve as a Boolean query in the
obvious way, that is, $I \models \varphi$ whenever $\varphi$ evaluates
to true on the instance $I$. The containment problem in MMSNP coincides
with logical implication. See 
\cite{DBLP:journals/siamdm/BodirskyCF12,DBLP:journals/jcss/BodirskyD13}
for more information on MMSNP.

It was shown in \cite{DBLP:journals/tods/BienvenuCLW14} 
that the complement of an MMSNP sentence can be translated into an
equivalent Boolean MDDLog program in polynomial time and vice versa. 
The involved complementation is irrelevant for the purposes of deciding
containment since for any two Boolean queries $q_1,q_2$, we have $q_1
\subseteq q_2$ if and only if $\neg q_1 \not\supseteq \neg
q_2$. Consequently, any upper bound for containment in MDDLog also
applies to MMSNP and so does any lower bound for containment between
Boolean MDDLog programs.
\begin{example}
  Let $\Sbf_E = \{ r \}$, $r$ binary. The complement of the MMSNP
  formula $\exists R \exists G \exists B \forall x \forall y \,\psi$
  over $\Sbf_E$ with $\psi$ the conjunction of
  $$
  \begin{array}{c}
    \top \rightarrow R(x) \vee G(x) \vee B(x) \\[\myeqnsep]
    C(x) \wedge r(x,y) \wedge C(y) \rightarrow \bot \quad \text{ for } C
    \in \{ R,G,B \}
  \end{array}
  $$
  is equivalent to the Boolean MDDLog program
  $$
  \begin{array}{rl}
    r(x,y) \rightarrow C(x) \vee \overline{C}(x) &\text{ for } C 
    \in \{ R,G,B \} \\[\myeqnsep]
    r(x,y) \rightarrow C(y) \vee \overline{C}(y) &\text{ for } C 
    \in \{ R,G,B \} \\[\myeqnsep]
    \overline{R}(x) \wedge \overline{G}(x) \wedge \overline{B}(x) \rightarrow \mn{goal}() \\[\myeqnsep]
    C(x) \wedge r(x,y) \wedge C(y) \rightarrow \mn{goal}() & \text{ for } C
    \in \{ R,G,B \}.
  \end{array}
  $$
\end{example}
%
%


\section{MDDLog and MMSNP: Lower Bounds}

The first main aim of this paper is to establish the following result.
Point~3 closes an open problem from \cite{DBLP:conf/cp/Madelaine10}.
\begin{theorem}
\label{thm:hardness1}    
The following containment problems are 2\NExpTime-complete:
  \begin{enumerate}

  \item of an MDDLog program in a CQ;

  \item of an MDDLog program in an MDDLog program;

  \item of two MMSNP sentences.


  \end{enumerate}
\end{theorem}
We prove the lower bounds by
reduction of a tiling problem.  
It 
suffices to show 
that containment between a Boolean MDDLog program and a Boolean CQ is
2\NExpTime-hard.  \newcommand{\torus}{square\xspace} A \emph{2-exp
  \torus tiling problem} is a triple $P=(\tiles,\horiz,\vertical)$
where
\begin{itemize}
\item $\tiles=\{T_{1},\ldots,T_{p}\}$, $p \geq 1$, is a finite set of
  \emph{tile types};
\item $\horiz \subseteq \tiles \times \tiles$ is a
\emph{horizontal matching relation}; 
\item $\vertical \subseteq \tiles \times
\tiles$ is a \emph{vertical matching relation}. 
\end{itemize}
An input to $P$ is a word $w \in \tiles^*$. Let $w=T_{i_0} \cdots
T_{i_n}$. A \emph{tiling for $P$ and $w$} is a map
$f:\{0,\dots,2^{2^n}-1\} \times \{0,\dots,2^{2^n}-1\} \rightarrow \tiles$
such that
$f(0,j)=T_{i_j}$ for $0 \leq j \leq n$,
$(f(i,j),f(i+1,j)) \in \horiz$ for $0 \leq i < 2^{2^n}$, and 
$(f(i,j),f(i,j+1)) \in \vertical$ for $0 \leq i < 2^{2^n}$. 
%
It is 2\NExpTime-hard to decide, given a 2-exp \torus tiling problem
$P$ and an input $w$ to $P$, whether there is a tiling for $P$ and
$w$.


For the reduction, let $P$ be a 2-exp \torus tiling problem and $w_0$
an input to $P$ of length~$n$. We construct a Boolean MDDLog
program~$\Pi$ and a Boolean CQ $q$ such that $\Pi \subseteq q$ iff
there is a tiling for $P$ and $w_0$. To get a first intuition, assume
that instances $I$ have the form of a (potentially partial) $2^{2^n}
\times 2^{2^n}$-grid in which the horizontal and vertical positions of
grid nodes are identified by binary counters, described in more detail
later on. We construct $q$ such that $I \models q$ iff $I$ contains a
counting defect, that is, if the counters in $I$ are not properly
incremented or assign multiple counter values to the same node. $\Pi$
is constructed such that on instances $I$ without counting defects, 
$I \models \Pi$ 
iff 
the partial grid in $I$ does not admit a tiling for $P$
and~$w_0$. Note that this gives the desired result: if there is no
tiling for $P$ and $w_0$, then an instance $I$ that represents the
full $2^{2^n} \times 2^{2^n}$-grid (without counting defects) shows
$\Pi \not\subseteq q$; conversely, if there is a tiling for $P$ and
$w_0$, then $I \not\models q$ means that there is no counting
defect in $I$ and thus $I \not\models \Pi$.

We now detail the exact form of the grid and the counters.
Some of the constants in the input instance serve as \emph{grid nodes}
in the $2^{2^n} \times 2^{2^n}$-grid while other constants serve 
different purposes described below. 
%
%
To identify the position of a grid node $a$, we use a binary counter
whose value is stored at the $2^m$ leaves of a binary \emph{counting
  tree} with root $a$ and depth $m:=n+1$. The depth of counting trees
is $m$ instead of $n$ because we need to store the horizontal position
(first $2^n$ bits of the counter) as well as the vertical position
(second $2^n$ bits). The binary relation $r$ is used to connect
successors. To distinguish left and right successors, every left
successor $a$ has an attached \emph{left navigation gadget}
$r(a,a_1),r(a_1,a_2),\mn{jump}(a,a_2)$ and every right successor $a$
has an attached \emph{right navigation gadget}
$r(a,a_1),\mn{jump}(a,a_1)$---these gadgets will be used in the
formulation of the query $q$ later on.

If a grid node $a_2$ represents the right neighbor of grid node $a_1$,
then there is some node $b$ such that $r(a_1,b), r(b,a_2)$. The node
$b$ is called a \emph{horizontal step node}. Likewise, if $a_2$
represents the upper neighbor of $a_1$, then there must also be some
$b$ with $r(a_1,b), r(b,a_2)$ and we call $b$ a \emph{vertical step
  node}. In addition, for each grid node $a$ there must be a node
$c$ such that $r(a,c), r(c,a)$ and we call $c$ a \emph{self step
  node}. We make sure that, just like grid nodes, all three types of
step node have an attached counting tree.
\begin{figure}[t!]
  \begin{center}
    \framebox[1\columnwidth]{
 \footnotesize
\begin{tikzpicture}[->,font=\scriptsize]
\fill[black!10!white](0.1,0.9) -- (0.25,0.3) -- (0.65,0.3)-- cycle; 
\fill[black!10!white](2.1,0.8) -- (2.25,0.3) -- (2.75,0.3)-- cycle; 
\node (as) at (-1,1) {$c_1$};
\node (as2) at (-1,3) {$c_3$};
\node (as3) at (1,1) {$c_2$};
\node (as4) at (1,3) {$c_4$};
\node (p1) at (-1,0) {};
\node (p2) at (0,-1) {};
\node (a) at (0,0) {$a_1$};
\node (a1) at (1,0) {$b_1$};
\node (a2) at  (2,0) {$a_2$};
\node (b1) at (0,1) {$b_2$};
\node (b2) at (0,2) {$a_3$};

\node(sa2h) at (3,0) {};
\node (sa2v1) at (2,1) {$b_3$};
\node (sa2v2) at (2,-1) {};
\node(sb2v) at (0,3) {};
\node (d1) at (1,2) {$b_4$};
\node (sb2h2) at (-1,2) {};
\node (d2) at (2,2)  {$a_4$};
\node (sd2h) at (2,3) {};
\node (sd2v) at (3,2) {};

\draw (a) -- (a1) node[midway,above] {$r$}; 
\draw (a1) -- (a2)  node[midway,above]  {$r$};  
\draw (a) to[bend angle = 15, bend right]  node[midway,above] {$r$} (as) ; 
\draw (as) to[bend angle = 15,bend right]  node[midway,above] {$r$} (a); 
\draw (b2) to[bend angle = 15, bend right]  node[midway,above] {$r$} (as2) ; 
\draw (as2) to[bend angle = 15,bend right]  node[midway,above] {$r$} (b2); 
\draw (d2) to[bend angle = 15, bend right]  node[midway,above] {$r$} (as4) ; 
\draw (as4) to[bend angle = 15,bend right]  node[midway,above] {$r$} (d2); 
\draw (a2) to[bend angle = 15, bend right]  node[midway,above] {$r$} (as3) ; 
\draw (as3) to[bend angle = 15,bend right]  node[midway,above] {$r$} (a2); 
\draw (a) -- (b1) node[midway,above,right] {$r$}; 
\draw (b1) -- (b2)  node[midway,above,left]  {$r$}; 
\draw (b2) -- (d1)  node[midway,above]  {$r$}; 
\draw (d1) -- (d2)  node[midway,above]  {$r$}; 
\draw (a2) -- (sa2v1)  node[midway,above,right]  {$r$}; 
\draw (sa2v1) -- (d2)  node[midway,above,left]  {$r$}; 

\draw[dashed] (p1) -- (a);
\draw[dashed] (p2) -- (a) ;
\draw[dashed] (a2) -- (sa2h);
\draw[dashed] (a2) -- (sa2v1) ;
\draw[dashed] (sa2v2) -- (a2);
\draw[dashed] (b2) -- (sb2v) ;
\draw[dashed] (d2) -- (sd2v);
\draw[dashed] (d2) -- (sd2h);
\draw[dashed] (sb2h2) -- (b2) ;

\fill[black!10!white](0.1,-0.1) -- (0.25,-0.7) -- (0.65,-0.7)-- cycle; 
\fill[black!10!white](1.1,-0.1) -- (1.25,-0.7) -- (1.65,-0.7)-- cycle; 
\fill[black!10!white](2.1,1.9) -- (2.25,1.3) -- (2.65,1.3)-- cycle; 
\fill[black!10!white](0.1,1.9) -- (0.25,1.3) -- (0.65,1.3)-- cycle; 
\fill[black!10!white](1.1,1.9) -- (1.25,1.3) -- (1.65,1.3)-- cycle;
\fill[black!10!white](2.1,1.9) -- (2.25,1.3) -- (2.65,1.3)-- cycle;
\fill[black!10!white](2.1,-0.1) -- (2.25,-0.7) -- (2.65,-0.7)-- cycle;

\fill[black!10!white](-1.1,0.9) -- (-0.9,0.3) -- (-1.3,0.3)-- cycle; 
\fill[black!10!white](0.9,0.9) -- (1.2,0.30) -- (0.8,0.30)-- cycle; 
\fill[black!10!white](-1.1,2.9) -- (-0.9,2.3) -- (-1.3,2.3)-- cycle; 
\fill[black!10!white](0.9,2.9) -- (1.2,2.30) -- (0.8,2.30)-- cycle; 

\end{tikzpicture}
}
    \caption{A grid cell (counting trees in grey).}
    \label{fig:cell}
  \end{center}
\end{figure}
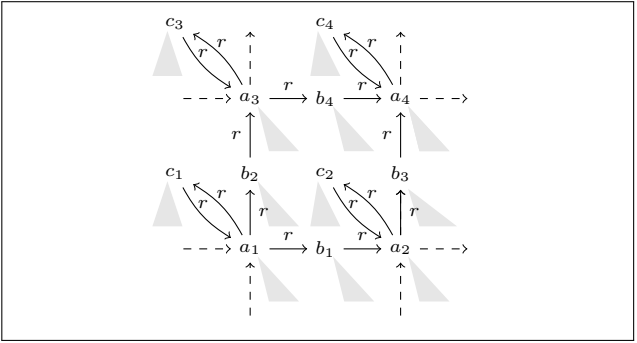
Figure~\ref{fig:cell} illustrates the representation of a single grid cell.

We need to make sure that counters are properly incremented when
transitioning to right and upper neighbors via step nodes. To achieve
this, each counting tree actually stores two counter values via
monadic relations $B_1, \overline{B}_1$ (first value) and
$B_2,\overline{B_2}$ (second value) at the leaves of the tree, where
$B_i$ indicates bit value one and $\overline{B}_i$ bit value zero.
While the $B_1$-value represents the actual position of the node in
the grid, the $B_2$-value is copied from the $B_1$-values of its
predecessor nodes (which must be identical). In fact,
%
%
%
the query $q$ to be defined later shall guarantee that 
\begin{description}

\item[(Q1)] whenever $r(a_1,a_2)$ and $a_1$ is associated (via a
  counting tree) with $B_1$-value $k_1$ and $a_2$ is associated with
  $B_2$-value $k_2$, then $k_1=k_2$;
%

\item[(Q2)] every node is associated (via counting trees) with at most
  one $B_1$-value. 

\end{description}
Between neighboring grid and step nodes, counter values are thus
copied as described in (Q1) above, but not incremented.
Incrementation takes place inside counting trees, as follows: at grid
nodes and at self step nodes, the two values are identical; at
horizontal (resp.\ vertical) step nodes, the $B_1$-value is obtained
from the $B_2$-value by incrementing the horizontal part and keeping
the vertical part (resp.\ incrementing the vertical part and keeping
the horizontal part). 

We now construct the program $\Pi$. As the EDB schema,
we use $\Sbf_E = \{ r, \mn{jump}, B_1,
B_2,\overline{B}_1,\overline{B}_2 \}$ where $r$ and \mn{jump} are
binary and all other relations are monadic.  We first define rules
which verify that a grid or self step node has a proper counting tree
attached to it (in which both counters are identical):
$$
\begin{array}{r@{\;}c@{\;}l}
  \mn{left}(x) & \leftarrow & r(x,y) \wedge r(y,z)
  \wedge\mn{jump}(x,z) \label{eq:leftgadget}\\[\myeqnsep]
  \mn{right}(x) &\leftarrow& r(x,y) \wedge
  \mn{jump}(x,y) \label{eq:rightgadget}  \\[\myeqnsep]
\mn{lrok}(x) &\leftarrow& \mn{left}(x) \\[\myeqnsep]
\mn{lrok}(x) &\leftarrow& \mn{right}(x) \label{eq:endLRblock} \\[\myeqnsep]
\mn{lev}^G_m(x) &\leftarrow& B_1(x) \wedge B_2(x) \wedge 
\mn{lrok}(x) \label{eq:gactivelastlev1} \\[\myeqnsep]
\mn{lev}^G_m(x) &\leftarrow&\overline{B}_1(x) \wedge \overline{B}_2(x) \wedge \mn{lrok}(x) \label{eq:gactivelastlev2}\\[\myeqnsep]
%
%
%
\mn{lev}^G_i(x) &\leftarrow& r(x,y_1) \wedge \mn{lev}^G_{i+1}(y_1) \wedge \mn{left}(y_1) 
  \, \wedge \\[\myeqnsep]
  &&r(x,y_2) \wedge  \mn{lev}^G_{i+1}(y_2) \wedge \mn{right}(y_2) 
\end{array}
$$
for $0 \leq i < m$.  We call a constant $a$ of an instance $I$
\emph{g-active} if it has all required structures attached to serve as
a grid node. 
Such constants are marked by the IDB
relation \mn{gactive}: 
$$
\mn{gactive}(x) \leftarrow \mn{lev}^G_0(x) \wedge r(x,y) \wedge \mn{lev}^G_0(y)
\wedge r(y,x) 
$$
We also want horizontal and vertical step nodes to be roots of the
required counting trees. 
The difference to the counting trees below grid / self step nodes is
that we need to increment the counters. This requires
modifying the rules with head relation $\mn{lev}^G_i$ above. We only
consider horizontal step nodes explicitly as vertical ones are very
similar.  The relations $B_1,\overline{B}_1$ and $B_2,\overline{B_2}$
give rise to a labeling of the leaf nodes that defines a word over
the alphabet $\Sigma=\{0,1\}^2$ where symbol $(i,j)$ means that the
bit encoded via $B_1,\overline{B}_1$ has value $i$ and the bit encoded
via $B_2,\overline{B}_2$ has value $j$. Ensuring that the
$B_1$-value is obtained by incrementation from the $B_2$-value
(least significant bit at the left-most leaf) then corresponds to
enforcing that the leaf word is from the regular language
$L=(0,1)^*(1,0)((0,0) + (1,1))^*$.
To
achieve this, we consider the languages $L_1=(0,1)^*$, $L_2=L$, and
$L_3=((0,0) + (1,1))^*$. Instead of level relations $\mn{lev}^G_i$, we
use relations $\mn{lev}^{H,\ell}_i$ where $\ell \in \{1,2,3\}$
indicates that the leaf word of the subtree belongs to the language
$L_{\ell}$:
$$
\begin{array}{r@{\;}c@{\;}l}
  \mn{lev}^{H,1}_m(x) &\leftarrow& \overline{B}_1(x) \wedge B_2(x) \wedge
                                \mn{lrok}(x) \\[\myeqnsep]
  \mn{lev}^{H,2}_m(x) &\leftarrow& B_1(x) \wedge \overline{B}_2(x) \wedge
                                \mn{lrok}(x) \\[\myeqnsep]
  \mn{lev}^{H,3}_m(x) &\leftarrow& B_1(x) \wedge B_2(x) \wedge
                                \mn{lrok}(x) \\[\myeqnsep]
  \mn{lev}^{H,3}_m(x) &\leftarrow& \overline{B}_1(x) \wedge
                                \overline{B}_2(x) \wedge \mn{lrok}(x) \\[\myeqnsep]
  \mn{lev}^{H,\ell_3}_i(x) &\leftarrow& r(x,y_1) \wedge
  \mn{lev}^{H,\ell_1}_{i+1}(y_1) \wedge \mn{left}(y_1) \, \wedge \\[\myeqnsep]
  &&r(x,y_2) \wedge  \mn{lev}^{H,\ell_2}_{i+1}(y_2) \wedge \mn{right}(y_2) 
\end{array}
$$ 
where $1 \leq i < m$ and $(\ell_1,\ell_2,\ell_3) \in \{
(1,1,1),(1,2,2)$, $(2,3,2), (3,3,3) \}$.  We call a constant of an
instance \emph{h-active} if it is the root of a counting tree that
implements incrementation of the horizontal position (left subtree of
the root) and does not change the vertical position (right subtree of
the root), identified by the IDB relation \mn{hactive}:
%
$$
\begin{array}{r@{\;}c@{\;}l}
\mn{hactive}(x) &\leftarrow&
r(x,y_1) \wedge \mn{lev}^{H,\ell_2}_{1}(y_1) \wedge \mn{left}(y_1) 
  \, \wedge \\[\myeqnsep]
&&r(x,y_2) \wedge  \mn{lev}^{H,\ell_3}_{1}(y_2) \wedge \mn{right}(y_2) 
\end{array}
$$
We omit the rules for the corresponding IDB relation $\mn{vactive}$.
Call the fragment of $\Pi$ that we have constructed up to this point
$\Pi_{\mn{tree}}$. 

Recall that we want an instance to make $\Pi$
true if it admits no tiling for $P$ and
$w$. 
We thus label all g-active nodes with a tile type:
$$
\bigvee_{T_i \in \tiles} T_i(x) \leftarrow \mn{gactive}(x) 
$$
It then remains to trigger the \mn{goal} relation whenever there is a
defect in the tiling. Thus add for all $T_i,T_j \in \tiles$ with
$(T_i,T_j) \notin H$:
  $$
  \begin{array}{r@{\;}c@{\;}l}
\mn{goal}() 
&\leftarrow&
     T_i(x) \wedge \mn{gactive}(x) 
 \wedge r(x,y) 
     \wedge \mn{hactive}(y) 
     \, \wedge\\[\myeqnsep]
&& r(y,z)      
\wedge T_j(z) \wedge \mn{gactive}(z) 
  \end{array}
  $$
  and for all $T_i,T_j \in \tiles$ with $(T_i,T_j) \notin V$:
  $$
  \begin{array}{r@{\;}c@{\;}l}
\mn{goal}() 
 &\leftarrow& 
     T_i(x) \wedge \mn{gactive}(x) 
 \wedge r(x,y) \wedge  \mn{vactive}(y) \,
     \wedge \\[\myeqnsep]
&& r(y,z)      
\wedge T_j(z) \wedge \mn{gactive}(z) 
  \end{array}
  $$
  The last kind of defect concerns the initial condition. Let $w_0 =
  T_{i_0} \cdots T_{i_{n-1}}$.  It is tedious but not difficult to
  write rules which ensure that, for all $ i < n$, every g-active
  element whose $B_1$-value represents horizontal position $i$
  and vertical position 0 satisfies the monadic IDB relation
  $\mn{pos}_{i,0}$.  We then put for all $j < n$ and all $T_\ell
  \in \tiles$ with $T_\ell \neq T_{i_j}$:
  $$
   \mn{goal}() \leftarrow \mn{pos}_{j,0}(x) \wedge T_\ell(x).
  $$
We now turn to the definition of $q$; recall that we want it
to achieve conditions (Q1) and (Q2)
above.  
%
%
%
%
Due to the presence of self step nodes and since the counting trees
below self step nodes and grid nodes must have identical values for
the two counters, it can be verified that (Q1) implies (Q2).
%
%
Therefore, we only need to achieve (Q1).  We use as $q$
a minor variation of a CQ constructed in
\cite{DBLP:conf/mfcs/BjorklundMS08} for a similar purpose.  We first
constuct a UCQ and show in the appendix how to replace it with a
CQ, which also involves some minor additions to the program
$\Pi_{\mn{tree}}$ above.

The UCQ $q$ makes essential use of the left and right navigation
gadgets in counting trees. It uses a subquery $q_m(x,y)$ constructed
such that $x$ and $y$ can only be mapped to corresponding leaves in
successive counting trees, that is, (i)~the roots of the trees are
connected by the relation $r$ and (ii)~$x$ can be reached from the
root of the first tree by following the same sequence of left and
right successors that one also needs to follow to reach $y$ from the
root of the second tree. To define $q_m(x,y)$, we inductively define
queries $q_i(x,y)$ for all $i \leq m$, starting with
$q_0(x,y)=r(x_0,y_0)$ and setting, for $0 < i \leq m$,
$$
\begin{array}{r@{}l}
q_i(x_i,&y_i) = \exists x_{i-1} \exists y_{i-1}
\exists z_{i,0} \cdots \exists z_{i,i+2}
\exists z'_{i,1} \cdots \exists z'_{i,i+3} \\[\myeqnsep]
  & q_{i-1}(x_{i-1},y_{i-1}) \wedge r(x_{i-1},x_i) \wedge r(y_{i-1},y_i)  \, \wedge \\[\myeqnsep]
  & \mn{jump}(x_i,z_{i,i+2}) \wedge \mn{jump}(y_i,z'_{i,i+3})   \, \wedge \\[\myeqnsep]
  & r(z_{i,0},z_{i,1}) \wedge \cdots \wedge r(z_{i,i+1},z_{i,i+2})   \, \wedge \\[\myeqnsep]
  & r(z_{i,0},z'_{i,1}) \wedge r(z_{i,1},z'_{i,2}) \wedge \cdots \wedge r(z'_{i,i+2},z'_{i,i+3})
\end{array}
$$
The $r$-atom in $q_0$ corresponds to the move from the root of one
counting tree to the root of a successive tree, the atoms
$r(x_{i-1},x_i)$ and $r(y_{i-1},y_i)$ in $q_i$ correspond to moving
down the $i$-th step in both trees, and the remaining atoms in $q_i$
make sure that both of these steps are to a left successor or to a
right successor. We make essential use of the \mn{jump} relation here,
which shortcuts an edge on the path to the root for left successors,
but not for right successors.  Additional explanation is provided in
the appendix.  It is now easy to define the desired UCQ that achieves
(Q1):
$$
\begin{array}{r@{\,}c@{\,}l}
q&=&\exists x_m \exists y_m \, q_m(x_m,y_m) \wedge B_1(x_m) \wedge
\overline{B}_2(y_m) \\[\myeqnsep]
&& \vee \, \exists x_m \exists y_m \, q_m(x_m,y_m) \wedge \overline{B}_1(x_m) \wedge
B_2(y_m)
\end{array}
$$
The first CQ in $q$ is displayed in Figure~\ref{fig:mainquery}.  
\begin{figure}[t!]
  \begin{center}
    \framebox[1\columnwidth]{
\begin{tikzpicture}[->,font=\scriptsize]
\node (x0) at (0,0) {$x_0$};
\node (x1) at (-2.4,-1) {$x_1$} ;
\node 	(x2) at (-2.4,-2) {$x_2$} ;
\node 	(xm1) at (-2.4,-3) {$x_{m-1}$} ;
\node (xm) at (-2.4,-4.5) {$x_m$} ;

\node(xb) at (-2.8, -4.6) {$B_1$};
\node(xb) at (2.8, -4.6) {$\overline{B}_2$};

\node (u1) at (-1,-1.5) {$z_{1,3}$};
\node (u2) at (-1,-2.5) {$z_{2,4}$} ;
\node 	(um1) at (-1,-3.5) {$z_{m-1,m+1}$} ;
\node 	(um) at (-1,-5) {$z_{m,m+2}$} ;

\node (y0) at (1,-0.3) {$y_0$};
\node (y1) at (2.4,-1) {$y_1$} ;
\node 	(y2) at (2.4,-2) {$y_2$} ;
\node 	(ym1) at (2.4,-3) {$y_{m-1}$} ;
\node 	(ym) at (2.4,-4.5) {$y_m$} ;

\node (v1) at (1,-1.5){$z_{1,4}$};
\node (v2) at (1,-2.5) {$z_{2,5}$} ;
\node 	(vm1) at (1,-3.5) {$z_{m-1,m+2}$} ;
\node 	(vm) at (1,-5) {$z_{m,m+3}$}  ;

\node (z1) at (0,-0.5) {$z_{1,0}$};
\node (z2) at (0,-1.5) {$z_{2,0}$} ;
\node 	(zm1) at (0,-2.5) {$z_{m-1,0}$} ;
\node 	(zm) at (0,-4) {$z_{m,0}$} ;

\draw (x0) to[bend right]   node[midway,left]  {$r$} (x1); 
\draw (x1) -- (x2)  node[midway,left]  {$r$}; 
\draw [dashed] (x2) -- (xm1); 
\draw (xm1) -- (xm)  node[midway,left]  {$r$}; 

\draw (x0) to  node[midway,right,above]  {$r$} (y0) ; 
\draw (y0) to[bend left]  node[midway,right]  {$r$} (y1) ; 
\draw (y1) -- (y2)  node[midway,right]  {$r$}; 
\draw [dashed] (y2) -- (ym1); 
\draw (ym1) -- (ym)  node[midway,right]  {$r$}; 

\draw (x1) -- (u1)  node[sloped, midway,above]  {$jump$}; 
\draw (x2) -- (u2)  node[sloped, midway,above]  {$jump$};  
\draw (xm1) -- (um1)  node[sloped, midway,above]  {$jump$}; 
\draw (xm) -- (um)  node[sloped, midway,above]  {$jump$}; 

\draw (y1) -- (v1)  node[sloped, midway,above]  {$jump$}; 
\draw (y2) -- (v2)  node[sloped, midway,above]  {$jump$};  
\draw (ym1) -- (vm1)  node[sloped, midway,above]  {$jump$}; 
\draw (ym) -- (vm)  node[sloped, midway,above]  {$jump$};

\draw (z1) -- (u1)  node[sloped,midway,above]  {$r^3$}; 
\draw (z2) -- (u2)  node[sloped,midway,above]  {$r^4$}; 
\draw (zm1) -- (um1) node[sloped,midway,above]  {$r^{m+1}$}	; 
\draw (zm) -- (um)  node[sloped,midway,above]  {$r^{m+2}$}; 

\draw (z1) -- (v1)  node[sloped,midway,above]  {$r^4$}; 
\draw (z2) -- (v2)  node[sloped,midway,above]  {$r^5$}; 
\draw (zm1) -- (vm1) node[sloped,midway,above]  {$r^{m+2}$}	; 
\draw (zm) -- (vm)  node[sloped,midway,above]  {$r^{m+3}$}; 
\end{tikzpicture}
}
    \caption{The first CQ in $q$.}
    \label{fig:mainquery}
  \end{center}
\end{figure}
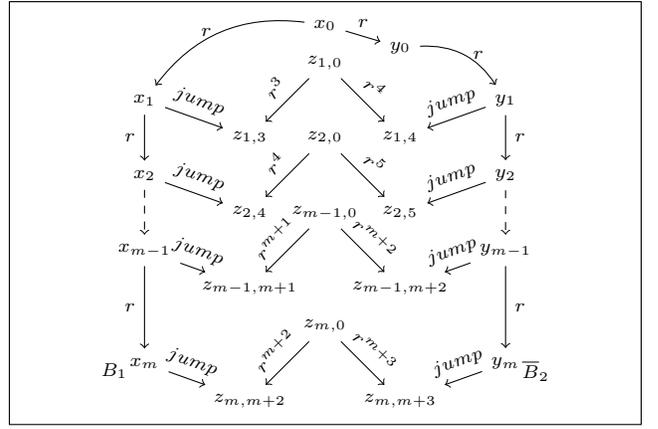

%
%
%
\begin{restatable}{lemma}{LEMcontlowercorr}\label{lem:contlowercorr}
  $\Pi \not\subseteq q$ iff there is no tiling for $P$
  and $w_0$.
\end{restatable}
This finishes the proof of the lower bounds stated in
Theorem~\ref{thm:hardness1}.  Before proceeding, we note that the
lower bound can be adapted to important rewritability questions. A
query is \emph{FO-rewritable} if there is an equivalent first-order
query and \emph{(monadic) Datalog-rewritable} if there is an
equivalent (non-disjunctive) (monadic) Datalog query. FO-Rewritability
of a query is desirable since it allows to use conventional SQL
database systems for query answering, and likewise for
Datalog-rewritability and Datalog engines.  For this reason, FO- and
Datalog-rewritability have received a lot of attention. For example,
they have been studied for OMQs in
\cite{DBLP:journals/tods/BienvenuCLW14} and for CSPs in
\cite{DBLP:journals/siamcomp/FederV98,DBLP:journals/lmcs/LaroseLT07}.
Monadic Datalog is an interesting target as it constitutes an
extremely well-behaved fragment of Datalog.  It is open whether the
known decidability of FO- and (monadic) Datalog-rewritability
generalizes from CSPs to MMSNP. We observe here that these problems
are at least 2{\sc NExpTime}-hard. The proof is by a simple
modification of the reduction presented above.
\begin{restatable}{theorem}{THMrewrite}
\label{thm:rewrite}
  For MDDLog programs and the complements of MMSNP sentences,
  rewritability into FO, into monadic Datalog, and into Datalog are 2{\sc NExpTime}-hard.
\end{restatable}

\section{MDDLog and MMSNP: Upper Bounds}

The aim of this section is to establish the upper bounds stated in
Theorem~\ref{thm:hardness1}. It suffices to concentrate on MDDLog
since the result for MMSNP follows.  We first consider only Boolean
MDDLog programs and then show how to extend the upper bound to MDDLog
programs of any arity.  

Our main algorithm is essentially the one described in
\cite{DBLP:journals/siamcomp/FederV98}. Since the
constructions are described by Feder and Vardi only on an extremely
high level of abstraction and without providing any analysis of the
algorithm's running time, we give full details and proofs (in the
appendix).  The algorithm for deciding $\Pi_1 \subseteq \Pi_2$
proceeds in three steps. First, $\Pi_1$ and~$\Pi_2$ are converted into
a simplified form, then containment between the resulting programs
$\Pi^S_1$ and $\Pi^S_2$ is reduced to a certain emptiness problem, and
finally that problem is decided. A technical complication is posed by
the fact that the construction of $\Pi^S_1$ and $\Pi^S_2$ 
does not preserve containment in a strict sense.  In fact, $\Pi_1
\subseteq \Pi_2$ only implies $\Pi^S_1 \subseteq \Pi^S_2$ on instances
of a certain minimum girth. To address
this issue, we have to be careful about the girth in all three steps
and can finally resolve the problem in the last
step. 

We now define the notion of girth. For an $n$-ary relation symbol $S$,
$\mn{pos}(S)$ is $\{1, \ldots, n\}$. A finite structure $I$ has a
\emph{cycle} of length~$n$ if it contains distinct facts
$R_0(\vect{a}_0),\dots, R_{n-1}(\vect{a}_{n-1})$, $\abf_i=a_{i,1}
\cdots a_{i,m_i}$, and there are positions $p_i,p'_i \in
\mn{pos}(R_i)$, $0 \leq i < n$ such that:
\begin{itemize}
\item $p_i \neq p'_i$ for $1 \leq i \leq n$;
\item $a_{i, p'_i}=a_{i\oplus1, p_{i\oplus1}}$ for $0 \leq i < n$,
  where
  $\oplus$ denotes addition modulo $n$.
\end{itemize}
The \emph{girth} of $I$ is the length of the shortest cycle in it and
$\infty$ if $I$ has no cycle (in which case we say that $I$ is a \emph{tree}).

%
%
%
For MDDLog programs $\Pi_1,\Pi_2$ over the same EDB schema and $k \geq
0$, we write $\Pi_1 \subseteq_{>k} \Pi_2$ if $\Pi_1(I) \subseteq
\Pi_2(I)$ for all \Sbf-instances of girth exceeding~$k$.

Throughout the proof, we have to carefully analyze the running time of
the algorithm, considering various measures for MDDLog
programs. The \emph{size} of an MDDLog program $\Pi$, denoted $|\Pi|$,
is the number of symbols needed to write $\Pi$ where relation and
variable names are counted as having length one.  The \emph{rule size}
of an MDDLog program is the maximum size of a rule in $\Pi$.  The
\emph{atom width} (resp.\ \emph{variable width}) of $\Pi$ is the
maximum number of atoms in any rule body (resp.\ variables in any
rule) in $\Pi$.

  

\subsubsection*{From Unrestricted to Simple Programs}

An MDDLog program $\Pi^S$ is \emph{simple} if it satisfies the
following conditions:
\begin{enumerate}

\item 
every rule in $\Pi^S$ comprises at most one EDB atom and this atom
contains all variables of the rule body, each variable exactly once;

\item rules without an EDB atom contain at most a single variable.

\end{enumerate} 
The conversion to simple form changes the EDB schema and thus the
semantics of the involved queries, but it (almost) preserves
containment, as detailed by the next theorem. The theorem is implicit
in \cite{DBLP:journals/siamcomp/FederV98} and our contribution is to
analyze the size of the constructed MDDLog programs and to
provide detailed proofs. The same applies to the other theorems stated
in this section.
\begin{restatable}{theorem}{THMsimplify}\label{thm:simplify}
\label{thm:simplify}
Let $\Pi_1,\Pi_2$ be Boolean MDDLog programs over EDB schema
$\Sbf_E$. Then one can construct simple Boolean MDDLog programs
$\Pi_1^S,\Pi_2^S$ over EDB schema $\Sbf'_E$ such that
\begin{enumerate}

\item $\Pi_1 \not\subseteq \Pi_2$ implies $\Pi^S_1 \not\subseteq \Pi^S_2$;
 
\item $\Pi^S_1 \not\subseteq_{>w} \Pi^S_2$ implies $\Pi_1
  \not\subseteq_{>w} \Pi_2$
 
\end{enumerate}
where $w$ is the atom width of $\Pi_1 \cup \Pi_2$. Moreover,  if
$r$ is the number of rules in $\Pi_1 \cup \Pi_2$ and $s$ the rule
size, then
  %
  \begin{enumerate}

  \setcounter{enumi}{2}

  %
  %

  \item $|\Pi_i^S| \leq p(r \cdot 2^s)$;


   \item the variable width of 
     $\Pi_i^S$ is bounded by that of $\Pi_i$;


  \item $|\Sbf'_E| \leq p(r \cdot 2^s)$;

  \end{enumerate}
  where $p$ is a polynomial.  The construction takes time polynomial
  in $|\Pi^S_1 \cup \Pi^S_2|$.
\end{restatable}
A detailed proof of Theorem~\ref{thm:simplify} is given in the
appendix.  Here, we only sketch the construction, which consists of
several steps. We concentrate on a single Boolean MDDLog program
$\Pi$.  In the first step, we extend $\Pi$ with all rules that can be
obtained from a rule in $\Pi$ by consistently identifying variables.
We then split up each rule in $\Pi$ into multiple rules by introducing
fresh IDB relations whenever
this is possible. After this second step, we obtain a
program which satisfies the following conditions:
\begin{enumerate}

\item[(i)] all rule bodies are biconnected,
that is, when any single variable is removed from the body (by
deleting all atoms that contain it), then the resulting rule body is
still connected; 

\item[(ii)] if $R(x,\dots,x)$ occurs in a rule body with $R$ EDB, then
  the body contains no other EDB atoms.

\end{enumerate}
In the final step, we transform every rule as follows: we replace all
EDB atoms in the rule body by a single EDB atom that uses a fresh EDB
relation which represents the conjunction of all atoms replaced.
Additionally, we need to take care of implications between the new EDB
relations, which gives rise to additional rules. The last step of the
conversion is the most important one, and it is the reason for why we
can only use instances of a certain girth in Point~2 of
Theorem~\ref{thm:simplify}. Assume, for example, that, before the last
step, the program had contained the following rules, where $A$ and $r$
are EDB relations:
$$
\begin{array}{r@{\;}c@{\;}l}
  P(x_3) 
&\leftarrow&
  A(x_1) \wedge r(x_1,x_2) \wedge r(x_2,x_3) \wedge r(x_3,x_1)
\\[\myeqnsep]
\mn{goal}()
&\leftarrow& 
  r(x_1,x_2) \wedge r(x_2,x_3) \wedge r(x_3,x_1)
  \, \wedge \\[\myeqnsep]
  && P(x_1)   \wedge P(x_2)   \wedge P(x_3)
\end{array}
$$
A new ternary EDB relation $R_{q_2}$ is introduced for the EDB body
atoms of the lower rule, where $q_2 = r(x_1,x_2) \wedge r(x_2,x_3)
\wedge r(x_3,x_1)$, and a new ternary EDB relation $R_{q_1}$ is
introduced for the upper rule, $q_1 = A(x_1) \wedge q_2$. Then the
rules are replaced with
$$
\begin{array}{r@{\;}c@{\;}l}
P(x_3) 
&\leftarrow& 
  R_{q_1}(x_1,x_2,x_3) 
\\[\myeqnsep]
\mn{goal}() 
&\leftarrow& 
  R_{q_2}(x_1,x_2,x_3)    \wedge P(x_1)   \wedge P(x_2)   \wedge
  P(x_3)
 \\[\myeqnsep]
\mn{goal}() 
&\leftarrow& 
  R_{q_1}(x_1,x_2,x_3)    \wedge P(x_1)   \wedge P(x_2)   \wedge
  P(x_3)
\end{array}
$$
Note that $q_1 \subseteq q_2$, which results in two copies of the
goal rule to be generated. To understand the issues with girth,
consider the $\Sbf'_E$-instance $I$ defined by
$$
  R_{q_1}(a,a',c'), R_{q_1}(b,b',a'), R_{q_1}(c,c',b').
$$
The goal rules from the simplified program do not apply. But
when translating into an $\Sbf_E$-instance $J$ in the obvious way, the
goal rule of the original program does apply. The intuitive
reason is that, when we translate $J$ back to $I$, we get additional
facts $R_{q_2}(a',b',c')$, $R_{q_2}(b',c',a')$, $R_{q_2}(c',a',b')$
that are `missed' in $I$. Such effects can only happen on instances
whose girth is at most $w$, such as $I$. 

\subsubsection*{From Containment to Relativized Emptiness}

A \emph{disjointness constraint} is a rule of the form
$\bot \leftarrow P_1(\vect{x}) \wedge \cdots \wedge P_n(\vect{x})$ where all relations are of the same arity and at most unary. Let
$\Pi$ be a Boolean MDDLog program over EDB schema $\Sbf_E$ and $D$ a
set of disjointness constraints over $\Sbf_E$. We say that $\Pi$ is
\emph{semi-simple w.r.t.~}$D$ if $\Pi$ is simple when all relations
that occur in~$D$ are viewed as IDB relations. We say that $\Pi$ is
\emph{empty w.r.t.~}$D$ if for all $\Sbf_E$-instances~$I$ with $I
\models D$, we have $I \not\models \Pi$. The problem of relativized
emptiness is to decide, given a Boolean MDDLog program $\Pi$ and a set
of disjointness constraints $D$ such that $\Pi$ is semi-simple w.r.t.\
$D$, whether $\Pi$ is empty w.r.t.~$D$.
\begin{theorem}
\label{thm:toempty} 
Let $\Pi_1,\Pi_2$ be simple Boolean MDDLog programs over EDB schema
$\Sbf_E$. Then one can construct a Boolean MDDLog program $\Pi$ over EDB
schema $\Sbf'_E$ and a set of disjointness constraints $D$ over
$\Sbf'_E$ such that $\Pi$ is semi-simple w.r.t.~$D$ and
\begin{enumerate}
  
\item if $\Pi_1 \not\subseteq \Pi_2$, then $\Pi$ is non-empty w.r.t.\ $D$;

\item if $\Pi$ is non-empty w.r.t.\ $D$ on instances of girth $>g$,
  for some $g>0$,  then $\Pi_1 \not\subseteq_{>g} \Pi_2$;


\end{enumerate}
Moreover, 
\begin{enumerate}
\setcounter{enumi}{2}


 \item $|\Pi| \leq |\Pi_1| \cdot 2^{|\Sbf_{I,2}|\cdot v_1}$, $|D| \leq 
   \Omc(|\Sbf_{I,2}|)$;

\item  the variable width of $\Pi \cup D$ is bounded by
the variable width of $\Pi_1 \cup \Pi_2$; 

\item $|\Sbf'_E| \leq |\Sbf_E|+|\Pi_2|$.

\end{enumerate}
where $v_1$ is the variable width of $\Pi_1$ and $\Sbf_{I,2}$ is the
IDB schema of $\Pi_2$.  The construction takes time polynomial in
$|\Pi \cup D|$.
\end{theorem}
Note that, in Point~2 of Theorem~\ref{thm:toempty}, girth one
instances are excluded. 

To prove Theorem~\ref{thm:toempty}, let $\Pi_1,\Pi_2$ be simple
Boolean MDDLog programs over EDB schema $\Sbf_E$. For $i \in \{1,2\}$,
let $\Sbf_{I,i}$ be the set of IDB relations in $\Pi_i$ with goal
relations $\mn{goal}_i \in \Sbf_{I,i}$ and assume w.l.o.g.\ that
$\Sbf_{I,1} \cap \Sbf_{I,2} = \emptyset$. Set $\Sbf'_{E} := \Sbf_E
\cup \Sbf_{I,2} \cup \{ \overline{P} \mid P \in \Sbf_{I,2} \}$.
%
%
%
%
The MDDLog program $\Pi$ is constructed in two steps. We first add to
$\Pi$ every rule that can be obtained from a rule $\rho$ in $\Pi_1$ by
extending the rule body with
\begin{itemize}

\item $P(x)$ or $\overline{P}(x)$, for every variable $x$ in $\rho$
  and every unary $P \in \Sbf_{I,2}$, and 

\item $P()$ or $\overline{P}()$ for every nullary $P \in \Sbf_{I,2}$;
  for $P=\mn{goal}_2$, we always include $\overline{P}()$ but never
  $P()$.

\end{itemize}
In the second step of the construction, we remove from $\Pi$ every
rule $\rho$ whose body being true implies that a rule from $\Pi_2$ is
violated, that is, there is a rule whose body is the CQ $q(\vect{x})$
and with head $P_1(\vect{y}_1) \vee \cdots \vee P_n(\vect{y}_n)$ and a
variable substitution $\sigma$ such that\footnote{Of course, each
  $\vect{y}_i$
consists of either zero or one variable.}
\begin{itemize}

\item $\sigma(q)$ is a subset of the body of $\rho$ and

\item $\overline{P_i}(\sigma(\vect{y}_i))$ is in the body of $\rho$,
  for $1 \leq i \leq n$.

\end{itemize}
%
%
The goal relation of $\Pi$ is $\mn{goal}_1()$. The set of disjointness
constraints $D$ then consists of all rules
$\bot \leftarrow P(x) \wedge \overline{P}(x)$ for each unary
$P \in \Sbf_{I,2}$ and $\bot \leftarrow P() \wedge \overline{P}()$
for each nullary $P \in \Sbf_{I,2}$. It is not hard to verify that
$\Pi$ and $D$ satisfy the size bounds from
Theorem~\ref{thm:toempty}. 
We show in the appendix that $\Pi$
satisfies Points~1 and~2 of Theorem~\ref{thm:toempty}.

\subsubsection*{Deciding Relativized Emptiness}

We now show how to decide emptiness of an MDDLog program $\Pi$ 
w.r.t.\ a set of disjointness constraints $D$ assuming that $\Pi$ is
semi-simple w.r.t.\ $D$.
%
%
\begin{theorem}
\label{thm:emptinessCompl} 
Given a Boolean MDDLog program $\Pi$ over EDB schema $\Sbf_E$ and a set of
disjointness constraints $D$ over $\Sbf_E$ such that $\Pi$ is semi-simple
w.r.t.\ $D$, one can decide non-deterministically in time
$\Omc(|\Pi|^3) \cdot 2^{\Omc(|D|\cdot v)}$ whether $\Pi$ is empty
w.r.t.~$D$, where $v$ is the variable width of~$\Pi$.
%
\end{theorem}
Let $\Pi$ be a Boolean MDDLog program over EDB schema $\Sbf_E$ and let
$D$ be a set of disjointness constraints over $\Sbf_E$ such that $\Pi$
is semi-simple w.r.t.\ $D$. To prove Theorem~\ref{thm:emptinessCompl},
we show how to construct a finite set of $\Sbf_E$-instances satisfying
$D$ such that $\Pi$ is empty w.r.t.\ $D$ if and only if it is empty in
the constructed set of instances. Let $\Sbf_D$ be the set of all EDB
relations that occur in $D$. For $i \in \{0,1\}$, an \emph{i-type} 
is a set $t$ of $i$-ary relation symbols from $\Sbf_D$ such that $t$
does not contain all EDB relations that co-occur in a disjointness
rule in $D$.  The \emph{0-type of an instance} $I$ is the set $\theta$
of all nullary $P \in \Sbf_D$ with $P() \in I$. For each constant $a$
of $I$, we use $t_a$ to denote the \emph{1-type that $a$ has in} $I$,
that is, $t_a$ contains all unary $P \in \Sbf_D$ with $P(a) \in I$.

We build an $\Sbf_E$-instance $K_\theta$ for each 0-type $\theta$. The
elements of $K_\theta$ are exactly the 1-types and $K_\theta$ consists
of the following facts:
\begin{itemize}

\item $P(t)$ for each 1-type $t$ and each $P \in t$;

\item $R(t_1,\dots,t_n)$ for each relation $R \in \Sbf_E \setminus
  \Sbf_D$
  and all 1-types $t_1,\dots,t_n$;

\item $P()$ for each nullary $P \in \theta$.
  
\end{itemize}
Note that, by construction, $K_\theta$ is an $\Sbf_E$-instance that
satisfies all constraints in~$D$.
%
%
%
\begin{restatable}{lemma}{LEMKthetaAreEnough}\label{lem:KthetaAreEnough}
$\Pi$ is empty w.r.t.\ $D$ iff $K_\theta \not\models \Pi$ for all
0-types $\theta$.
%
%
%
\end{restatable}
By Lemma~\ref{lem:KthetaAreEnough}, we can decide emptiness of $\Pi$
by constructing all instances $K_\theta$ and then checking whether
$K_\theta \not\models \Pi$. The latter is done by guessing an
extension $K'_\theta$ of $K_\theta$ to the IDB relations in $\Pi$ that
does not contain the goal relation, and then verifying by an iteration
over all possible homomorphisms from rule bodies in $\Pi$ to
$K'_\theta$ that all rules in $\Pi$ are satisfied in $K'_\theta$. 
\begin{restatable}{lemma}{LEMruntime}
  The algorithm for deciding relativized emptiness runs in time
  $\Omc(|\Pi|^3) \cdot 2^{\Omc(|D|\cdot v)}$.
\end{restatable}
%
%
%
%
%

We still have to address the girth restrictions in
Theorems~\ref{thm:simplify} and~\ref{thm:toempty}, which are not
reflected in Theorem~\ref{thm:emptinessCompl}.  In fact, it suffices
to observe that relativized emptiness is independent of the girth of
witnessing structures. This is made precise by the following result.
\begin{restatable}{lemma}{THMemptinessGirth}\label{thm:emptinessGirth}
For every Boolean MDDLog program $\Pi$ over EDB schema $\Sbf_E$ and set of
disjointness constraints $D$ over $\Sbf_E$ such that $\Pi$ is
semi-simple w.r.t.\ $D$, the following are equivalent for any $g \geq
0$:
\begin{enumerate}

\item  $\Pi$ is empty regarding $D$ and 

\item $\Pi$ is empty regarding $D$ and instances of girth
  exceeding~$g$.

\end{enumerate}
\end{restatable}
The proof of Lemma~\ref{thm:emptinessGirth} uses a translation of
semi-simple MDDLog programs with disjointness constraints into a
constraint satisfaction problem (CSP) and invokes a combinatorial
lemma by Feder and Vardi (and, originally, Erd\H{o}s), to transform
instances into instances of high girth while preserving certain
homomorphisms.

\subsubsection*{Deriving Upper Bounds} 
We exploit the results just obtained to derive upper complexity
bounds, starting with Boolean MDDLog programs and MMSNP sentences.  In
the following theorem, note that for deciding $\Pi_1 \subseteq \Pi_2$,
the contribution of $\Pi_2$ to the complexity is exponentially larger
than that of $\Pi_1$.
\begin{restatable}{theorem}{THMmainupper}\label{thm:mainupper} 
Containment between Boolean MDDLog programs and between MMSNP
sentences is
 in {\sc 2NExpTime}.
More precisely, for Boolean MDDLog programs $\Pi_1$ and $\Pi_2$,
it can be decided non-deterministically in time $2^{2^{p(|\Pi_2| \cdot
    \mn{log}|\Pi_1|)}}$ whether $\Pi_1 \subseteq \Pi_2$, $p$ 
a polynomial.
\end{restatable}
We now extend Theorem~\ref{thm:mainupper} to MDDLog programs of
unrestricted arity. Since this is easier to do when constants can be
used in place of variables in rules, we actually generalize
Theorem~\ref{thm:mainupper} by allowing both constants in rules and
unrestricted arity. For clarity, we speak about MDDLog$^c$ programs
whenever we allow constants in rules. First, we show how to (Turing)
reduce containment between MDDLog$^c$ programs of unrestricted arity
to containment between Boolean MDDLog$^c$ programs. The idea 
essentially is to replace answer variables with fresh constants.

Let $\Pi_1$, $\Pi_2$ be MDDLog$^c$ programs of arity $k$ 
and let \Cbf be the set of constants in $\Pi_1 \cup \Pi_2$, extended
with $k$ fresh constants. We define Boolean MDDLog$^c$ programs
$\Pi^\abf_1$, $\Pi^\abf_2$ for each tuple $\abf$ over \Cbf of
arity~$k$. If $\abf=(a_1,\dots,a_k)$, then $\Pi^\abf_i$ is obtained from
$\Pi_i$ by modifying each goal rule $\rho = \mn{goal}(\vect{x})
\leftarrow q$ with $\vect{x}=(x_1,\dots,x_k)$ as follows:
\begin{itemize}

\item if there are $i,j$ such that $x_i=x_j$ and $a_i \neq a_j$, then
  discard $\rho$;

\item otherwise, replace $\rho$ with $\mn{goal}() \leftarrow q'$ where
  $q'$ is obtained from $q$ by replacing each $x_i$ with $a_i$. 

\end{itemize}
In the appendix, we show the following. 
\begin{restatable}{lemma}{LEMconstredok}\label{lem:constredok}
  $\Pi_1 \subseteq \Pi_2$ iff $\Pi^\abf_1 \subseteq \Pi^\abf_2$ for all
  $\abf \in \Cbf^k$.
\end{restatable}
This provides the desired Turing reduction to the Boolean case, with
constants. Note that the size of $\Pi^\abf_i$ is bounded by that of
$\Pi_i$, and likewise for all other relevant measures. The number of
required containment tests is bounded by $2^{|\Pi_1 \cup \Pi_2|^2}$, a
factor that is absorbed by the bounds in Theorem~\ref{thm:mainupper}.

It remains to reduce containment between Boolean MDDLog$^c$ programs
to containment between Boolean MDDLog programs. The general idea is to
replace constants with fresh monadic EDB relations. Of course, we have
to be careful because the extension of these fresh relations in an
instance need not be a singleton set. Let $\Pi_1$, $\Pi_2$ be Boolean
MDDLog$^c$ programs over EDB schema $\Sbf_E$ and let \Cbf be the set
of constants in $\Pi_1 \cup \Pi_2$. The EDB schema $\Sbf'_E$ is
obtained by extending $\Sbf_E$ with a monadic relation $R_a$ for each
$a \in \Cbf$. 
For
$i \in \{1,2\}$, the Boolean MDDLog program $\Pi'_i$ over EDB schema
$\Sbf'_E$ contains all rules that can be obtained from a rule $\rho$
from~$\Pi_i$ by choosing a partial function $\delta$ that maps the terms
(variables and constants) in $\rho$ to the relations in $\Sbf'_E
\setminus \Sbf_E$ such that each constant $a$ is mapped to $R_a$ and
then
\begin{enumerate}

\item replacing every occurrence of a term $t \in \mn{dom}(\delta)$
in the body of $\rho$ with a fresh variable and every
  occurrence of $t$ in the head of $\rho$ with one of the variables
  introduced for $t$ in the rule body;

%

\item adding $R_a(x)$ to the rule body whenever some occurrence of
  a variable $x_0$ in the original rule has been replaced with $x$
  and $\delta(x_0)=R_a$.



\end{enumerate}
For example, the rule 
$P_1(y) \vee P_2(y) \leftarrow r(x,y,y) \wedge s(y,z)
$ 
in
$\Pi_i$ gives rise, among others, to the following rule in~$\Pi'_i$:
$$
\begin{array}{r@{\;}c@{\;}l}
\!\!\!P_1(y_3) \vee P_2(y_1) &\leftarrow& r(x_1,y_1,y_2) \wedge s(y_3,z) \wedge R_{a_1}(x_1) \, \wedge \\[\myeqnsep]
&& 
\!\!\!\!\!
R_{a_2}(y_1) \wedge R_{a_2}(y_2) \wedge R_{a_2}(y_3).
\end{array}
$$
The above rule treats the case where the variable $x$ from the
original rule is mapped to the constant $a_1$, $y$ to $a_2$, and $z$
not to any constant in \Cbf. Note that the original variables $x$ and
$y$ have been duplicated because $R_{a_1}$ and $R_{a_2}$ need not be
singletons while $a_1$ and $a_2$ denote a single object. So
intuitively, $\Pi'_i$ treats its input instance $I$ as if it was the
quotient $I'$ of $I$ obtained by identifying all $a_1$, $a_2$ with
$R_b(a_1),R_b(a_2) \in I$ for some $b \in \Cbf$. In addition to the
above rules, $\Pi'_2$ also contains $\mn{goal}() \leftarrow R_{a_1}(x)
\wedge R_{a_2}(x)$ for all distinct $a_1,a_2 \in \Cbf$.
\begin{restatable}{lemma}{LEMconstredtwook}\label{lem:constredtwook}
  $\Pi_1 \subseteq \Pi_2$ iff $\Pi'_1 \subseteq \Pi'_2$.
\end{restatable}
It can be verified that $|\Pi'_i| \leq 2^{|\Pi'_i|^2}$ and that the
rule size of $\Pi'_i$ is bounded by twice the rule size of
$\Pi_i$. Because of the latter, the simplification of the programs
$\Pi'_i$ according to Theorem~\ref{thm:simplify} yields programs
whose size is still bounded by $2^{p(|\Pi_i|)}$, as in the proof of
Theorem~\ref{thm:mainupper}, and whose variable width is bounded by
twice the variable width of $\Pi_i$. It is thus easy to check that we
obtain the same overall bounds as stated in
Theorem~\ref{thm:mainupper}.
\begin{theorem}
\label{thm:mainupperPLUS} 
Containment between MDDLog programs of any arity and with constants 
is
in {\sc 2NExpTime}.  More precisely, for 
programs $\Pi_1$ and $\Pi_2$, it can be decided non-deterministically 
in time $2^{2^{p(|\Pi_2| \cdot \mn{log}|\Pi_1|)}}$ whether $\Pi_1 
\subseteq \Pi_2$, $p$ a polynomial. 
\end{theorem}

\section{Ontology-Mediated Queries}

We now consider containment between ontology-mediated queries based on
description logics, which we introduce next.

An \emph{ontology-mediated query (OMQ)} over a schema $\Sbf_E$ is
a triple $(\Tmc,\Sbf_E,q)$, where \Tmc is a TBox formulated in a
description logic and $q$ is a query over the schema $\Sbf_E \cup
\mn{sig}(\Tmc)$, with $\mn{sig}(\Tmc)$ the set of relation symbols
used in \Tmc. 
The TBox can introduce symbols that are not in $\Sbf_E$, which allows
it to enrich the schema of the query~$q$. As the TBox language, we use
the description logic \ALC, its extension $\mathcal{ALCI}$ with
inverse roles, and the further extension $\mathcal{SHI}$ of
$\mathcal{ALCI}$ with transitive roles and role hierarchies.  Since
all these logics admit only unary and binary relations, we assume that
these are the only allowed arities in schemas throughout the
section. As the actual query language, we use UCQs and CQs. The OMQ
languages that these choices give rise to are denoted with
$(\ALC,\text{UCQ})$, $(\ALCI,\text{UCQ})$,
$(\mathcal{SHI},\text{UCQ})$, and so on. In OMQs $(\Tmc,\Sbf_E,q)$
from $(\mathcal{SHI},\text{UCQ})$, we disallow superroles of
transitive roles in $q$; it is known that allowing transitive roles in
the query poses serious additional complications, which are outside
the scope of this paper, see e.g.\
\cite{DBLP:conf/dlog/BienvenuELOS10,DBLP:conf/icalp/GottlobPT13}.  The
semantics of an OMQ is given in terms of \emph{certain answers}.  We
refer to the appendix for further details and only give an example of
an OMQ from $(\ALC,\text{UCQ})$.
\begin{example}
Let the OMQ $Q=(\Tmc,\Sbf_E,q)$ be given by
$$
\begin{array}{r@{\;}c@{\;}l}
  \Tmc &=& \{ \; \exists \mn{manages} . \mn{Project} 
           \sqsubseteq 
           \mn{Manager}, \\[\myeqnsep]
  && \ \; \,  \mn{Employee} \sqsubseteq \mn{Male} \sqcup \mn{Female}
     \; \} \\[\myeqnsep]
  \Sbf_E &=& \{ \mn{Employee},\mn{Project},\mn{Male},\mn{Female},\mn{manages} \} \\[\myeqnsep]
  q(x) &=& \mn{Manager}(x) \wedge \mn{Female}(x)
\end{array}
$$
On the $\Sbf_E$-instance
$$
\begin{array}{l}
  \mn{manages}(e_1,e_2), \mn{Female}(e_1), \mn{Project}(e_2), \\[\myeqnsep]
  \mn{manages}(e'_1,e_2), \mn{Employee}(e'_1),
\end{array}
$$
%
 the only certain answer to $Q$ is $e_1$.
\end{example}
Let $Q_i=(\Tmc_i,\Sbf_E,q_i)$, $i \in \{1,2\}$. Then $Q_1$ \emph{is
  contained in} $Q_2$, written $Q_1 \subseteq Q_2$, if for every
$\Sbf_E$-instance~$I$, the certain answers to $Q_1$ on $I$ are a
subset of the certain answers to $Q_2$ on $I$. The query containment
problems between OMQs considered in \cite{DBLP:conf/kr/BienvenuLW12}
are closely related to ours, but concern different (weaker) OMQ
languages. One difference in setup is that, there, the definition of
``contained in'' does not refer to all $\Sbf_E$-instances $I$, but
only to those that are consistent with both $\Tmc_1$ and $\Tmc_2$. Our
results apply to both notions of consistency. In fact, we show in the
appendix that consistent containment between OMQs can be reduced in
polynomial time to unrestricted containment as studied in this paper,
and in our lower bound we use TBoxes that are consistent w.r.t.\ all
instances. We use $|\Tmc|$ and $|q|$ to denote the size of a TBox
\Tmc and a query $q$, defined as for MDDLog programs.

The following is the main result on OMQs established in this paper.
It solves an open problem from \cite{DBLP:conf/kr/BienvenuLW12}.
\begin{theorem}
\label{thm:alcicompl}
  The following containment problems are 2\NExpTime-complete:
  \begin{enumerate}

  \item of an $(\mathcal{ALC},\text{UCQ})$-OMQ in a CQ;

  \item of an $(\mathcal{ALC},\text{UCQ})$-OMQ in an
    $(\mathcal{ALC},\text{UCQ})$-OMQ;

  \item of an $(\mathcal{ALCI},\text{CQ})$-OMQ in an $(\mathcal{ALCI},\text{CQ})$-OMQ;

  \item of a $(\mathcal{SHI},\text{UCQ})$-OMQ in a
    $(\mathcal{SHI},\text{UCQ})$-OMQ.

  \end{enumerate}
The lower bounds apply already when the TBoxes of the two OMQs are 
identical. 
\end{theorem}
We start with the lower bounds.  For Point~1 and~2 of
Theorem~\ref{thm:alcicompl}, we make use of the lower bound that we
have already obtained for MDDLog. It was observed in
\cite{DBLP:conf/pods/BienvenuCLW13,DBLP:journals/tods/BienvenuCLW14}
that both $(\ALC,\text{UCQ})$ and $(\mathcal{SHI},\text{UCQ})$ have
the same expressive power as MDDLog restricted to unary and binary EDB
relations. In fact, every such MDDLog program can be translated into
an equivalent OMQ from $(\ALC,\text{UCQ})$ in polynomial time.  Thus,
the lower bounds in Point~1 and~2 of Theorem~\ref{thm:alcicompl} are a
consequence of those in Theorem~\ref{thm:hardness1}.  

The lower bound stated in Point~3 of Theorem~\ref{thm:alcicompl} is
proved by a non-trivial reduction of the 2-exp torus tiling problem.
Compared to the reduction that we have used for MDDLog, some major
changes are required. In particular, the queries used there do not
seem to be suitable for this case, and thus we replace them by a
different set of queries originally introduced in
\cite{DBLP:conf/cade/Lutz08}. Details are in the appendix. It can be
shown exactly as in the proof of Theorem~3 in
\cite{DBLP:conf/kr/BienvenuLW12} that, in the lower bounds in Points~2
and~3, we can assume the TBoxes of the two OMQs to be identical.

We note in passing that we again obtain corresponding lower bounds
for rewritability.
%
\begin{restatable}{theorem}{THMrewriteOMQ}\label{thm:rewriteOMQ}
In $(\mathcal{ALC},\text{UCQ})$ and $(\mathcal{ALCI},\text{CQ})$,
rewritability into FO, into monadic Datalog, and into Datalog is 2{\sc 
  NExpTime}-hard. 
\end{restatable}
Now for the upper bounds in Theorem~\ref{thm:alcicompl}.  The
translation of OMQs into MDDLog programs is more involved than the
converse direction, a naive attempt resulting in an MDDLog program
with existential quantifiers in the rule heads. We next analyze
the blowups involved. The construction used in the proof of the
following theorem is a refinement of a construction from
\cite{DBLP:conf/pods/BienvenuCLW13}, resulting in improved bounds.
\begin{restatable}{theorem}{THMthatreductionagain}\label{thm:thatreductionagain} 
For every OMQ $Q=(\Tmc,\Sbf_E,q)$ from $(\mathcal{SHI},\text{UCQ})$, one can
construct an equivalent MDDLog program $\Pi$ such that
  \begin{enumerate}

  \item $|\Pi| \leq 2^{2^{p(|q| \cdot \mn{log}|\Tmc|)}}$;

  \item the IDB schema of $\Pi$ is of size $2^{p(|q| \cdot \mn{log}|\Tmc|)}$;

  \item the rule size of $\Pi$ is bounded by $|q|$

  \end{enumerate}
  where $p$ is a polynomial. The construction takes time polynomial
in $|\Pi|$.
\end{restatable}
%
%
%
%
We now use Theorem~\ref{thm:thatreductionagain} to derive an upper
complexity bound for containment in
$(\mathcal{SHI},\text{UCQ})$. While there are double exponential
blowups both in Theorem~\ref{thm:mainupperPLUS} and in
Theorem~\ref{thm:thatreductionagain}, a careful analysis reveals that
they do not add up and, overall, still give rise to a 2\NExpTime upper
bound. In contrast to Theorem~\ref{thm:mainupper}, though, we only get
an algorithm whose running time is double exponential in both inputs
$(\Tmc_1,\Sbf_E,q_1)$ and $(\Tmc_2,\Sbf_E,q_2)$. However, it is double
exponential only in the size of the actual queries $q_1$ and $q_2$
while being only single exponential in the size of the TBoxes $\Tmc_1$
and~$\Tmc_2$. This is good news since the size of $q_1$ and $q_2$ is
typically very small compared to the sizes of $\Tmc_1$
and~$\Tmc_2$. For this reason, it can even be reasonable to assume
that the sizes of $q_1$ and $q_2$ are constant, in the same way in
which the size of the query is assumed to be constant in classical
data complexity. Note that, under this assumption, we obtain a {\sc
  NExpTime} upper bound for containment.
\begin{restatable}{theorem}{THMALCI}\label{thm:ALCI}
Containment between OMQs from $(\mathcal{SHI},\text{UCQ})$ is
in
{\sc 2NExpTime}.  More precisely, for OMQs
$Q_1=(\Tmc_1,\Sbf_E,q_1)$ and $Q_2=(\Tmc_2,\Sbf_E,q_2)$, it can be decided
non-deterministically in time 
$2^{2^{p(|q_1|\cdot|q_2|\cdot\mn{log}|\Tmc_1|\cdot\mn{log}|\Tmc_2|)}}$
whether $Q_1 \subseteq Q_2$, $p$ a
polynomial.
\end{restatable}

\section{Outlook}

There are several interesting questions left open. One is whether
decidability of containment in MMSNP generalizes to GMSNP, where IDB
relations can have any arity and rules must be frontier-guarded
\cite{DBLP:journals/tods/BienvenuCLW14} or even to frontier-guarded
disjunctive TGDs, which are the extension of GMSNP with
existential quantification in the rule head
\cite{DBLP:conf/ijcai/BourhisMP13}. We remark that an extension of
Theorem~\ref{thm:thatreductionagain} to frontier-\emph{one}
disjunctive TGDs (where rule body and head share only a single
variable) 
seems not too hard.

Other open problems concern containment between OMQs. In
particular, it would be good to know the complexity of containment in
$(\ALC,\text{CQ})$ which must lie between \NExpTime and 2\NExpTime.
Note that our first lower bound crucially relies on \emph{unions} of
conjunctive queries to be evailable, and the second one on inverse
roles. It is known that adding inverse roles to \ALC tends to increase
the complexity of querying-related problems
\cite{DBLP:conf/cade/Lutz08}, so the complexity of containment in
$(\ALC,\text{CQ})$ might indeed be lower than 2\NExpTime.  It would
also be interesting to study containment for OMQs from
$(\ALCI,\text{UCQ})$ where the actual query is connected and has at
least one answer variable. In the case of query answering, such a
(practically very relevant) assumption causes the complexity to drop
\cite{DBLP:conf/cade/Lutz08}. Is this also the case for containment?

\smallskip
\noindent {\bf Acknowledgements.} Bourhis was supported by
CPER Data and Lutz was supported by  ERC grant 647289 CODA. 



\cleardoublepage

\appendix

\section{MDDLog Hardness: Missing Details}
\label{app:mddloglower}

\begin{figure}[t!]
  \begin{center}
    \framebox[1\columnwidth]{
 \footnotesize
\begin{tikzpicture}[->,font=\scriptsize]
 \footnotesize
\node (a) at (0,0.75) {$a$};
\node(b) at (-0.5,0) {$b$};
\node(c) at (0.5,0) {$c$};
\node (b1) at (-1.25,-0.75) {$b_1$};
\node (b2) at  (-1.25,-1.5) {$b_2$};
\node (c1) at (1.25,-0.75) {$c_1$};
\draw (a) -- (b) node[midway,above,left] {$r$}; 
\draw (b1) -- (b2)  node[midway,above,left]  {$r$};  
\draw (a) -- (c) node[midway,above,right] {$r$}; 
\draw (b)-- (b1) node[midway,above,left] {$r$}; 
\draw (c) -- (c1) node[midway,above,right] {$r$}; 
\draw (b) .. controls +(up:1cm) and +(left:1cm) .. (b2) node[midway,above,left] {$jump$}; 
\draw (c) .. controls +(up:1cm) and +(right:1cm) .. (c1) node[midway,above,right] {$jump$};
\fill[black!10!white](0.5,-0.1) -- (0.25,-0.75) -- (0.75,-0.75)-- cycle;
\fill[black!10!white](-0.5,-0.1) -- (-0.25,-0.75) -- (-0.75,-0.75)-- cycle;
\end{tikzpicture}
}
    \caption{The counting gadgets.}
    \label{fig:gadgets}
  \end{center}
\end{figure}
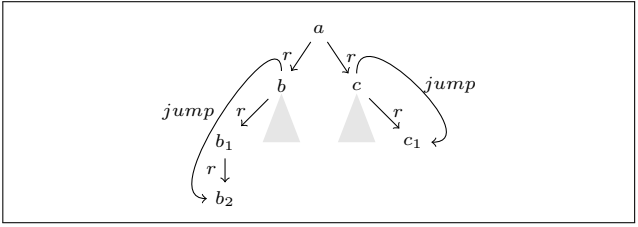
We first repeat the details of the construction of a UCQ which
achieves (Q1) along with additional information, then prove
Lemma~\ref{lem:contlowercorr}, and subsequently describe how the UCQ
can be replaced by a CQ. Finally, we prove
Theorem~\ref{thm:rewrite}. Since the constructed queries will make use
of the counting gadgets in counting trees, we show these gadgets again
in Figure~\ref{fig:gadgets}. There, $a$ is a node in a counting tree,
$b$ is its left successor, and $c$ is its right successor.

To define the UCQ that achieves (Q1), set $q_0(x,y)=r(x_0,y_0)$ and,
for $0 < i \leq m$,
$$
\begin{array}{r@{}l}
q_i(x_i,&y_i) = \exists x_{i-1} \exists y_{i-1}
\exists z_{i,0} \cdots \exists z_{i,i+2}
\exists z'_{i,1} \cdots \exists z'_{i,i+3} \\[\myeqnsep]
  & q_{i-1}(x_{i-1},y_{i-1}) \wedge r(x_{i-1},x_i) \wedge r(y_{i-1},y_i)  \, \wedge \\[\myeqnsep]
  & \mn{jump}(x_i,z_{i,i+2}) \wedge \mn{jump}(y_i,z'_{i,i+3})   \, \wedge \\[\myeqnsep]
  & r(z_{i,0},z_{i,1}) \wedge \cdots \wedge r(z_{i,i+1},z_{i,i+2})   \, \wedge \\[\myeqnsep]
  & r(z_{i,0},z'_{i,1}) \wedge r(z_{i,1},z'_{i,2}) \wedge \cdots \wedge r(z'_{i,i+2},z'_{i,i+3})
\end{array}
$$
The idea is that $I \models q_m[a,b]$ if $a$ and $b$ are leafs in
successive counting trees that are at the same leaf position, that is,
(i)~the roots of the trees are connected by the relation $r$ and
(ii)~$a$ can be reached from the root of the first tree by following
the same sequence of left and right successors that one also needs to
follow to reach $b$ from the root of the second tree. In fact, the
$r$-atom in $q_0$ corresponds to the move from the root of one
counting tree to the root of a successive tree, the atoms
$r(x_{i-1},x_i)$ and $r(y_{i-1},y_i)$ in $q_i$ correspond to moving
down the $i$-th step in both trees, and the remaining atoms in $q_i$
make sure that both of these steps are to a left successor or to a
right successor.

To understand the latter, note that \mn{jump} is the relation used in
the navigation gadgets attached to tree nodes. The variable $z_{i,i+2}$
can only be mapped to the target of the \mn{jump} relation in the
navigation gadget at~$x_i$, and likewise for $z'_{i,i+3}$ and the
target of the \mn{jump} relation in the navigation gadget at $y_i$.
Note that there must be a $z_{i,0}$ from which $z_{i,i+2}$ can be
reached along an $r$-path of length $i+2$ and from which $z'_{i,i+3}$
can be reached along an $r$-path of length $i+3$.  If $x_i$ and $y_i$
are both left successors, then this $z_{i,0}$ is the root of the first
counting tree. If $x_i$ and $y_i$ are both right successors, then
$z_{i,0}$ is the $r$-predecessor of the root of the first counting
tree, which must exist at all relevant nodes: at grid nodes because of
the self step nodes and at horizontal/vertical step nodes because they
have a grid node as $r$-predecessor.  If $x_i$ is a left successor and
$y_i$ a right successor or vice versa, then there is no target for
$z_{i,0}$ because this target would have to reach the root of the
first counting tree on a path of length one and on a path of length
zero, but there is no reflexive loop at the root of counting
trees (only a length two loop via self step nodes).

It is now easy to define the desired UCQ:
$$
\begin{array}{r@{\,}c@{\,}l}
q&=&\exists x_m \exists y_m \, q_m(x_m,y_m) \wedge B_1(x_m) \wedge
\overline{B}_2(y_m) \\[\myeqnsep]
&& \vee \, \exists x_m \exists y_m \, q_m(x_m,y_m) \wedge \overline{B}_1(x_m) \wedge
B_2(y_m)
\end{array}
$$
The first CQ in this UCQ is displayed in Figure~\ref{fig:mainquery}.  
\begin{figure}[t!]
  \begin{center}
    \framebox[1\columnwidth]{
\begin{tikzpicture}[->,font=\scriptsize]
\node (x0) at (0,0) {$x_0$};
\node (x1) at (-2.4,-1) {$x_1$} ;
\node 	(x2) at (-2.4,-2) {$x_2$} ;
\node 	(xm1) at (-2.4,-3) {$x_{m-1}$} ;
\node (xm) at (-2.4,-4.5) {$x_m$} ;

\node(xb) at (-2.8, -4.6) {$B_1$};
\node(xb) at (2.8, -4.6) {$\overline{B}_2$};

\node (u1) at (-1,-1.5) {$z_{1,3}$};
\node (u2) at (-1,-2.5) {$z_{2,4}$} ;
\node 	(um1) at (-1,-3.5) {$z_{m-1,m+1}$} ;
\node 	(um) at (-1,-5) {$z_{m,m+2}$} ;

\node (y0) at (1,-0.3) {$y_0$};
\node (y1) at (2.4,-1) {$y_1$} ;
\node 	(y2) at (2.4,-2) {$y_2$} ;
\node 	(ym1) at (2.4,-3) {$y_{m-1}$} ;
\node 	(ym) at (2.4,-4.5) {$y_m$} ;

\node (v1) at (1,-1.5){$z_{1,4}$};
\node (v2) at (1,-2.5) {$z_{2,5}$} ;
\node 	(vm1) at (1,-3.5) {$z_{m-1,m+2}$} ;
\node 	(vm) at (1,-5) {$z_{m,m+3}$}  ;

\node (z1) at (0,-0.5) {$z_{1,0}$};
\node (z2) at (0,-1.5) {$z_{2,0}$} ;
\node 	(zm1) at (0,-2.5) {$z_{m-1,0}$} ;
\node 	(zm) at (0,-4) {$z_{m,0}$} ;

\draw (x0) to[bend right]   node[midway,left]  {$r$} (x1); 
\draw (x1) -- (x2)  node[midway,left]  {$r$}; 
\draw [dashed] (x2) -- (xm1); 
\draw (xm1) -- (xm)  node[midway,left]  {$r$}; 

\draw (x0) to  node[midway,right,above]  {$r$} (y0) ; 
\draw (y0) to[bend left]  node[midway,right]  {$r$} (y1) ; 
\draw (y1) -- (y2)  node[midway,right]  {$r$}; 
\draw [dashed] (y2) -- (ym1); 
\draw (ym1) -- (ym)  node[midway,right]  {$r$}; 

\draw (x1) -- (u1)  node[sloped, midway,above]  {$jump$}; 
\draw (x2) -- (u2)  node[sloped, midway,above]  {$jump$};  
\draw (xm1) -- (um1)  node[sloped, midway,above]  {$jump$}; 
\draw (xm) -- (um)  node[sloped, midway,above]  {$jump$}; 

\draw (y1) -- (v1)  node[sloped, midway,above]  {$jump$}; 
\draw (y2) -- (v2)  node[sloped, midway,above]  {$jump$};  
\draw (ym1) -- (vm1)  node[sloped, midway,above]  {$jump$}; 
\draw (ym) -- (vm)  node[sloped, midway,above]  {$jump$};

\draw (z1) -- (u1)  node[sloped,midway,above]  {$r^3$}; 
\draw (z2) -- (u2)  node[sloped,midway,above]  {$r^4$}; 
\draw (zm1) -- (um1) node[sloped,midway,above]  {$r^{m+1}$}	; 
\draw (zm) -- (um)  node[sloped,midway,above]  {$r^{m+2}$}; 

\draw (z1) -- (v1)  node[sloped,midway,above]  {$r^4$}; 
\draw (z2) -- (v2)  node[sloped,midway,above]  {$r^5$}; 
\draw (zm1) -- (vm1) node[sloped,midway,above]  {$r^{m+2}$}	; 
\draw (zm) -- (vm)  node[sloped,midway,above]  {$r^{m+3}$}; 
\end{tikzpicture}
}
    \caption{The first CQ in $q$ (again).}
    \label{fig:mainquery}
  \end{center}
\end{figure}
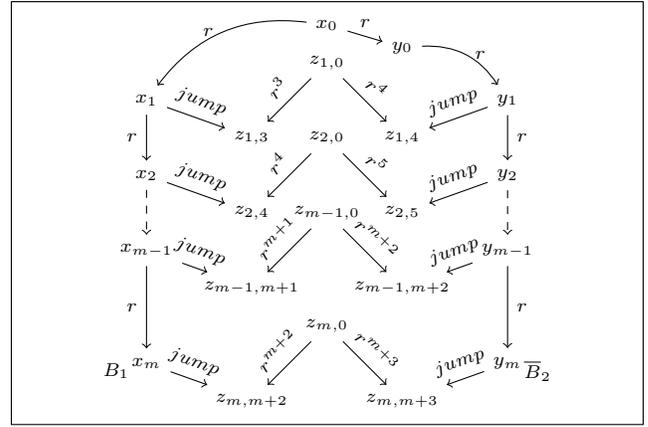
As required, it evaluates to true on an instance if there are
successive counting trees (whose roots have an $r$-predecessor and)
which contain two leafs at the same position that are labeled
differently regarding $B_1,\overline{B}_1$ and $B_2,\overline{B}_2$.
This finishes the construction for the case of UCQs. We first
establish correctness and then show how to replace the UCQ with a CQ.

\noindent 
\LEMcontlowercorr*

\noindent
\begin{proof}
  (sketch)
  Assume first that there is no tiling for $P$ and $w_0$. Let $I$ be
  the instance that represents the $2^{2^n} \times 2^{2^n}$-grid with counting trees
  in the way described above.  It can be verified that
  $I \not\models q$.  We aim to show that $I \models \Pi$ and thus $I$
  witnesses $\Pi \not\subseteq q$. Assume to the contrary that
  $I \not\models \Pi$.  Then there is an extension $J$ of $I$ that
  satisfies all rules in $\Pi$, but does not contain $\mn{goal}()$.
  In particular, $J$ must contain at least one atom $T_i(c)$ for each
  constant $c$ with $\mn{gactive}(c) \in J$, thus we can choose a
  concrete $T_i(c)$ for each such $c$.  Since none of the goal rules
  in $\Pi$ applies, these chosen atoms must represent a tiling for $P$
  and $w_0$. We have thus obtained a contradiction to the assumption
  that no such tiling exists.

\smallskip

Now assume that there is a tiling $f$ for $P$ and $w_0$. Take an
instance $I$ with $I \models \Pi$. Assume to the contrary of what is
to be shown that $I \not \models q$. Then $I$ satisfies
Conditions~(Q1) and~(Q2).  Extend $I$ to a new instance $J$ as
follows. Since $I$ satisfies (Q2), every g-active constant $c$ in $I$
is associated with a unique counter value, thus with a unique
horizontal position $x \in \{0,\dots,2^{2^n}-1\}$ and a unique
vertical position $y \in \{0,\dots,2^{2^n}-1\}$. Include $T_i(c) \in
J$ if $f(x,y)=T_i$ and then exhaustively apply all non-disjunctive
rules from $\Pi_{\mn{tree}}$. One can verify that $J$ satisfies all
rules in $\Pi$ while making the goal relation false, in contradiction
to $I \models \Pi$. In particular, satisfaction of (Q1) and the way
in which we have added facts $T_i(c)$ to $J$ imply that none of
the goal rules that check for a tiling defect applies.
\end{proof}
To replace the UCQ $q$ by a CQ, we again use a coding trick from
\cite{DBLP:conf/mfcs/BjorklundMS08}. 
The basic
idea is to replace $B_1,\overline{B}_1$ and $B_2,\overline{B}_2$ with
suitable \emph{bit gadgets} and then to use a construction that is
very similar to the one used above for ensuring that we consistenly
follow left successors or right successors in corresponding steps of
the navigation in the two involved trees.

We replace $B_1(x)$ with the following \emph{bit one gadget}:
$$
\begin{array}{c}
  r(x,x_1) \wedge r(x_1,x_2) \wedge r(x_2,x_3) \wedge r(x_3,x_4)
   \, \wedge \\[\myeqnsep]
  \mn{jump_1}(x,x_1) \wedge \mn{jump_1}(x,x_4)
\end{array}
$$
where $\mn{jump}_1$ is a fresh EDB relation and $\overline{B}_1(x)$
with the following \emph{bit zero gadget}:
$$
\begin{array}{c}
  r(x,x_1) \wedge r(x_1,x_2) \wedge r(x_2,x_3) \wedge r(x_3,x_4)
   \, \wedge \\[\myeqnsep]
  \mn{jump_1}(x,x_2) \wedge \mn{jump_1}(x,x_3).
\end{array}
$$
$B_2$ and $\overline{B}_2$ are replaced with corresponding gadgets in
which only $\mn{jump}_1$ is replaced with $\mn{jump}_2$. The existence
of these bit gadgets needs to be verified in the rules of
$\Pi_\mn{tree}$ that ensure the existence of counting trees.  In
addition to that, we require one further modification to
$\Pi_\mn{tree}$: the self step loops of length two at each grid node
are replaced with self step nodes of length four. All three
intermediate nodes on these loops behave exactly like a self step node
before. We then replace the above UCQ by
$$
\begin{array}{r@{\,}c@{\,}l}
q&=&\exists x_m \exists y_m \exists z_0 \cdots \exists z_{m+2}
\exists z'_1 \cdots \exists z_{m+5} \\[\myeqnsep]
&& q_m(x_m,y_m) \, \wedge \\[\myeqnsep]
&& \mn{jump}_1(x_m,z_{m+2}) \wedge \mn{jump}_2(y_m,z'_{m+5})   \, \wedge \\[\myeqnsep]
&& r(z_{0},z_{1}) \wedge \cdots \wedge r(z_{m+1},z_{m+2})   \, \wedge \\[\myeqnsep]
 && r(z_{0},z'_{1}) \wedge r(z_{1},z'_{2}) \wedge \cdots \wedge
    r(z'_{m+4},z'_{m+5}) 
\end{array}
$$
Note that the $x_m$ and $y_m$ must be leafs in successive counting
trees with the same leaf position. Additionally, the variable
$z_{m+2}$ can only be mapped to a target of the $\mn{jump}_1$ relation
in the bit gadget at~$x_m$, and likewise for $z'_{m+5}$ and a target
of the $\mn{jump}_2$ relation in the bit gadget at $y_m$.  There must
also be a $z_{0}$ from which $z_{m+2}$ can be reached along an
$r$-path of length $m+2$ and from which $z'_{m+5}$ can be reached
along an $r$-path of length $m+5$ (thus the difference in lengths is
three).  If the bit value at $x_m$ is zero and the bit value at $y_m$
is one, then this $z_{0}$ is the root of the first counting tree. If
the bit value at $x_m$ is one and the bit value at $y_m$ is zero, the
we can use for $z_{0}$ an $r$-predecessor of the root of the first
counting tree. It can be verified that when the bit values at $x_m$
and $y_m$ are identical, then the possible target for $z_0$ must have
$r$-paths to the root of the first counting tree of length $i$ and $j$
steps, for some
$$
(i,j) \in \{ (0,2), (0,3), (1,0), (1,3) \}.
$$
However, since we have extended the length of self loops at grid nodes
from two to four, there is no such target.  This finishes the
construction of the CQ $q$ and establishes the lower bounds stated in
Theorem~\ref{thm:hardness1}.

\medskip
We now come to the proof of Theorem~\ref{thm:rewrite}.

\noindent
\THMrewrite*

\noindent
\begin{proof}
  It suffices to consider Boolean MDDLog programs.  First note that,
  by Rossman's theorem, any such program that is rewritable into FO is
  rewritable into a UCQ. Consequently, FO-rewritability implies
  monadic Datalog-rewritability implies Datalog-rewritability. Based
  on this observation, we deal with all three kinds of rewritability
  in a single proof: we show that from a 2-exp \torus tiling problem
  $P$ and an input $w_0$ to $P$, we can construct in polynomial time
  a Boolean MDDLog program $\Pi'$ such that 
  \begin{enumerate}

  \item if there is a tiling for $P$ and $w_0$, then $\Pi'$ is
    FO-rewritable;

  \item if there is no tiling for $P$ and $w_0$, then $\Pi'$ is
    not Datalog-rewritable.

  \end{enumerate}
  Reconsider the reduction of the 2-exp square tiling problem to
  MDDLog containment given above. Given a 2-exp \torus tiling problem
  $P$ and an input $w_0$ to $P$, we have shown how to construct a
  Boolean MDDLog program~$\Pi$ and a Boolean CQ $q$ such that
  $\Pi \subseteq q$ iff there is a tiling for $P$ and $w_0$. Let
  $\Sbf_E$ be the EDB schema of $\Pi$ and $q$. To obtain the desired
  program $\Pi'$, we modify $\Pi$ as follows:
  \begin{enumerate}

\item in every goal rule,  change the head to $A(x)$;

\item add $\mn{goal}() \leftarrow q$,

\item add $R(x) \vee G(x) \vee B(x) \leftarrow A(x)$,
  $
    \mn{goal}() \leftarrow C_1(x) \wedge C_2(x) 
  $ for all  distinct  $C_1,C_2 \in \{R,G,B\}$,
  and $\mn{goal}() \leftarrow C(x) \wedge s(x,y) \wedge C(y)$
  for all $C \in \{R,G,B\}$
\end{enumerate}
where $s$ is a fresh EDB relation and $A,R,G,B$ are IDB
relations. Let $\Sbf'_E = \Sbf_E \cup \{s\}$. We now show that
$\Pi'$ satisfies Points~1 and~2 above.

For Point~1, assume that there is a tiling for $P$ and~$w_0$.  Then
$\Pi \subseteq q$. We claim that $q$ is a rewriting of $\Pi'$. By
construction of $\Pi'$, $I \models q$ clearly implies $I
\models \Pi'$ for all $\Sbf'_E$-instances $I$. For the converse,
let $I \models \Pi'$. First assume $I \not\models \Pi$. Then
there is an extension $J$ of $I$ to the IDB relations in $\Pi'$
such that the extension of $A$ is empty. Consequently,
we must have $I \models \Pi'$ because $I \models q$ and we
are done. Now assume $I \models \Pi$. Since $\Pi \subseteq q$,
this implies $I \models q$ as desired.

For Point~2, assume there is no tiling for $P$ and $w_0$. Then
$\Pi \not\subseteq q$. Given an undirected graph $G=(V,E)$, let the
instance $I^+_G$ be defined as the disjoint union of the instance $I_0$
which represents the $2^{2^n}$-grid plus counting gadgets and the
instance $I_G$ that contains the fact $s(v_1,v_2)$ for every
$\{v_1,v_2\} \in V$.

Since there is no tiling for $P$ and $w_0$, we have $I_0 \models \Pi$
and thus $I^+_G \models \Pi$. By construction of $\Pi'$ and since
$\Pi \not\models q$, this implies that $I_G \models \Pi'$ iff $G$ is
not 3-colorable. Assume to the contrary of what is to be shown that
there is a Datalog-rewriting $\Gamma$ of $\Pi'$. It is not difficult
to modify $\Gamma$ so that its EDB schema is $\Sbf''_E = \{ s \}$ and
on any $\Sbf''_E$-instance $I_G$ representing an undirected graph $G$,
the modified program $\Gamma'$ yields the same result that $\Gamma$
yields on $I^+_G$. We only sketch the idea: for every $n$-ary IDB
relation $S$ of $\Gamma$, all sets of positions
$P=\{i_1,\dots,i_k \}\subseteq \{ 1,\dots,n\}$, and all tuples
$t=(d_1,\dots,d_k)$ of elements of $I_0$, introduce a fresh $n-k$-ary
IDB $S_{P,t}$. Intuitively, $S_{P,t}$ is used to represent facts
$S(d_1,\dots,d_k)$ where every position $i_j \in P$ is mapped to the
element $d_j$ (which is not part of the input instance) and every 
position that is not in $P$ is mapped to an element of the input
instance $I_G$. It remains to introduce additional versions of each
rule that use the new predicates, in all possible combinations.

We have thus shown that non-3-colorability of graphs can be expressed
in Datalog, which is not the case
\cite{DBLP:journals/jcss/AfratiCY95}.
\end{proof}

\section{MDDLog Upper Bound: Missing Details}

\subsubsection*{From Unrestricted to Simple Programs}

\noindent
\THMsimplify*

\medskip
\noindent
To prove Theorem~\ref{thm:simplify}, we first concentrate on a single
Boolean MDDLog program $\Pi$ 
over EDB schema $\Sbf_E$. We first construct from $\Pi$ an equivalent
MDDLog program $\Pi'$ such that the following conditions are satisfied:
\begin{enumerate}

\item[(i)] all rule bodies are biconnected,
that is, when any single variable is removed from the body (by
deleting all atoms that contain it), then the resulting rule body is
still connected; 

\item[(ii)] if $R(x,\dots,x)$ occurs in a rule body with $R$ EDB, then
  the body contains no other EDB atoms.

\end{enumerate}
To construct $\Pi'$, we first extend $\Pi$ with all rules that can be
obtained from a rule in $\Pi$ by consistently identifying
variables and then exhaustively apply the following
rules:
\begin{itemize}

\item replace every rule
  $p(\vect{y}) \leftarrow q_1(\vect{x}_1) \wedge q_2(\vect{x}_2)$
  where $\vect{x}_1$ and $\vect{x}_2$ share exactly one variable $x$
  but both contain also other variables with the rules
$p_1(\vect{y}_1) \vee Q(x) \leftarrow q_1(\vect{x}_1)$ and
  $p_2(\vect{y}_2) \leftarrow Q(x) \wedge q_2(\vect{x}_2)$, where $Q$
  is a fresh monadic IDB relation and $p_i(\vect{y}_i)$ is the
  restriction of $p(\vect{y})$ to atoms that are nullary or contain a
  variable from $q_i$, $i \in \{1,2\}$;

\item replace every rule
  $p(\vect{y}) \leftarrow q_1(\vect{x}_1) \wedge q_2(\vect{x}_2)$
  where $\vect{x}_1$ and $\vect{x}_2$ share no variables and are both
  non-empty with the rules $p_1(\vect{y}_1) \vee Q() \leftarrow
  q_1(\vect{x}_1)$ and $p_2(\vect{y}_2) \leftarrow
  Q() \wedge q_2(\vect{x}_2)$, where
  $Q()$ is a fresh nullary IDB relation and the $p_i(\vect{y}_i)$ are
  as above;

\item replace every rule
  $p(\vect{y}) \leftarrow R(x,\dots,x) \wedge q(\vect{x})$ where $R$
  is an EDB relation and $q$ contains at least one EDB atom and the
  variable $x$, with the rules $Q(x) \leftarrow R(x,\dots,x)$ and
  $p(\vect{y}) \leftarrow Q(x) \wedge q(\vect{x})$, where $Q$ is a
  fresh monadic IDB relation.
  
\end{itemize}
It is easy to see that the program $\Pi'$ is equivalent to the
original program $\Pi$.  We next construct from $\Pi'$ the desired
simplification $\Pi^S$ of $\Pi$. The intuition is that in every rule
of $\Pi'$, we replace all EDB atoms in the rule body by a single EDB
atom that uses a fresh EDB relation which represents the conjunction
of all atoms replaced.  We also need to take care of implications
between the new EDB
relations. 
In the 
following, we make the construction precise. 

For every conjunctive query $q(\vect{x})$ and schema \Sbf, we use
$q(\vect{x})|_\Sbf$ to denote the restriction of $q(\vect{x})$
to $\Sbf$-atoms. The EDB schema $\Sbf'_E$ of $\Pi^S$ consists of
the 
relations $R_{q(\vect{x})|_{\Sbf_E}}$, 
$p(\vect{y}) \leftarrow
q(\vect{x})$ a rule in $\Pi'$; the
arity of $R_{q(\vect{x})|_{\Sbf_E}}$ is the number of variables in
$q(\vect{x})$ (equivalently: in
$q(\vect{x})|_{\Sbf_E}$). 
%
%

Let $\Sbf_I$ be the IDB schema of~$\Pi'$. The program $\Pi^S$ consists
of the following rules:
\begin{itemize}

\item[($*$)] whenever $p(\vect{y}) \leftarrow
q_1(\vect{x}_1)$ is a rule in $\Pi'$, $R_{q_2(\vect{x}_2)}$ an
  EDB relation in $\Sbf'_E$, and $h: \vect{x}_1 \rightarrow
  \vect{x}_2$ an injective homomorphism from
  $q_1(\vect{x}_1)$ to $q_2(\vect{x}_2)$, then $\Pi^S$
  contains the rule $p(\vect{y}) \leftarrow h^{-1}(R_{q_2}(\vect{x}_2)) \wedge
  q_1(\vect{x}_1)|_{\Sbf_I}$


\end{itemize}
%
where $h^{-1}(R_{q_2}(\vect{x}_2))$ denotes the result of replacing in
$R_{q_2}(\vect{x}_2)$ every variable $x$ with $y$ if $x=h(y)$ and
every variable that does not occur in the range of $h$ with a fresh
variable.  The case where $q_1(\vect{x}_1)$ is identical to
$q_2(\vect{x}_2)$ and $h$ is the identity corresponds to adapting
rules in $\Pi'$ to the new EDB signature and the other cases take care
of implications between EDB relations, as announced.

The last step of the conversion just described is the most important
one, and it is the reason for why we can only use instances of a
certain girth in Point~2 of Theorem~\ref{thm:simplify}.  An example
illustarting this step and the issue with girth can be found in the
main part of the paper. We next analyze the size of the constructed 
program $\Pi^S$.
\begin{lemma}
\label{lem:sizelem1}
Let $r$ be the number of rules in $\Pi$ and $s$ the rule size of
$\Pi$. Then there is a polynomial $p$ such that
  \begin{enumerate}

  \item  $|\Pi^S| \leq p(r \cdot 2^s)$; 

   \item the variable width of 
     $\Pi^S$ is bounded by that of $\Pi$;

  \item $|\Sbf'_E| \leq p(r \cdot 2^s)$. 

  \end{enumerate}
\end{lemma}
\begin{proof}
  We start with Point~1. Let $r$ be the number of rules in $\Pi$, $v$
  the variable width, and $w$ the atom width.  Regarding the size of
  $\Pi^S$, note that the first (identification) step replaces each
  rule of $\Pi$ with at most $k!$ rules, where $k$ is the number of
  variables in the original rule. After this step, we thus have at
  most $r \cdot v!$ rules.  The subsequent rewriting in the
  construction of $\Pi'$ splits each rule into at most one rule per
  atom in the original rule. The number of rules in $\Pi'$ is thus
  bounded by $r \cdot v! \cdot w$. The number of rules in $\Pi^S$ is
  clearly at most quadratic in the number of rules in $\Pi'$, thus
  their number is bounded by $(r \cdot v! \cdot w)^2$. None of the
  steps increases the rule size, i.e., the rule size of $\Pi^S$ is
  bounded by the rule size of $\Pi$. This yields the bound stated in
  Point~1.

  Point~2 is very easy to verify by analyzing the construction of
  $\Pi^S$.

  For Point~3, note that the program $\Pi'$ has the same EDB schema as
  $\Pi$. In the construction of $\Pi^S$, the number of EDB relations
  is bounded by the number of rules in $\Pi'$, thus by $r \cdot v!
  \cdot w$.
\end{proof}
%
%
%
%
%
So far, we have concentrated on a single program. To obtain
Theorem~\ref{thm:simplify}, we have to jointly simplify that two
involved programs $\Pi_1$ and $\Pi_2$. This only means that, when
constructing $\Pi^S_i$ from $\Pi'_i$ in the second step of the
normalization procedure, then we use the set of EDB relations
introduced for \emph{both} $\Pi'_1$ and $\Pi'_2$ instead of only those
for $\Pi'_i$. The bounds in Lemma~\ref{lem:sizelem1} then cleary give
rise to those in Theorem~\ref{thm:simplify}. It remains to show that
the joint simplification $\Pi^S_1,\Pi^S_2$ of $\Pi_1,\Pi_2$ behaves as
expected regarding containment. 
\begin{lemma}
\label{lem:correctnessfirst}
~\\[-5mm]
\begin{enumerate}

\item $\Pi_1 \not\subseteq \Pi_2$ implies $\Pi^S_1 \not\subseteq \Pi^S_2$;
 
\item $\Pi^S_1 \not\subseteq_{>w} \Pi^S_2$ implies $\Pi_1
  \not\subseteq 
\Pi_2$.
 
\end{enumerate}
where $w$ is the atom width of $\Pi_1 \cup \Pi_2$.
\end{lemma}
\begin{proof}
  It is obvious that the construction of the program $\Pi'_i$ from
  $\Pi_i$ preserves equivalence. We can therefore assume that the
  programs $\Pi_1$ and $\Pi_2$ in Lemma~\ref{lem:correctnessfirst} are
  in fact the programs $\Pi'_1$ and $\Pi'_2$.  For $i \in \{1,2\}$,
  let $\Sbf_{I,i}$ be the IDB-schema of $\Pi_i$ (and thus also
  of~$\Pi^S_i$), and let $\mn{goal}_i$ be the goal relation of
  $\Pi_i$.

  \smallskip

  For Point~1, let $I$ be an $\Sbf_E$-instance such that $I \models
  \Pi_1$ and $I \not\models \Pi_2$. Let $J$ be the $\Sbf'_E$-instance
  that consists of all facts $R_q(a_1,\dots,a_n)$ such that
$I \models
  q[a_1,\dots,a_n]$. It remains to show that $J \models \Pi^S_1$ and
  $J \not\models \Pi^S_2$. 

  For $J \models \Pi^S_1$, assume to the contrary of what is to be
  shown that there is an extension $J'$ of $J$ to schema
  $\Sbf'_E \cup \Sbf_{I,1}$ that satisfies all rules of $\Pi^S_i$ and
  does not contain the $\mn{goal}_1$ relation. Let $I'$ be the
  corresponding extension of $I$, that is, $I'$ extends $I$ with the
  $\Sbf_{I,1}$-facts from $J'$. It suffices to show that $J'$
  satisfies all rules in $\Pi_1$ to obtain a contradiction against
  $I \models \Pi_1$. Thus, let $p(\vect{y}) \leftarrow q(\vect{x})$ be
  a rule in $\Pi_1$ and let $h$ be a homomorphism from $q(\vect{x})$
  to $I'$. Then $\Pi^S_1$ contains the rule
  $p(\vect{y}) \leftarrow R_{q(\vect{x})|_{\Sbf_E}}(\vect{x}) \wedge
  q(\vect{x})|_{\Sbf_{I_1}}$.
  By construction of $J'$ and $I'$, $h$ is also a homomorphism from
  $R_{q(\vect{x})|_{\Sbf_E}}(\vect{x}) \wedge
  q(\vect{x})|_{\Sbf_{I_1}}$
  to $J'$. Since $J'$ satisfies all rules of $\Pi^S_1$, one of the
  disjuncts of $p(\vect{y})$ is satisfied under $h$, as required.

  For $J \not\models \Pi^S_2$, let $I'$ be an extension of $I$ to
  $\Sbf_E \cup \Sbf_{I,2}$ that satisfies all rules of $\Pi_2$ and
  does not contain the $\mn{goal}_2$ relation. Let $J'$ be the
  corresponding extension of $J$. It suffices to show that $J'$
  satisfies all rules of $\Pi^S_2$. Thus, let
  $p(\vect{y}) \leftarrow q(\vect{x})$ be a rule in $\Pi^S_2$ and let
  $h$ be a homomorphism from $q(\vect{x})$ to $J'$. Then there is a
  rule $p(\vect{y}) \leftarrow q'(\vect{x}')$ in $\Pi_2$, a relation
  $R_{q''(\vect{x}'')}$ in $\Sbf'_{E}$, and an injective homomorphism
  $g$ from $q'(\vect{x}')$ to $ q''(\vect{x}'')$ such that
  $q(\vect{x})=g^{-1}(R_{q''(\vect{x}'')}(\vect{x}'')) \wedge
  q'(\vect{x}')|_{\Sbf_{I_2}}$.
  By construction of $J$ and $J'$, $h$ is also a homomorphism from
  $q'(\vect{x}')$ to~$I$.  Since $I'$ satisfies
  $ p(\vect{y}) \leftarrow q'(\vect{x}')$, one of the disjuncts of
  $p(\vect{y})$ is satisfied under~$h$, as required.
  
    \smallskip

    Now for Point~2. Let $I$ be an $\Sbf_E'$-instance of girth
    exceeding $w$ such that $I \models \Pi^S_1$ and $I\not\models
    \Pi^S_2$. Let $J$ be the $\Sbf_E$-instance that consists of all
    facts $r(a_{i_1},\dots,a_{i_k})$ such that for some fact
    $R_{q(x_1,\dots,x_n)}(a_1,\dots,a_n)$, we have
    $r(x_{i_1},\dots,x_{i_k}) \in q$. We show that $J \models \Pi_1$
    and $J \not\models \Pi_2$.

    For $J \models \Pi_1$, assume to the contrary of what is to be
    shown that there is an extension $J'$ of $J$ to schema
    $\Sbf_E \cup \Sbf_{I,1}$ such that all rules of $\Pi_1$ are
    satisfied and $\mn{goal}_1 \notin J'$. Let $I'$ be the
    corresponding extension of $I$ to $\Sbf'_E \cup \Sbf_{I,1}$. It
    suffices to show that all rules of $\Pi^S_1$ are satisfied in $I'$
    to obtain a contradiction against $I \models \Pi^S_1$. Thus, let
    $p(\vect{y}) \leftarrow q(\vect{x})$ be a rule in $\Pi^S_1$ and let
    $h$ be a homomorphism from $q(\vect{x})$ to $I'$. Then there is a
    rule $p(\vect{y}) \leftarrow q'(\vect{x}')$ in $\Pi_1$, a
    relation $R_{q''(\vect{x}'')}$ in $\Sbf'_{E}$, and an injective
    homomorphism $g$ from $q'(\vect{x}')$ to $ q''(\vect{x}'')$ such
    that
    $q(\vect{x})=g^{-1}(R_{q''(\vect{x}'')}(\vect{x}'')) \wedge
    q'(\vect{x}')|_{\Sbf_{I_2}}$.
    By construction of $J$ and $I'$, $h$ is also a homomorphism from
    $q'(\vect{x}')$ to~$I$. Since $J'$ satisfies
    $ p(\vect{y}) \leftarrow q'(\vect{x}')$, one of the disjuncts of
    $p(\vect{y})$ is satisfied under~$h$, as required.

    For $J \not\models \Pi_2$, let $I'$ be an extension of $I$ to
    $\Sbf'_E \cup \Sbf_{I,2}$ that satisfies all rules of $\Pi^S_2$
    and does not contain the $\mn{goal}_2$ relation.  Let $J'$ be the
    corresponding extension of $J$. It suffices to show that $J'$
    satisfies all rules of~$\Pi_2$.  Thus, let $ p(\vect{y})
    \leftarrow q(\vect{x})$ be a rule in $\Pi_2$ and let $h$ be a
    homomorphism from $q(\vect{x})$ to $J'$.  Let the query
    $q'(\vect{x}')$ be obtained from $q(\vect{x})$ by identifying all
    variables that $h$ maps to the same target. For simplicity, let us
    assume first that $q'$ is connected.

Partition the EDB atoms of $q'(\vect{x}')$ into components
    as follows: every reflexive atom $r(x,\dots,x)$ forms a component
    and every maximal biconnected set of non-reflexive atoms forms a
    component.  Let $q_1(\vect{x}_1),\dots,q_n(\vect{x}_n)$ be the
    components obtained in this way, enriched with IDB atoms in the
    following way: if $P(x)$ is in $q(\vect{x})$ with $P$ IDB and $x$
    occurs in $q_i(\vect{x}_i)$, then $q_i(\vect{x}_i)$ contains
    $P(x)$. It can be verified that distinct components share at most
    one variable and that the undirected graph obtained by taking the
    non-reflexive components as nodes and putting edges between
    components that share a variable is a tree. For ech tree in the
    tree, choose a component that is the root to turn the undirected
    tree into a directed one, allowing us to speak about successors,
    predecessors, etc. Slightly extend the tree by adding each
    reflexive component as a leaf below some node that contains the
    variable in the reflexive component; if there is no such node, the
    component forms an extra tree.

    For every component $q_i(\vect{x}_i)$, $\Pi_2$ contains an
    associated rule. Recall that we assume $\Pi_2$ to be the result of
    the first step of the construction of $\Pi_2^S$. If the
    $q_i(\vect{x}_i)$ is a leaf in the tree, then the associated rule
    takes the form
    $p_i(\vect{y}_i) \vee Q_i(x_i) \leftarrow q_i(\vect{x}_i)$ where
    $p_i(\vect{y}_i)$ is the restriction of $p(\vect{y})$ to atoms
    that are nullary or contain a variable from $\vect{x}_i$, $Q_i$ is
    a fresh unary relation, and $x_i$ is the variable that the
    component shares with the component which is its predecessor in
    the tree. For non-leafs, the rule body is additionally enriched
    with atoms $Q_j(x)$ where $Q_j$ is a fresh IDB introduced for a
    successor node. For the root node, no fresh IDB relation is
    introduced.
      
    We now make a bottom-up pass over the tree as follows. Consider
    the rule $p'_i(\vect{y}_i) \leftarrow q'_i(\vect{x}_i)$ associated
    with the current node. The homomorphism $h$ from above is also a
    homomorphism from $q'_i(\vect{x}_i)$ to $J'$; this is clear for
    leaf nodes and can inductively be verified for inner nodes. Take
    the corresponding rule $\rho$ in $\Pi_2^S$, the one that
    introduces a new EDB relation for the EDB atoms in
    $q'_i(\vect{x}_i)$. By construction of $J'$, since
    $q'_i(\vect{x}_i)$ is biconnected without reflexive loops or a
    single reflexive loop and because the girth of $I'$ is higher than
    that of $q'_i(\vect{x}_i)$, the $h$-image of all EDB-atoms in
    $q'_i(\vect{x}_i)$ must have been derived from a single fact
    in~$I$. There is another rule $\rho'$ in $\Pi_2^S$ in which the
    EDB-relation in the body of $\rho$ is replaced with the relation
    from that fact. This rule applies in $I'$ and thus one of the
    atoms from its head is true. If this is an atom from
    $p(\vect{x})$, we are done with the entire proof. If this is a
    fresh IDB atom, we are done with this tree node and can continue
    in our bottom-up tree walk. We will be done at the root at latest
    since the associated rule contains no fresh IDB relation.

    This finishes the case where the query $q'(\xbf')$ is connected.
    the general case, the tree has to be replaced by a forest. There
    will be exactly one tree in the forest in which the rule
    associated with the root does not have a fresh IDB relation in the
    head. For those trees where the root rule has a fresh such
    relation, we find another tree where that relation occurs in the
    body of the rule associated with a leaf. This defines an order on
    the trees, which is acyclic. We process the trees in this order,
    each single tree essentially as in the connected case.
\end{proof}

\subsubsection*{From Containment to Relativized Emptiness}

\begin{lemma}
  For any $g > 0$, 
  \begin{enumerate}

  \item if $\Pi_1 \not\subseteq \Pi_2$, then $\Pi$ is non-empty
    w.r.t.\ $D$.

  \item if $\Pi$ is non-empty w.r.t.\ $D$ on instances of girth $>g$,
    for any $g > 0$, then $\Pi_1 \not\subseteq_{>g} \Pi_2$.

  \end{enumerate}
\end{lemma}
\begin{proof}
  For Point~1, assume $\Pi_1 \not\subseteq \Pi_2$. Then there is an
  $\Sbf_E$-instance $I$ such that $I \models \Pi_1$ and $I \not
  \models \Pi_2$. Let $J$ be an extension of $I$ to signature $\Sbf_E
  \cup \Sbf_{I,2}$ such that all rules in $\Pi_2$ are satisfied and
  $\mn{goal}_2() \notin J$. Add $\overline{P}(a)$ to $J$ if $P(a)
  \notin J$ for all unary $P \in \Sbf_{I,2}$ and $a \in \mn{adom}(J)$,
  and add $\overline{P}()$ to $J$ if $P() \notin J$ for all nullary $P
  \in \Sbf_{I,2}$. Clearly, (the extended) $J$ is over schema
  $\Sbf'_E$ and satisfies $D$. To show that $\Pi$ is non-empty w.r.t.\
  $D$, it thus remains to argue that $J \models \Pi$.

  Assume to the contrary that this is not the case.  Then there is an
  extension $J'$ of $J$ to signature $\Sbf'_E \cup \Sbf_{I,1}$ such
  that all rules of $\Pi$ are satisfied, but
  $\mn{goal}_1() \notin J'$. Since the restriction of $J'$ to
  $\Sbf_E$-facts is $I$, it remains to argue that $J'$ satisfies all
  rules in $\Pi_1$. Let $\rho p(\vect{y}) \leftarrow = q(\vect{x})$ be
  such a rule and let $h$ be a homomorphism from $q$ to~$J'$. We
  define an extension $q'$ of $q$ as follows: conjunctively add to $q$
  all $P() \in J \cap \Sbf_{I,2}$ and all $\overline{P}() \in J$ with
  $P \in \Sbf_{I,2}$; moreover, for all variables $x$ in~$q$ and all
  facts $P(h(x)) \in J$ with $P \in \Sbf_{I,2}$, conjunctively add the
  atom $P(x)$ to $q$, and likewise for facts $\overline{P}(h(x))$ and
  atoms $\overline{P}(x)$.  Since $J'$ satisfies all rules in $\Pi_2$,
  it can be verified that the rule
  $p(\vect{y}) \leftarrow q'(\vect{x})$, which is added to $\Pi$ in
  the first step of its construction, is not removed in the second
  step of the construction. Since $h$ is a homomorphism from $q'$ to
  $J'$ and $J'$ satisfies all rules in $\Pi$,
  $J' \models p[h(\vect{y})]$ and thus $J'$ satifies $\rho$ as
  required.
  
\smallskip

Now for Point~2. Assume that $\Pi$ is non-empty w.r.t.\ $D$ on
instances of girth $>g$, with $g>0$. Then there is an
$\Sbf'_E$-instance $I$ of girth $>g$ with $I \models \Pi$ and $I
\models D$. By construction of $\Pi$, this implies that (i)~$P() \in
I$ or $\overline{P}() \in I$ for every nullary $P \in \Sbf_{I,2}$
(otherwise no rule of $\Pi$ would be applicable, implying $I
\not\models \Pi$).  Moreover, we can assume w.l.o.g.\ that (ii)~every
$a \in \mn{adom}(I)$ satisfies $P(a) \in I$ or $\overline{P}(a) \in I$
for every unary $P \in \Sbf_{I,2}$; in fact, it this is not the case,
then we can simply replace $I$ with its restriction to those elements
$a$ that satisfy the condition (elements which do not can never be
involved in rule applications). Let $J$ be the restriction of $I$ to
schema $\Sbf_E$. By construction of $\Pi$ and due to Conditions~(i)
and~(ii), $I \models \Pi$ implies $J \models \Pi_1$. To finish the
proof, it would this be sufficient to show that $I$ witnesses $J \not
\models \Pi_2$. This, however, need not be the case: while we know
that $\mn{goal}_2() \in I$ as otherwise no rule of $\Pi$ would be
applicable, it need not be the case that all rules in $\Pi_2$ are
satisfied in $I$. We thus show how to first manipulate $I$ such that
$I \models \Pi$ still holds, $I$ does still not contain
$\mn{goal}_2()$, and $I$ satisfies all rules in $\Pi_2$. In fact,
we exhaustively apply the following.

Assume that there is a rule $\rho$ in $\Pi_2$ that is not satisfied in
$I$.  Since $\Pi_2$ is simple, $\rho$ has the form
$Q_1(y_1) \vee \cdots \vee
Q_\ell(y_\ell) \leftarrow
A(\vect{x}) \wedge P_{1}(\vect{x}_{1}) \wedge \cdots \wedge
P_{k}(\vect{x}_{k})$ with $A$ EDB and all $P_i$ and $Q_i$ IDB. Let $h$ be a
homomorphism from the rule body to $I$ such that $Q_i(h(y_i)) \notin
I$ for $1 \leq i \leq \ell$.  We modify $I$ by removing the fact
$A(h(\vect{x}))$, resulting in instance $I^-$. Clearly, the
application of $\rho$ via $h$ is no longer possible and it remains to
show that $I^- \models \Pi$. Assume that this is not the case, that
is, there is an extension $J^-$ of $I^-$ to schema $\Sbf_E' \cup
\Sbf_{I,1}$ such that $\mn{goal}_1() \notin J^-$ and $J^-$ satisfies
all rules of $\Pi$. Let $J$ be the corresponding extension of $I$,
that is, $J$ and $J^-$ differ only in the presence of the fact
$A(h(\vect{x}))$. We show that $J$ satisfies all rules of $\Pi$,
contradicting $I \models \Pi$. Clearly, we need to consider only rules
$\rho'$ whose only $\Sbf_E$-atom is of the form $A(\vect{x}')$ and
only homomorphisms $h'$ from the body of $\rho'$ to $J$ such that
$h(\vect{x}')=h(\vect{x})$. Fix such a $\rho'$ and $h'$. Since
$I$ and thus also $J$ has girth $>1$ and since $\vect{x}'$
contains all variables from the body of $\rho'$, $h'$ must be
injective.  By definition of $\Pi$ and because of the homomorphism
$h'$, we must have
\begin{itemize}

\item $P() \in J$ (resp.\ $\overline{P}() \in J$) implies that the
  body of $\rho'$ has a conjunct $P()$ (resp.\ $\overline{P}$) for
  every nullary $P \in \Sbf_{I,2}$;

\item $P(h'(x)) \in J$ (resp.\ $\overline{P}(h'(x)) \in J$) implies
  that the body of $\rho'$ has a conjunct $P(x)$ (resp.\
  $\overline{P}(x)$) for unary $P \in \Sbf_{I,2}$ and all variables
  $x$ from $\vect{x}'$.
 
\end{itemize}
Because of the homomorphism $h$ and since $Q_i(h(y_i)) \notin I$ for
$1 \leq i \leq \ell$, the rule $\rho \in \Pi_2$ and the variable
substitution $h \circ {h'}^{-1}$ mean that the rule $\rho'$ was
removed during the second step of the construction of $\Pi$,
in contradiction to $\rho' \in \Pi$. 
\end{proof}

\subsubsection*{Deciding Relativized Emptiness}

\smallskip
\noindent
\LEMKthetaAreEnough*

\noindent
\begin{proof}
  Clearly, $K_\theta \models \Pi$ means that $K_\theta$ is a witness
  for $\Pi$ being non-empty w.r.t.\ $D$. Conversely, assume that there
  is an $\Sbf_E$-instance $I$ with $I \models D$ and $I \models
  \Pi$. Let $\theta$ be the 0-type of $I$. Then the mapping $h$
  defined by setting $h(a)=t_a$ for all constants $a$ in $I$ is a
  homomorphism from $I$ to $K_\theta$. It is well-known (and can be
  proved using a disjunctive version of the chase procedure) that
  truth of MDDLog queries is preserved under homomorphisms, thus $I
  \models \Pi$ implies $K_\theta \models \Pi$.
\end{proof}

\noindent
\LEMruntime*

\noindent
\begin{proof}
  Let us first analyze the time that it takes to check whether
  $K_\theta \not\models \Pi$, for one instance $K_\theta$.  The number
  of elements in $K_\theta$ is bounded by $2^{|D|}$. Note that the
  interpretation of the relations in $\Sbf_E \setminus \Sbf_D$ is
  trivial, and thus we do not need to explicitly construct these
  relations when building~$K_\theta$. Thus, $2^{\Omc(|D|)}$ is a bound
  on the number of facts in (the constructed part of) $K_\theta$ and
  on the time needed to build it. The number of guesses to be taken
  when constructing $K'_\theta$ is bounded by $2^{|D|} \cdot |\Sbf_I|$
  where $\Sbf_I$ is the IDB schema of $\Pi$, thus by
  $2^{|D|} \cdot |\Pi|$.  For each rule in $\Pi$, the number of
  candidate functions for homomorphisms from the rule body to
  $K'_\theta$ is bounded by $2^{|D|\cdot v}$ and it can be checked in
  time $\Omc(|\Pi| \cdot (|D|+|\Sbf_I|))$ whether a candidate is a
  homomorphism.  We thus need time
  $\Omc(|\Pi|^2) \cdot 2^{\Omc(|D|\cdot v)}$ per rule and time
  $\Omc(|\Pi|^3) \cdot 2^{\Omc(|D|\cdot v)}$ overall to deal with the
  instance $K_\theta$. There are at most $2^{|D|}$ many instances
  $K_\theta$, a factor that is absorbed by the bound that we have
  already computed.
\end{proof}

The next lemma is a main ingredient to the proof of
Lemma~\ref{thm:emptinessGirth}. Intuitively, it shows that a
semi-simple MDDLog programs with disjointness constraints can be
understood as a constraint satisfaction problem (CSP).  For two
instances $I$ and $J$, we write $I \rightarrow J$ if there is a
homomorphism from $I$ to $J$, which is defined in the standard way. Of
course, homomorphisms also have to respect nullary relations. In the
following, we call a finite instance a \emph{template} when we use it
as a homomorphism target, as in the CSP literature.
\begin{restatable}{lemma}{LEMNFtoCSP}\label{lem:NFtoCSP}
%
%
  For every 0-type $\theta$, there are templates $T_0,\dots,T_n$ in
  signature $\Sbf_{E}$ such that, for every $\Sbf_E$-instance $I$ of
  some 0-type $\theta' \subseteq \theta$, we have $I \not\models \Pi$
  iff $I \rightarrow T_i$ for some $i \leq n$.
\end{restatable}
\begin{proof}
  Just like the construction of the instances $K_\theta$, the
  construction of the templates $T_0,\dots,T_n$ is based on types.
  However, the types used for this purpose are formed over the schema
  $\Sbf=\Sbf_{I} \cup \Sbf_{D}$ instead of over the schema $\Sbf_D$,
  where $\Sbf_I$ is the IDB schema of~$\Pi$. To emphasize the
  difference between the two kinds of types, we from now on call the
  types introduced above \emph{types for} $\Sbf_D$.

  A \emph{type for~$\Sbf$} is a set $t \subseteq \Sbf$ that satisfies
  all rules in~$\Pi$, that is, if a rule $\rho$ mentions only nullary
  and unary relations (and thus involves at most a single variable
  since $\Pi$ is semi-simple w.r.t.\ $D$) and all these relations are
  in $t$, then at least one of the relations from the head of $\rho$
  is in $t$. Note that, in contrast to 0-types for $\Sbf_D$ and
  1-types for $\Sbf_D$ defined before, a type for \Sbf contains both
  unary and nullary relation symbols.  The restriction $\delta$ of a
  type $t$ for $\Sbf$ to nullary relations is a \emph{0-type
    for~$\Sbf$}. There is no need to define 1-types for \Sbf. We say
  that a type $t$ for $\Sbf$ is \emph{compatible} with a 0-type
  $\delta$ for $\Sbf$ if $\delta$ is the restriction of $t$ to nullary
  relations.

  Let $\theta$ be a 0-type for $D$. To construct the set of templates
  stipulated in Lemma~\ref{lem:NFtoCSP}, we define one template
  $T_\delta$ for every 0-type $\delta$ for $\Sbf$ that agrees with
  $\theta$ on relations from $\Sbf_D$ and does not contain the
  \mn{goal} relation. The elements of $T_\delta$ are the types for
  $\Sbf$ that are compatible with $\delta$ and $T_\delta$ consists of
  the following facts:
  \begin{enumerate}

  \item all facts $R(t_1,\dots,t_n)$ such that $R \in \Sbf_E \setminus
    \Sbf_D$ is of arity $n$ and there is no rule $\rho$ in $\Pi$ and
    variables $x_1,\dots,x_n$ such that the following conditions are
    satisfied:
  \begin{itemize}

  \item $R(x_1,\dots,x_n)$ occurs in the body of $\rho$;

  \item if $P(x_i)$ or $P()$ occurs in the body of $\rho$ with $P \in \Sbf$,
    then $P \in t_i$;

  \item for none of the disjuncts $P(x_i)$ in the head of $\rho$, we 
    have $P \in t_i$;

  \item for none of the disjuncts $P()$ in the head of $\rho$, we have
    $P \in \delta$;

  \end{itemize}

%
%
%
%
  %
  %

\item all facts $P(t)$ with $P \in t \cap \Sbf_D$;

\item all facts $P()$ with $P \in \delta \cap \Sbf_D$.

  \end{enumerate}
  We have to show that the templates $T_\delta$ are as required, that
  is, for every $\Sbf_E$-instance $I$ that has some 0-type $\theta'
  \subseteq \theta$ for $\Sbf_D$, we have $I \not\models \Pi$ iff $I
  \rightarrow T_\delta$ for some 0-type $\delta$ for $\Sbf$.

  \smallskip ``if''. Assume that $I$ is an $\Sbf_E$-instance of 0-type
  $\theta' \subseteq \theta$ for $\Sbf_D$ and that there is a
  homomorphism $h$ from $I$ to $T_\delta$. Extend $I$ to an $\Sbf_E
  \cup \Sbf_{I}$-instance $J$ by adding $P()$ whenever $P() \in \delta
  \cap \Sbf_I$ and $P(a)$ whenever $P \in h(a) \cap \Sbf_{I}$. Since
  the goal relation is not in $\delta$, it is not in~$J$; it thus
  remains to show that every rule $\rho$ of $\Pi$ is satisfied in
  $J$. First for rules $\rho$ that contain a body atom whose relation
  is from $\Sbf_E \setminus \Sbf_D$.  Since $\Pi$ is semi-simple
  w.r.t.~$D$, the body of $\rho$ consists of one atom
  $R(x_1,\dots,x_n)$ with $R \in \Sbf_E \setminus \Sbf_D$ plus atoms
  of the form $P(x_i)$ and $P()$ with $P \in \Sbf$. Assume that $g$ is
  a homomorphism from the body of $\rho$ to $J$. Then we have
  $R(g(x_1),\dots,g(x_n)) \in T_\delta$.  Moreover,
\begin{enumerate}

\item when $P(x_i)$ is in the body of $\rho$ and $P \in \Sbf$, then
  $P \in g(x_i)$ and 

\item when $P()$ is in the body of $\rho$ and $P \in \Sbf$, then
  $P \in g(x_i)$ for $1 \leq i \leq n$.

\end{enumerate}
In fact, Point~1 follows from Point~2 of the construction of
$T_\delta$ when $P \in \Sbf_D$ and from the definition of $J$ when
$P \in \Sbf_I$. Point~2 similarly follows from Point~3 of the
construction of $T_\delta$. With Points~1 and~2 above, Condition~1
from the construction of $T_\delta$ yields that the head of $\rho$
contains an atom $P(x_i)$ with $P \in g(x_i)$ or $P()$ with $P \in J$.
Thus, $\rho$ is satisfied in $J$. For rules $\rho$ that contain no
atom with a relation from $\Sbf_E \setminus \Sbf_D$, one can first
observe that, since $\Pi$ is semi-simple w.r.t.\ $D$, there is at 
most one variable in the rule. We can now argue in a similar way as
before that the rule is satisfied in $J$, using in particular the fact that
types for \Sbf satisfy the rules in $\Pi$.
  
\smallskip ``only if''. Assume that $I$ is an $\Sbf_E$-instance of
0-type $\theta' \subseteq \theta$ for $\Sbf_D$ and that $I \not\models \Pi$. Then
there is an extension $J$ of $I$ to signature $\Sbf_E \cup \Sbf_{I}$
that satisfies rules of $\Pi$ and does not contain the \mn{goal}
relation. For every constant $a$ in $I$, let $t_a$ be the set of all
unary relations $P \in \Sbf$ such that $P(a) \in J$ and of all nullary
relations $P \in \Sbf$ such that $P() \in J$. Clearly, $t_a$ is a
type for $\Sbf$. Set $h(a)=t_a$ for all constants $a$ in
$I$. It remains to verify that $h$ is a homomorphism from
$I$ to~$T_\delta$.

First consider facts $P(a) \in I$ with $P \in \Sbf_D$.  Then $P
\in t_a$, thus $P \in h(a)$, thus $P(a) \in T_\delta$ by definition of
$T_\delta$.

Now consider facts $R(a_1,\dots,a_n) \in I$ with $R \in \Sbf_E
\setminus \Sbf_D$. Assume to the contrary of what is to be shown that
$R(h(a_1),\dots,h(a_n)) \notin T_\delta$. By definition of $T_\delta$,
there is a rule $\rho$ of $\Pi$ and variables $x_1,\dots,x_n$ such
that
\begin{enumerate} 

  \item $R(x_1,\dots,x_n)$ occurs in the body of $\rho$, 

  \item if $P(x_i)$ or $P()$ occurs in the body of $\rho$ with $P \in \Sbf$,
    then $P \in t_i$;

  \item for none of the disjuncts $P(x_i)$ in the head of $\rho$, we 
    have $P \in t_i$;

  \item for none of the disjuncts $P()$ in the head of $\rho$, we have
    $P \in \delta$.

\end{enumerate}
By definition of $T_\delta$ and of $h$, Points~2 to~4 imply that 
\begin{enumerate}

\setcounter{enumi}{4}

\item if $P(x_i)$ (resp\ $P()$) occurs in the body of $\rho$ with $P \in \Sbf$, then
$P(a_i) \in J$ (resp.\ $P() \in J$);

\item for none of the disjuncts $P(x_i)$ in the head of $\rho$, we 
  have $P(a_i) \in J$;

\item for none of the disjuncts $P()$ in the head of $\rho$, we 
  have $P() \in J$;

\end{enumerate}
Since $R(x_1,\dots,x_n)$ occurs in the body of $\rho$ and $\Pi$ is
semi-simple w.r.t.\ $D$, the variables $x_1,\dots,x_n$ are all
distinct. We can thus define a function $g$ by setting $g(x_i)=a_i$.
Since $R(a_1,\dots,a_n) \in I$ and by Point~5, $g$ is a homomorphism
from the body of $\rho$ to $I$. By Points~6 and~7, $g$ witnesses that
$\rho$ is violated in $J$, in contradiction to $J$ satisfying all
rules in~$\Pi$.
\end{proof}

A second important ingredient to the proof of
Theorem~\ref{thm:emptinessGirth} is the following well-known lemma
which is originally due to Erd\H{o}s and concerns graphs of high girth
and high chromatic number and was adapted to the following formulation
in \cite{DBLP:journals/siamcomp/FederV98}.
\begin{lemma} 
\label{lem:erdoes}
  Let $\Sbf_E$ be a schema. For every $\Sbf_E$-instance $I$ and $g,s
  \geq 0$, there is an $\Sbf_E$-instance $I'$ such that
  \begin{enumerate}

  \item $I' \rightarrow I$, 

  \item $I'$ has girth exceeding $g$, and 

  \item for  every $\Sbf_E$-instance $J$ with at most $s$ elements,
    $J \rightarrow I$ iff $J \rightarrow I'$.

  \end{enumerate}
\end{lemma}

\smallskip
\noindent
\THMemptinessGirth*

\noindent
\begin{proof}
Clearly, we only
need to prove that if $\Pi$ is non-empty w.r.t.~$D$, then this is
witnessed by an instance of girth exceeding $g$. By
Lemma~\ref{lem:KthetaAreEnough}, $\Pi$ being non-empty w.r.t.~$D$
implies that $K_\theta \models \Pi$ for some 0-type $\theta$.  By
Lemma~\ref{lem:NFtoCSP}, we find templates $T_0,\dots,T_n$ in
signature $\Sbf_E$ such that, for every $\Sbf_E$-instance~$I$ of some
0-type $\theta' \subseteq \theta$, we have $I \not\models \Pi$ iff $I
\rightarrow T_i$ for some $i \leq
n$. 
Thus $K_\theta \not\rightarrow T_i$ for all $i \leq n$. By
Lemma~\ref{lem:erdoes}, there is an $\Sbf_E$-instance $K'_\theta$ such
that $K'_\theta \rightarrow K_\theta$ (and thus $K'_\theta$ satisfies
the constraints in $D$ and has some 0-type $\theta' \subseteq
\theta$), $K'_\theta$ has girth exceeding $h$, and for every
$\Sbf_E$-instance $I$ of size at most $s:=\max\{|T_0|,\dots,|T_n|\}$,
we have $K'_\theta \rightarrow I$ iff $K_\theta \rightarrow I$. The
latter implies $K'_\theta \not\rightarrow T_i$ for all $i \leq n$, and
thus $K'_\theta \models \Pi$ as desired. \qed
\end{proof}

\subsubsection*{Deriving Upper Bounds}

\smallskip
\noindent
\THMmainupper*

\noindent
\begin{proof}
  To decide whether $\Pi_1 \subseteq \Pi_2$, we first jointly simplify
  $\Pi_1$ and $\Pi_2$ as per Theorem~\ref{thm:simplify}, giving
  programs $\Pi^S_1$ and $\Pi^S_2$. These programs and
  Theorem~\ref{thm:toempty} give another program $\Pi$ and a set of
  constraints $D$ such that $\Pi$ is semi-simple w.r.t.\ $\Pi$. We
  decide whether $\Pi$ is empty w.r.t.\ $D$ and return the result. The
  size and complexity bounds given in
  Theorems~\ref{thm:simplify},~\ref{thm:toempty},
  and~\ref{thm:emptinessCompl} give the complexity bound stated in
  Theorem~\ref{thm:mainupper}. In fact, it can be verified that
  $|\Pi^S_i| \leq 2^{p(|\Pi_i|)}$, $|\Pi| \leq 2^{2^{p(|\Pi_2| \cdot
      \mn{log}|\Pi_1|)}}$, and $|D| \leq 2^{p(\Pi_2)}$ where $p$ is a
  polynomial.  Moreover, the variable width of $\Pi$ is bounded by
  that of $\Pi_1 \cup \Pi_2$ and it remains to plug these bounds
  into the time bounds stated in Theorem~\ref{thm:emptinessCompl}.

  It remains to argue that the algorithm is correct. If $\Pi_1
  \not\subseteq \Pi_2$, then $\Pi_1^S \not\subseteq \Pi_2^S$, thus
  $\Pi$ is non-empty w.r.t.\ $D$, thus ``no'' is returned. Let $w$ be
  the atom width of $\Pi_1 \cup \Pi_2$ (with the exception that $w=2$
  if that atom width is one).  If ``no'' is returned by our algorithm,
  then $\Pi$ is non-empty w.r.t.\ $D$. By
  Theorem~\ref{thm:emptinessGirth}, $\Pi$ is non-empty w.r.t.\ $D$ on
  instances of girth $>w$. Thus $\Pi_1^S \not\subseteq_{>w} \Pi_2^S$,
  implying $\Pi_1 \not\subseteq\Pi_2$.
\end{proof}

For the following lemma, let $\Sbf_E$ be the EDB schema of
$\Pi_1$ and $\Pi_2$, and let $\Sbf_{I,i}$ be the IDB schema of
$\Pi_i$, $i \in \{1,2\}$.

\noindent
\LEMconstredok*  

\smallskip
\noindent
\begin{proof}
  Assume that $\Pi_1 \not\subseteq \Pi_2$. Then there is an
  $\Sbf_E$-instance $I$ and a tuple $\abf \subseteq \mn{adom}(I)^k$
  such that $\abf \in \Pi_1(I) \setminus \Pi_2(I)$. Let $\Cbf_\Pi$ be
  the constants which occur in $\Pi_1 \cup \Pi_2$, and observe that
  $\Cbf_\Pi \subseteq \Cbf$. We can assume w.l.o.g.\ that all
  constants in $\abf$ are from \Cbf; if this is not the case, we can
  first rename constants in $I$ and $\abf$ that are from $\Cbf
  \setminus (\Cbf_\Pi \cup \abf)$ with fresh constants and then rename
  constants in $I$ and $\abf$ that are from $\abf \setminus \Cbf$ with
  constants from $\Cbf \setminus \Cbf_\Pi$.  By choice of $\Cbf$,
  there are enough constants of the latter kind.  It can now be
  verified that $I \models \Pi_1^\abf$ and $I \not\models
  \Pi_2^\abf$. In particular, an extension $J$ of $I$ to schema
  $\Sbf_E \cup \Sbf_{I,1}$ that satisfies all rules in $\Pi_1^\abf$
  and does not contain the goal relation gives rise to an extension
  $J'$ of $I$ that satisfies all rules in $\Pi_1^\abf$ and does not
  contain $\mn{goal}(\abf)$: start with $J'$ and then apply all goal
  rules of $\Pi_1$. By construction of $\Pi_1^\abf$ and since
  $\mn{goal}() \notin J'$, we have $\mn{goal}(\abf) \notin J'$.
  Similarly, an extension $J$ of $I$ to schema $\Sbf_E \cup
  \Sbf_{I,2}$ that satisfies all rules in $\Pi_2$ and does not
  contain $\mn{goal}(\abf)$ gives rise to an extension $J'$ of $I$
  that satisfies all rules in $\Pi_2^\abf$ and does not contain
  $\mn{goal}()$.
  
  Conversely, assume that $\Pi^\abf_1 \not\subseteq \Pi^\abf_2$.  Then
  there is an $\Sbf_E$-instance $I$ with $I \models \Pi^\abf_1$ and $I
  \not\models \Pi^\abf_1$. By considering appropriate extensions of
  $I$ to the schemas $\Sbf_E \cup \Sbf_{I,i}$, it can be verified in a 
  very similar way as above that $\abf \in \Pi_1(I) \setminus \Pi_2(I)$.
\end{proof}
  
\LEMconstredtwook*

\smallskip
\noindent
  \begin{proof}
    Assume that $\Pi_1 \not\subseteq \Pi_2$.  Then there is an
    $\Sbf_E$-instance $I$ such that $I \models \Pi_1$ and $I \not
    \models \Pi_2$. Let $J$ be the $\Sbf'_E$-instance obtained from
    $I$ by adding $R_a(a)$ whenever $a \in \mn{adom}(I) \cap \Cbf$.
    It can be verified that $J \models \Pi'_1$ and $J \not\models
    \Pi'_2$.

    \smallskip

    Conversely, assume that $\Pi'_1 \not\subseteq \Pi'_2$. Then there
    is an $\Sbf'_E$-instance $I$ such that $I \models \Pi'_1$ and $I
    \not \models \Pi'_2$. Let $J$ be the $\Sbf_E$-instance which is 
    obtained from $I$ by 
    \begin{itemize}

    \item dropping all facts that use a relation from $\Sbf'_E
      \setminus \Sbf_E$ and then

    \item taking the quotient according to the following equivalence
      relation on  $\mn{adom}(I)$:
      $$
        \{ (a,b) \mid \exists c \in \Cbf: R_c(a), R_c(b) \in I \};
      $$
      note that this is indeed an equivalence
      relation because the rules $\mn{goal}() \leftarrow R_{a_1}(x)
      \wedge R_{a_2}(x)$ in $\Pi'_2$ (for all distinct $a_1,a_2$) 
      imply that for any $a \in \mn{adom}(I)$, there is at most one
      $b \in \Cbf$ with $R_b(a) \in I$.

    \end{itemize}
    It can be verified that $J \models \Pi_1$ and $J \not\models
    \Pi_2$.
\end{proof}

\section{Ontology-Mediated Queries}

\subsubsection*{Preliminaries}

We first introduce the relevant OMQ languages and then provide missing
proofs.

\smallskip

A \emph{$\mathcal{ALCI}$-concept} is formed according to the syntax rule
$$
\begin{array}{r@{\;}c@{\;}l}
   C,D &::=& \top \mid \bot \mid A \mid \neg C \mid C \sqcap D \mid C
   \sqcup D \, \mid \\[\myeqnsep]
&& \exists r . C \mid \exists r^-.C \mid \forall r . C
   \mid \forall r^- . C
\end{array}
$$
where $A$ ranges over a fixed countably infinite set of \emph{concept
  names} and $r$ over a fixed countably infinite set of \emph{role
  names}.  An \emph{\ALC-concept} is an $\mathcal{ALCI}$-concept in
which the constructors $\exists r^- . C$ and $\forall r^-.C$ are not
used.  An \ALC-TBox (resp.\ \ALCI-TBox) is a finite set  of concept
inclusions $C \sqsubseteq D$, $C$ and $D$ $\ALC$-concepts (resp.\
\ALCI-concepts). A \emph{$\mathcal{SHI}$-TBox} is a finite set of
\begin{itemize}
\item \emph{concept inclusions}
$C \sqsubseteq D$, $C$ and $D$ $\mathcal{SHI}$-concepts, 

\item \emph{role
  inclusion} $r \sqsubseteq s$, $r$ and $s$ role names, and

\item
\emph{transitivity statements} $\mn{trans}(r)$, $r$ a role
name.

\end{itemize}
DL semantics is given in terms of interpretations. An
\emph{interpretation} takes that form $\Imc=(\Delta^\Imc,\cdot^\Imc)$
where $\Delta^\Imc$ is a non-empty set called the \emph{domain}
and $\cdot^\Imc$ is the \emph{interpretation function} which maps
each concept name $A$ to a subset $A^\Imc \subseteq \Delta^\Imc$
and each role name $r$ to a binary relation $r^\Imc \subseteq r^\Imc
\times r^\Imc$. The interpretation functions is extended to concepts
in the standard way, for example
$$
\begin{array}{r@{\;}c@{\;}l}
  (\exists r . C)^\Imc &=& \{ d \in \Delta^\Imc \mid \exists e \in 
                           C^\Imc: (d,e) \in r^\Imc \} \\[\myeqnsep] 
  (\exists r^- . C)^\Imc &=& \{ d \in \Delta^\Imc \mid \exists e \in 
                           C^\Imc: (e,d) \in r^\Imc \}.
\end{array}
$$
We refer to standard references such as \cite{DBLP:conf/dlog/2003handbook} for full
details. An intepretation is a \emph{model} of a TBox \Tmc if it
\emph{satisfies} all statements in \Tmc, that is,
\begin{itemize}

\item $C \sqsubseteq D \in \Tmc$ implies $C^\Imc \subseteq D^\Imc$;

\item $r \sqsubseteq s \in \Tmc$ implies $r^\Imc \subseteq s^\Imc$;

\item $\mn{trans}(r) \in \Tmc$ implies that $r^\Imc$ is transitive.

\end{itemize}

In description logic, data is typically stored in so-called ABoxes.
For uniformity with MDDLog, we use instances instead, identifying
unary relations with concept names, binary relations with role names,
and disallowing relations of any other arity. An interpretation \Imc
is a \emph{model} of an instance $I$ if $A(a) \in I$ implies
$a \in A^\Imc$ and $r(a,b) \in I$ implies $(a,b) \in r^\Imc$. We say
that an instance $I$ is \emph{consistent} with a TBox \Tmc if $I$ and
\Tmc have a joint model. We write $\Tmc \models r \sqsubseteq s$
if every model \Imc of \Tmc satisfies $r^\Imc \subseteq s^\Imc$.

An \emph{ontology-mediated query (OMQ)} takes the form
$Q=(\Tmc,\Sbf_E,q)$ with \Tmc a TBox, $\Sbf_E$ a set of concept and
role names, and $q$ a UCQ.  We use $(\Lmc,\Qmc)$ to refer to the set
of all OMQs whose TBox is formulated in the language \Lmc and where
the actual queries are from the language \Qmc. For example,
$(\ALC,\text{UCQ})$ refers to the set of all OMQs that consist of an
\ALC-TBox and a UCQ. For OMQs $(\Tmc,q)$ from
$(\mathcal{SHI},\mathcal{UCQ})$, we adopt the following additional
restriction: when \Tmc contains a transitivity $\mn{trans}(r)$ and
$\Tmc \models r \sqsubseteq s$, we disallow the use of $s$ in the
query $q$. Let $I$ be an $\Sbf_E$-instance and $\vect{a}$ a tuple of
constants from $I$. We write $I \models Q[\vect{a}]$ and call
$\vect{a}$ a \emph{certain answer to $Q$ on} $I$ if for all models
\Imc of $I$ and~\Tmc, we have $\Imc \models q[\vect{a}]$ (defined in
the usual way).

Containment between OMQs is defined in analogy with containment
between MDDLog programs: $Q_1=(\Tmc_1,\Sbf_E,q_1)$ is \emph{contained}
in $Q_2=(\Tmc_2,\Sbf_E,q_2)$, written $Q_1 \subseteq Q_2$, if for every
$\Sbf_E$-instance $I$ and tuple $\vect{a}$ of constants from $I$,
$I \models Q_1[\vect{a}]$ implies $I \models Q_2[\vect{a}]$.
This is different from the notion of containment considered in
\cite{DBLP:conf/kr/BienvenuLW12}, here called \emph{consistent
  containment}. We say that $Q_1=(\Tmc_1,\Sbf_E,q_1)$ is
\emph{consistently contained} in $Q_2=(\Tmc_2,\Sbf_E,q_2)$, written
$Q_1 \subseteq^\cbf Q_2$, if for every $\Sbf_E$-instance $I$ that is
consistent with $\Tmc_1$ and $\Tmc_2$ and every tuple $\vect{a}$
of constants from $I$, $I \models Q_1[\vect{a}]$ implies
$I \models Q_2[\vect{a}]$. We observe the following.
\begin{lemma}
  In $(\mathcal{SHI},\text{UCQ})$, consistent containment can be
  reduced to containment in polynomial time.
\end{lemma}
\begin{proof}(sketch) Let $Q_1=(\Tmc_1,\Sbf_E,q_1)$ and
  $Q_2=(\Tmc_2,\Sbf_E,q_2)$ be OMQs from
  $(\mathcal{SHI},\text{UCQ})$. Assume without loss of generality that
  all concept that occur in $\Tmc_1$ are in negation normal form, that
  is, negation is only applied to concept names but not to compound
  concepts. For every concept name $A$ in
  $\Sbf_E \cup \mn{sig}(\Tmc_1)$, introduce fresh concept names $A'$ and
  $\overline{A}'$ that do not occur in $Q_1$ and $Q_2$.  For every
  role name $r$ in $\Sbf_E \cup \mn{sig}(\Tmc_1)$, introduce a fresh
  role name $r'$. Define the TBox $\Tmc'_2$ as the extension of
  $\Tmc_2$ with the following:
  \begin{itemize}

  \item $A \sqsubseteq A'$ for all concept names $A$ in
    $\Sbf_E \cup \mn{sig}(\Tmc_1)$;

  \item $r \sqsubseteq r'$ for all role names $r$ in
    $\Sbf_E \cup \mn{sig}(\Tmc_1)$;

  \item every concept inclusion, role inclusion, and transitivity
    statement from $\Tmc_1$, each concept name $A$ replaced with $A'$,
    each subconcept $\neg A$ replaced with $\overline{A}'$, and each
    role name replaced with $r'$;

  \item the inclusions $\top \sqsubseteq A' \sqcup \overline{A}'$ and
    $A' \sqcap \overline{A}' \sqsubseteq B$ for all concept names $A$
    in $\Sbf_E \cup \mn{sig}(\Tmc_1)$, where $B$ is a fresh concept name.
    
  \end{itemize}
  Set $q'_2 = q_2 \vee \exists x \, B(x)$. It suffices to establish
  the following claim. The proof is not difficult and left to the
  reader.
  \\[2mm]
  {\bf Claim.} $Q_1 \subseteq^\cbf Q_2$ iff $Q_1 \subseteq Q'_2$.
\end{proof}

\subsubsection*{Upper Bound}

\noindent 
\THMthatreductionagain*

\noindent
\begin{proof}
  Let $Q=(\Tmc,\Sbf_E,q_0)$ be an OMQ from
  $(\mathcal{SHI},\text{UCQ})$. We use $\mn{sub}(\Tmc)$ to denote the
  set of subconcepts of (concepts occurring in) \Tmc. Moreover, let
  $\Gamma$ be the set of all tree-shaped conjunctive queries that can
  be obtained from a CQ in $q_0$ by first quantifying all answer
  variables, then identifying variables, and then taking a subquery.
  Here, a conjunctive query $q$ is \emph{tree-shaped} if (i)~the
  undirected graph $(V,\{\{ x,y \} \mid r(x,y) \in q \})$ is a tree
  (where $V$ is the set of variables in $q$), (ii)~$r_1(x,y),r_2(x,y)
  \in q$ implies $r_1=r_2$, and (iii)~$r(x,y) \in q$ implies $s(y,x)
  \notin q$ for all $s$. Every $q \in \Gamma$ can be viewed as a
  $\mathcal{ALCI}$-concept provided that we additionally choose a root
  $x$ of the tree. We denote this concept with $C_{q,x}$. For example,
  the query $\exists x \exists y \exists z \, r(x,y) \wedge A(y)
  \wedge s(x,z)$ yields the $\mathcal{ALCI}$-concept $\exists r . A \sqcap
  \exists s . \top$. Let $\mn{con}(q_0)$ be the set of all these
  concepts $C_{q,x}$ and let $\Sbf_I$ be the schema that consists of
  monadic relation symbols $P_C$ and $\overline{P}_C$ for each $C \in
  \mn{sub}(\Tmc) \cup \mn{con}(q_0)$ and nullary relation symbols
  $P_q$ and $\overline{P}_q$ for each $q \in \Gamma$. We are going to
  construct an MDDLog program $\Pi$ over EDB schema $\Sbf_E$ and IDB
  schema $\Sbf_I$ that is equivalent to $Q$.

  By a \emph{diagram}, we mean a conjunction $\delta(\vect{x})$ of
  atoms over the schema $\Sbf_E \cup \Sbf_I$. For an interpretation
  \Imc, we write $\Imc \models \delta(\vect{x})$ if there is a
  homomorphism from $\delta(\vect{x})$ to \Imc, that is, a map
  $h:\vect{x} \rightarrow \Delta^\Imc$ such that:
  \begin{enumerate}

  \item $A(x) \in \delta$ with $A \in \Sbf_E$ implies $h(x) \in A^\Imc$;

  \item $r(x,y) \in \delta$ with $r \in \Sbf_E$ implies $(h(x),h(y)) \in A^\Imc$;

  \item $P_q() \in \delta$ implies $\Imc \models q$ and
    $\overline{P}_q() \in \delta$ implies $\Imc \not \models q$;

  \item $P_C(x) \in \delta$ implies $h(x) \in C^\Imc$ and
    $\overline{P}_C() \in \delta$ implies $h(x) \notin C^\Imc$.

  \end{enumerate}
  We say that $\delta(\vect{x})$ is \emph{realizable} if there is
  an interpretation \Imc with $\Imc \models \delta(\vect{x})$.  A
  diagram $\delta(\vect{x})$ \emph{implies} a CQ
  $q(\vect{x}')$, with $\vect{x}'$ a sequence of variables
  from $x$, if every homomorphism from $\delta(\vect{x})$ to some
  interpretation \Imc is also a homomorphism from $q(\vect{x}')$
  to \Imc. The MDDLog program $\Pi$ consists of the following rules:
  \begin{enumerate}

  \item the rule $P_q() \vee \overline{P}_q() \leftarrow R(\vect{x})$
    for each $q \in \Gamma$, each $R \in \Sbf_E$, and each $R \in
    \Sbf_E$ where $\xbf = x_1,\dots,x_n$, $n$ the arity of $R$;

  \item the rule $P_C(x) \vee \overline{P}_C(x) \leftarrow
    R(\vect{x})$ for each $C \in \mn{sub}(\Tmc) \cup \mn{con}(q_0)$,
    each $R \in \Sbf_E$, and each tuple $\vect{x}$ that
    can be obtained from $x_1,\dots,x_n$ by replacing a single $x_i$
    with $x$ ($n$ the arity of $R$);


  \item the rule $\bot \leftarrow \delta(x)$ for each 
    non-realizable diagram $\delta(x)$ that contains a single
    variable $x$ and only atoms of the form $P_C(x)$, $C \in
    \mn{sub}(\Tmc) \cup \mn{con}(q_0)$;

  \item the rule $\bot \leftarrow \delta(\vect{x})$ for each 
    non-realizable connected diagram $\delta(\vect{x})$ that contains 
    at most two variables and at most three atoms;

  \item the rule $ \mn{goal}(\vect{x}') \leftarrow \delta(\vect{x})$
    for each diagram $\delta(\vect{x})$ that implies $q_0(\vect{x})$,
    has at most $|q_0|$ variable occurrences, and uses only relations
    of the following form: $P_q$, $P_C$ with $C$ a concept name that
    occurs in $q_0$, and role names from $\Sbf_E$ that occur in $q_0$.
    
  \end{enumerate}
  %
  To understand $\Pi$, a good first intuition is that rules of type~1
  and~2 guess an interpretation \Imc, rules of type~3 and~4 take care
  that the independent guesses are consistent with each other, with
  the facts in $I$ and with the inclusions in the TBox~\Tmc, and rules
  of type~5 ensure that $\Pi$ returns the answers to $q_0$ in~$\Imc$.

  However, this description is an oversimplification. Guessing \Imc is
  not really possible since \Imc might have to contain additional
  domain elements to satisfy existential quantifiers in \Tmc which may
  be involved in homomorphisms from (a CQ in) $q_0$ to $\Imc$, but new
  elements cannot be introduced by MDDLog rules. Instead of
  introducing new elements, rules of type~1 and~2 thus only guess the
  tree-shaped queries that are satisfied by those
  elements. Tree-shaped queries suffice because $\mathcal{SHI}$ has a
  tree-like model property and since we have disallowed the use of
  roles in the query that have a transitive subrole. The notion of
  `diagram implies query' used in the rules of type~4 takes care that
  the guessed tree-shaped queries are taken into account when looking
  for homomorphisms from $q_0$ to the guessed model. A more detailed
  explanation can be found in the proof of Theorem~1 of
  \cite{DBLP:conf/pods/BienvenuCLW13}. In fact, the construction used
  there is identical to the one used here, with two exceptions. First,
  we use predicates $P_C$ and $\overline{P}_C$ for every concept $C
  \in \mn{sub}(\Tmc) \cup \mn{con}(q_0)$ while the mentioned proof
  uses a predicate $P_t$ for every subset $t \subseteq \mn{sub}(\Tmc)
  \cup \mn{con}(q_0)$. And second, our versions of Rules~3-5 are
  formulated more carefully. 
  It can be verified that the correctness proof
  given in \cite{DBLP:conf/pods/BienvenuCLW13} is not affected by
  these modifications.  The modifications do make a difference
  regarding the size of $\Pi$, though, which we analyse next.

  It is not hard to see that the number of rules of type~1 is bounded
  by $2^{|q|^2}$, the number of rules of type~2 is bounded by
  $|\Tmc|$, the number of rules of type~3 is bounded by $2^{2^{|q|
      \cdot \mn{log}|\Tmc|}}$, the number of rules of type~4 is
  bounded by $2^{p(|q|\cdot \mn{log}|\Tmc|)}$ for some polynomial $p$,
  and the number of rules of type~5 is bounded by $2^{p(|q|)}$.
  Consequently, the overall number of rules is bounded by $2^{2^{p(|q|
      \cdot \mn{log}|\Tmc|)}}$ and so is the size of $\Pi$. The bounds
  on the size of the IDB schema and number of rules in $\Pi$ stated in
  Theorem~\ref{thm:thatreductionagain} are easily verified.  The
  construction can be carried out in double exponential time since for
  a given diagram $\delta(\vect{x})$ and CQ $q(\vect{x}')$, with
  $\vect{x}'$ a sequence of variables from $x$, it can be decided in
  2{\sc ExpTime} whether $\delta(\vect{x})$ implies $q(\vect{x}')$.
\end{proof}

\noindent 
\THMALCI*

\begin{proof}
  We convert $Q_1$ and $Q_2$ into MDDLog programs as per
  Theorem~\ref{thm:thatreductionagain} and then remove the answer
  variables according to the proof of Theorem~\ref{thm:mainupperPLUS}.
  Analyzing the latter construction reveals that it produces programs
  of size $r \cdot 2s \cdot a^s$ where $r$ is the number of rules of
  the input program, $s$ is the rule size, and $a$ the arity. Moreover
  the IDB schema is not changed and rule size at most doubles. The
  $\Pi_1,\Pi_2$ obtained by these two first steps thus still satisfy
  Conditions~1-3 of Theorem~\ref{thm:thatreductionagain} except that
  $|q|$ in the last point has to be replaced by $2|q|$.  

  The joint simplifications $\Pi^S_1$ and $\Pi^S_2$ from
  Theorem~\ref{thm:simplify} then have size $|\Pi_i^S| \leq
  2^{2^{p(|q_i| \cdot \mn{log}|\Tmc_i|)}}$ and their variable width is
  bounded by (the rule size of $\Pi_i$ and thus by) $2|q|$. Let us
  analyze the size of the IDB schema of $\Pi^S_i$. First note that the
  initial variable identification step can be ignored. In fact, we
  start with at most $2^{2^{p(|q_i|\cdot \mn{log}|\Tmc_i|)}}$ rules, each
  of size at most $2|q_i|$. Thus variable identification results in 
  a factor of $(2|q_i|)!$, which is absorbed by $2^{2^{p(|q_i|\cdot \mn{log}|\Tmc|)}}$.
  The other parameters are not changed by variable identification.


  When making rules biconnected in the construction of $\Pi^S_1$ and
  $\Pi^S_2$, we need not worry about rules of type~1-2 and~4-5. The
  reason is that there are only $2^{p(|q_i| \cdot \mn{log}|\Tmc_i|)}$
  many such rules, each of size at most $2|q_i|$, and thus the number
  of additional IDB relations introduced for making them biconnected
  is also bounded by $2^{p(|q_i| \cdot \mn{log}|\Tmc_i|)}$. Rules of
  type~3, on the other hand, are of a very restricted form, namely
  $$
    \bot \leftarrow P_{C_1}(x) \wedge \cdots \wedge P_{C_n}(x)
  $$
  with $C_1,\dots,C_n \in \mn{sub}(\Tmc) \cup \mn{con}(q_0)$.
  These rules are biconnected and thus we are done in the Boolean
  case. In the non-Boolean case, rules of the above form are 
  manipulated in the second step of the reduction of non-Boolean
  MDDLog programs to Boolean MDDLog programs. The result
  are rules of exactly the same shape, but also rules of the form
  $$
    \bot \leftarrow P_{C_1}(x_1) \wedge R_a(x_1) \wedge \cdots \wedge
    P_{C_n}(x_n) \wedge R_a(x_n).
  $$
  The latter rules have to be split up to be made biconnected. This
  will result in rules of the form
  $$
    \bot \leftarrow Q_1() \wedge \cdots \wedge Q_n() \quad
    \text{ and }
    \quad
    Q_i() \leftarrow P_C(x) \wedge R_a(x)
  $$
  where $R_a$ is one of the fresh IDB relations introduced in the
  mentioned reduction. Clearly, there are only $2^{p(|q_i| \cdot
    \mn{log}|\Tmc_i|)}$ many rule bodies of the latter form and thus
  it suffices to introduce at most the same number of fresh IDB
  relations $Q_i$. In summary, we have shown that a careful
  construction of $\Pi^S_i$ can ensure that the size of the IDB schema
  of $\Pi^S_i$ is bounded by $2^{p(|q_i| \cdot \mn{log}|\Tmc_i|)}$.

  It can now be verified that the program $\Pi$ from
  Theorem~\ref{thm:toempty} has size at most
  $2^{2^{p(|q_1|\cdot|q_2|\cdot\mn{log}|\Tmc_1|\cdot\mn{log}|\Tmc_2|)}}$
  and $D$ has size $2^{p(|q_2|\cdot\mn{log}|\Tmc_2|)}$. Applying
  Theorem~\ref{thm:emptinessCompl} gives the complexity bound stated
  in Theorem~\ref{thm:ALCI}.
\end{proof}


\subsubsection*{Lower Bound}

\noindent
The following result establishes the lower bound in Point~3 of
Theorem~\ref{thm:alcicompl}. We state it here even in a slightly
stronger form. $\mathcal{ELIU}$ denotes the description logic that
admits only the concept constructors $\top$, $\sqcap$, $\sqcup$, and
$\exists$ and $\ELI_\bot$ denotes the DL with the constructors $\top$,
$\bot$, $\sqcap$, and $\exists$.  With BAQ, we denote the class of
\emph{Boolean atomic queries}, that is, queries of the form $\exists x
\, A(x)$ with $A$ a concept name.
\begin{restatable}{theorem}{THMALCIlower}
  Containment of an $(\mathcal{ELIU},\text{BAQ})$-OMQ in an
  $(\mathcal{ELI}_\bot,\text{CQ})$-OMQ is 2{\sc NExpTime}-hard.
\end{restatable}

\smallskip
\noindent
The overall strategy of the proof is similar to that of our proof of
the lower bounds stated in Theorem~\ref{thm:hardness1}, but the
details differ in a number of respects. Instead of reducing 2-exp
\torus tiling problem, we now reduce the 2-exp torus tiling problem.
The definition is identical except that a tiling $f$ for the
latter problem additionally needs to satisfy
$(f(2^{2^n}-1,i),f(0,i)) \in \horiz$ and
$(f(i,2^{2^n}-1),f(i,u)) \in \vertical$ for all $i < 2^{2^n}$.  

We first implement the reduction using UCQs instead of CQs and
then adapt the proof to CQs.  In the previous reduction, the role name
$r$ was used to connect neighboring grid nodes and nodes in counting
trees. In the current reduction, we replace $r$ with the role
composition $r^-;r$ where $r^-$ denotes the inverse of $r$ and which
behaves like a reflexive-symmetric role. We use $S$ as an abbreviation
for $r^-;r$. In particular, $\exists S . C$ stands for
$\exists r^- . \exists r .C$ and $\forall S . C$ stands for
$\forall r^- . \forall r
.C$. 
Some other details of the reduction are also different than
before. Counting trees now have depth $m+2$ instead of $m$, but no
branching occurs on the last two levels of the tree. We also have
three different versions of counting trees: one which uses the
concept names $B_1,\overline{B}_1$ and $B_2,\overline{B}_2$ to store
the two counters, one that uses $B_3,\overline{B}_3$ and
$B_4,\overline{B}_4$, and one that uses $B_5,\overline{B}_5$ and
$B_6,\overline{B}_6$. We say that the trees are of type 0, 1, or 2 to
distinguish between the different versions.  In the grid
representation, we cycle through the types: from left to right and
bottom to top, every tree of type 0 is succeeded by trees of type 1
which are succeeded by trees of type 2 which are succeeded by trees of
type 3. Note that this refers to trees below grid nodes, but also to
trees below horizontal and vertical step nodes. All this prepares for
the construction of the UCQ later on.

Let $P$ be a 2-exp torus tiling problem and $w_0$ an input to $P$ of
length $n$. We construct TBoxes $\Tmc_1,\Tmc_2$ and OMQs
$Q_i=(\Tmc_i,\Sbf_E,q_i)$, $i \in \{1,2\}$, such that
$Q_1 \subseteq Q_2$ iff there is a tiling for $P$ and $w_0$. The
schema $\Sbf_E$ consists of
\begin{itemize}

\item the EDB symbols $r$, $B_i$, $\overline{B}_i$, $i \in
  \{1,\dots,6\}$;

\item concept names $A_0,\dots,A_{m-1}$ and
  $\overline{A}_0,\dots,\overline{A}_{m-1}$ to implement a binary
  counter that identifies the position of each leaf in a counting
  tree;

\item concept names $L_0,\dots,L_{m+2}$ that identify the
  levels in counting trees.

\end{itemize}
We now construct the TBox $\Tmc_1$. We first define concept inclusions
which verify that a grid node has a proper attached counting tree. We
start with identifying nodes on level $m+2$ by the concept name
$\mn{lev}_{m+2}^{G,t}$ where $t \in \{0,1,2\}$ describes the type
of the counting tree as explained above. We only give the construction
explicitly for $\mn{lev}_{m+2}^{G,0}$, which implements the two
counters using $B_1,\overline{B}_1$ and $B_2,\overline{B}_2$:
$$
\begin{array}{c}
  A_i \sqsubseteq  V_i \qquad \overline{A}_i  \sqsubseteq  V_i \qquad 0 \leq i <m\\[\myeqnsep]
  V_0 \sqcap \cdots \sqcap V_{m-1} \sqcap B_1\sqcap 
  B_2 \sqcap L_{m+2}
\sqsubseteq \mn{lev}_{m+2}^{G,0} \\[\myeqnsep]
  V_0 \sqcap \cdots \sqcap V_{m-1} \sqcap \overline{B}_1\sqcap 
  \overline{B}_2  \sqcap L_{m+2}
\sqsubseteq
  \mn{lev}_{m+2}^{G,0}.
\end{array}
$$
%
%
To make the UCQ work later on, we need that level $m+1$-nodes are
labeled complementarily regarding the concept names
$A_i,\overline{A}_i$, $i \leq m$.  We thus identify nodes on level
$m+1$ as follows:
$$
\begin{array}{r@{\;}c@{\;}ll}
  A_i \sqcap \exists S . (\mn{lev}_{m+2}^{G,0} \sqcap \overline{A}_i) &
  \sqsubseteq & A\mn{ok}_i & 0 \leq i \leq m\\[\myeqnsep]
  \overline{A}_i \sqcap \exists S . (\mn{lev}_{m+2}^{G,0} \sqcap A_i) & \sqsubseteq &
  A\mn{ok}_i & 0 \leq i \leq m \\[\myeqnsep]
  \multicolumn{4}{c}{
  A\mn{ok}_0 \sqcap \cdots \sqcap A\mn{ok}_{m-1} \sqcap 
L_{m+1} \sqsubseteq \mn{lev}_{m+1}^{G,0}}  
\end{array}
$$
Note that the first two lines may speak about different
$S$-successors. It is thus not clear that they achieve the intended
complementary labeling.  Moreover, we have not yet made sure that
level $m+2$-nodes are labeled with only one of $A_i,\overline{A}_i$
for each $i$ and with only one $B_j,\overline{B}_j$ for each $j \in
\{1,2\}$.
We fix these problem by including the
following concept inclusions in $\Tmc_2$:
$$
\begin{array}{r@{\,}c@{\,}l}
  L_{m+1} \sqcap \exists S . (L_{m+2} \sqcap A_i) \sqcap  \exists S 
  . (L_{m+2} \sqcap \overline{A}_i) &\sqsubseteq& \bot  \\[\myeqnsep]
   L_{m+1} \sqcap \exists S . (L_{m+2} \sqcap B_j) \sqcap  \exists S 
   . (L_{m+2} \sqcap \overline{B}_j) &\sqsubseteq& \bot
\end{array}
$$
where $i$ ranges over $0,\dots,m-1$ and $j$ over $\{1,2\}$.  
These
inclusions ensure that all relevant successors are labeled identically
regarding the relevant concept names: otherwise the instance is 
inconsistent with $\Tmc_2$ and thus makes $Q_2$ true, which rules
it out as a witness for non-containment.

We next make sure that every level $m$ node has a level $m+1$ node
attached and that its labeling is again complementary (in other words,
the labeling of the level $m$ node agrees with the labeling of the
level $m+2$ node below the attached level $m+1$ node):
$$
\begin{array}{r@{\;}c@{\;}ll}
  A_i \sqcap \exists S . (\mn{lev}_{m+1}^{G,0} \sqcap \overline{A}_i) &
  \sqsubseteq & A\mn{ok}'_i & 0 \leq i \leq m\\[\myeqnsep]
  \overline{A}_i \sqcap \exists S . (\mn{lev}_{m+1}^{G,0} \sqcap A_i) & \sqsubseteq &
  A\mn{ok}'_i & 0 \leq i \leq m \\[\myeqnsep]
  \multicolumn{4}{c}{
  A\mn{ok}'_0 \sqcap \cdots \sqcap A\mn{ok}'_{m-1} 
  \sqcap L_{m+1} \sqsubseteq \mn{lev}_{m}^{G,0}}  
\end{array}
$$
We also include the following in $\Tmc_2$:
$$
  L_{m} \sqcap \exists S . (L_{m+1} \sqcap A_i) \sqcap  \exists S 
  . (L_{m+1} \sqcap \overline{A}_i) \sqsubseteq \bot 
$$
where $i$ ranges over $0,\dots,m-1$. 
We next
verify the remaining levels of the tree. To make sure that the
required successors are present on all levels, we branch on the
concept names $A_i$, $\overline{A}_i$ at level $i$ of a counting tree
and for all $j< i$, keep our choice of $A_j$, $\overline{A}_j$:
$$
\begin{array}{r@{\;}c@{\;}l}
   \exists S . (\mn{lev}^{G,0}_{i+1} \sqcap A_{i}) \sqcap 
   \exists S . (\mn{lev}^{G,0}_{i+1} \sqcap \overline{A}_{i}) & \sqsubseteq & \mn{Succ} \\[\myeqnsep]
   A_j \sqcap \exists S . (\mn{lev}^{G,0}_{i+1} \sqcap A_j) &\sqsubseteq& \mn{Ok}_j \\[\myeqnsep]
   \overline{A}_j \sqcap \exists S . (\mn{lev}^{G,0}_{i+1} \sqcap \overline{A}_j) &\sqsubseteq&
   \mn{Ok}_i \\[\myeqnsep]
   \mn{Succ} \sqcap \mn{Ok}_0 \sqcap \cdots \sqcap \mn{Ok}_{i-1}
   \sqcap L_i
   &\sqsubseteq& \mn{lev}^{G,0}_{i}
\end{array}
$$
where $i$ ranges over $0,\dots,m-1$ and $j$ over $0,\dots,i-1$. Again, lines one
to three may speak about different successors and we need to make sure
that all those successors are labeled identically. This is done by
adding the following inclusions to $\Tmc_2$:
$$
L_{i} \sqcap \exists S . (L_{i+1} \sqcap A_j)  \sqcap \exists S 
. (L_{i+1} \sqcap \overline{A}_j) \sqsubseteq \bot 
$$
where the ranges of $i$ and $j$ are as above. 
This finishes the verification of the counting tree. We do
not use self step nodes in the current reduction, so a grid node is
simply the root of a counting tree where both counter values are
identical:
$$
  \mn{lev}^{G,0}_0 \sqsubseteq \mn{gactive}^0.
$$
The superscript $\cdot^0$ in $\mn{gactive}^0$ indicates that the
counting tree of which this node is the root is of type~0. Concept
inclusions that set $\mn{gactive}^1$ and $\mn{gactive}^2$ are defined
analogously, replacing $B_1,\overline{B}_1$ and $B_2,\overline{B}_2$
appropriately, as explained above.  In a similar way, we can verify
the existence of counting trees below horizontal step nodes and
vertical step nodes, signalling the existence of such trees by the
concept names $\mn{hactive}^t$ and $\mn{vactive}^t$, $t \in
\{0,1,2\}$. In contrary to the counting trees between grid nodes,
counting trees below horizontal and vertical step nodes need to
properly increment the counters, as in the previous reduction.
Details are slightly tedious but straightforward and thus omitted.


We also use $\Tmc_1$ to enforce that all grid nodes are labeled with a
tile type. However, as we shall see below we cannot use all nodes
labeled with a $\mn{gactive}^t$-concept as grid nodes, but only those
ones that have an $S$-neighbor which is labeled with
$\mn{hactive}^{t \oplus 1}$ or with $\mn{vactive}^{t \oplus 1}$ where
$\oplus$ denotes addition modulo three.\footnote{This is the reason
  why we reduce torus tiling instead of square tiling: in a square,
  grid nodes on the upper and right border are missing the required
  successors.} We call such nodes g-active and add the following to
$\Tmc_1$:
$$
  \mn{gactive}^t \sqcap \exists S . \mn{hactive}^{t \oplus 1}
\sqsubseteq
\bigsqcup _{T_i \in \tiles} T_i \qquad  t  \in \{0,1,2\}
$$
We next add inclusions to $\Tmc_1$ which identify a defect in the
tiling and signal this by making the concept name $D$ true:
\begin{enumerate}

\item horizontally neighboring tiles match;
  for all $T_i,T_j \in \tiles$ with $(T_i,T_j) \notin H$
  and $t \in \{0,1,2\}$:
  $$
  T_i \sqcap \mn{gactive}^t \sqcap \exists S . (\mn{hactive}^{t \oplus 
    1} \sqcap \exists S . (\mn{gactive}^{t \oplus 2} \sqcap T_j))
  \sqsubseteq D 
  $$

\item vertically neighboring tiles match;
  for all $T_i,T_j \in \tiles$ with $(T_i,T_j) \notin V$
    and $t \in \{0,1,2\}$:
  $$
  T_i \sqcap \mn{gactive}^t \sqcap \exists S . (\mn{vactive}^{t \oplus 
    1} \sqcap \exists S . (\mn{gactive}^{t \oplus 2} \sqcap T_j)) 
  \sqsubseteq D 
  $$

\item The tiling respects the initial condition. Let $w_0 = T_{i_0}
  \cdots T_{i_{n-1}}$.  As in the previous reduction, it is tedious
  but not difficult to write concept inclusions to be included in \Tmc
  which ensure that, for $0 \leq i < n$, every element that is in
  $\mn{gactive}^t$ for some $t$ and whose $B_\ell$-value
  represents horizontal position $i$ and vertical position 0,
  satisfies the concept name $\mn{pos}_{i,0}$. Here, $(t,\ell)$ ranges
  over $(0,1),(1,3),(2,5)$.  We then add the following CQ to $q_2$ for
  $0 \leq j < n$ and all $T_\ell \in \tiles$ with $T_\ell \neq
  T_{i_j}$:
  $$
  \mn{pos}_{j,0} \sqcap T_\ell \sqsubseteq D.
  $$

\end{enumerate}
Note that Points~1 and~2 achieve the desired cycling through the three
different types of counting trees: horizontal and vertical
neighborships of g-active nodes whose types are not as expected are
simply ignored (i.e., not treated as neighborships in the first place).

This completes the construction of the TBox $\Tmc_1$. The query $q_1$
simply takes the form $\exists x \, D(x)$, thus $Q_1$ is true in an
instance $I$ iff the (potentially partial) grid in $I$ does not admit
a tiling, as desired. The construction of $\Tmc_2$ is also finished at
this point. It thus remains to construct $q_2$. As in the previous
reduction, the purpose of $q_2$ is to ensure that counter values
are copied appropriately to neighboring counting trees and that the
two counter values below each grid and step node are unique. We call
two counting trees \emph{neighboring} if their roots are connected by
the relation $S$. Since $S$ is symmetric, we cannot distinguish
successor counting trees from predecessor ones. The three different
types of counting trees still allow us to achieve the desired copying
of counter values. More precisely, we need to ensure that
\begin{description}

\item[(Q1)] the $B_i$-value of a counting tree coincides with the
  $B_{i+3}$-value of neighboring trees, for all $i \in \{1,2,3\}$;

\item[(Q2)] every g-active node is associated (via counting trees)
  with at most one $B_i$-value, for each $i \in \{1,\dots,6\}$.

\end{description}
\begin{figure}[t!]
  \begin{center}
    \framebox[1\columnwidth]{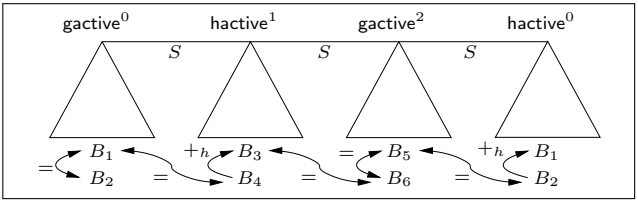}
    \caption{Counting strategy.}
    \label{fig:tree}
  \end{center}
\end{figure}
The counting strategy is illustrated in Figure~\ref{fig:tree}, 
displaying a horizontal fragment of the grid. Arrows annotated with
``$=$'' indicate identical counter values and arrows annotated with
``$+_h$'' indicate incrementation of the horizontal component of the
counter. A vertical fragment would look identical, except that
$\mn{hactive}^i$ is replaced with $\mn{gactive}^i$ and incrementation
of the horizontal counter component with incrementation of the
vertical component.

Before we start constructing $q_2$, we observe that it actually
suffices to ensure (Q1) because (Q2) is then guaranteed automatically.
The reason is that we are only interested in nodes that are g-active
and thus have an $S$-neighbor which is labeled with $\mn{hactive}^{t
  \oplus 1}$ or with $\mn{vactive}^{t \oplus 1}$. Assume for example
that a node $a$ is in $\mn{gactive}^0$ and has an $S$-neighbor $b$
which is labeled with $\mn{hactive}^1$. By (Q1), \emph{all} leaves
with a given position in the tree below $a$ must agree regarding their
$B_1,\overline{B}_1$-value with the $B_4,\overline{B}_4$-value of
\emph{all} leaves with the same position in the tree below $b$. Since
both trees contain at least one leaf for each position, this achieves
(Q2) for the $B_1$-value. It also achieves (Q2) for the $B_2$-value
since that value is identical to $B_1$-value. Obviously, the other
cases are analogous. 

The UCQ $q_1$ for achieving (Q1) includes six CQs. We first construct
a query $q$ which is true in an instance $I$ if
\begin{itemize}

\item[($*$)] there are two leaves in neighboring counting trees that
  have the same position and such that one leaf is labeled with $B_1$
  and the other one with $\overline{B}_4$.

\end{itemize} 
The other five queries are then minor variations of~$q$.
\begin{figure}[t!]
  \begin{center}
    \framebox[1\columnwidth]{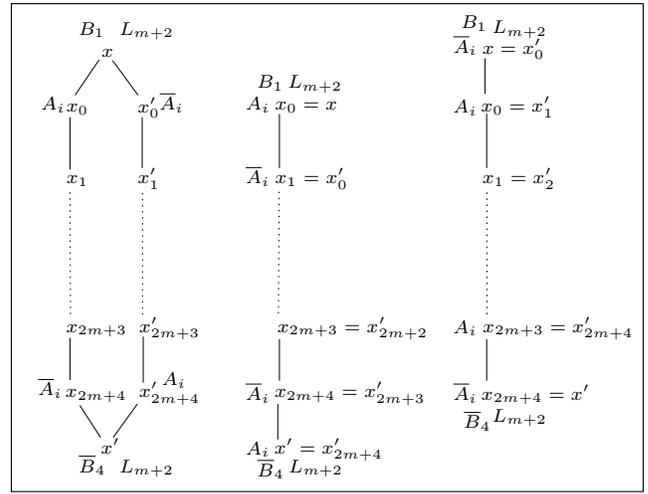}
    \caption{Component query for (Q1) and two 
identifications.}
    \label{fig:q1}
  \end{center}
\end{figure}
We construct $q$ from component queries $p_0,\dots,p_{m-1}$, which all
take the form of the query show on the left-hand side of
Figure~\ref{fig:q1}. Note that all edges are $S$-edges and that the
only difference between the component queries is which concept names
$A_i$ and $\overline{A}_i$ are used. We assemble $p_0,\dots,p_{m-1}$
into the desired query $q$ by taking variable disjoint copies of
$p_0,\dots,p_{m-1}$ and then identifying (i)~the $x$-variables of all
components and (ii)~the $x'$-variables of all components. 

To see why $q$ achieves~($*$), first note that the variables $x$ and
$x'$ must be mapped to leaves of counting trees because of their
$L_{m+2}$-label. Call these leaves $a$ and~$a'$. Since $x$ is labeled
with $B_1$ and $x'$ with $\overline{B}_4$, $a$ and $a'$ must be in
different trees. Since they are connected to $x$ in the query, both
$x_0$ and $x'_0$ must then be mapped either to $a$ or to its
predecessor; likewise, $x_{2m+4}$ and $x'_{2m+4}$ must be mapped
either to $a'$ or to its predecessor.  Because of the labeling of $a$
and $a'$ and the predecessors in the counting tree with $A_i$ and
$\overline{A}_i$, we are actually even more constrained: exactly one
of $x_0$ and $x'_0$ must be mapped to $a$, and exactly one of
$x_{2m+4}$ and $x'_{2m+4}$ to~$a'$. Since the paths between leaves in
different trees in the instance have length at least $2m+5$ and $q$
contains paths from $x_0$ to $x_{2m+4}$ and from $x'_0$ to $x'_{2m+4}$
of length $2m+4$, only the following cases are possible:
\begin{itemize}

\item $x_0$ is mapped to $a$, $x'_0$ to the predecessor of $a$,
  $x'_{2m+4}$ to $a'$, and $x_{2m+4}$ to the predecessor of $a'$;

\item $x'_0$ is mapped to $a$, $x_0$ to the predecessor of $a$,
  $x_{2m+4}$ to $a'$, and $x'_{2m+4}$ to the predecessor of $a'$.

\end{itemize}
This gives rise to the two variable identifications in each query
$p_i$ shown in Figure~\ref{fig:q1}. Note that the first case implies
that $a$ and $a'$ are both labeled with $A_i$ while they are both
labeled with $\overline{A}_i$ in the second case. In summary, $a$ and
$a'$ must thus agree on all concept names $A_i$, $\overline{A}_i$.
Note that with the identification $x_0=x$ (resp.\ $x'_0=x$), there is
a path from $x$ to $x'$ in the query of length $2m+5$. Thus, $a$ and
$a'$ are in neighboring counting trees. Since $a$ must satisfy $B_1$
and $a'$ must satisfy $\overline{B}_4$ due to the labeling of $x$ and
$x'$, we have achieved ($*$).

\medskip

We now show how to replace the UCQ used in the reduction with a 
CQ. This requires the following changes:
\begin{enumerate}

\item every node on level $m$ now has two successors instead of one
  (while nodes on level $m+1$ still have a single successor); in
  $\mn{gactive}^0$-trees, one of the leafs below the same level $m$
  node carries the $B_1$,$\overline{B_1}$-label while the other leaf
  carries the $B_2$,$\overline{B_2}$-label; the labeling of the two
  leaves with $A_i$,$\overline{A}_i$ is identical; similarly for
  $\mn{gactive}^1$ and $\mn{gactive}^2$;

\item the predecessors of leaf nodes in counting trees receive
  additional labels: when the leaf node is labeled with $B_i$ (resp.\
  $\overline{B}_i$), then its predecessor is labeled with
  $\overline{B}_i$ (resp.\ $B_i$) and with $B_j$ and
  $\overline{B}_j$ for all \mbox{$j \in \{ 1,\dots,6\} \setminus \{i\}$};
  these concept names are added to $\Sbf_E$;

\item the roots of counting trees receive an additional label
  $R_0$ or $R_1$, alternating with neighboring trees; these
  concept names are added to $\Sbf_E$, too;

\item the query construction is modified.

\end{enumerate}
Points~2 and~3 are important for the CQ to be constructed to work
correctly. Point~2 does not make sense without Point~1.
Note that Points~1 to~3 can be achieved in a straightforward way by modifying
the previous reduction, details are omitted. We thus concentrate on
Point~4. The desired CQ $q$ is again constructed from component
queries. We use $m$ components as shown in Figure~\ref{fig:q1}, except
that the $B_1$ and $\overline{B}_4$-labels are dropped. We additionally
take the disjoint union with the component (partially) shown in
Figure~\ref{fig:cq} where again $x$ and $x'$ are the variables shared with
the other components. 
\begin{figure}[t!]
  \begin{center}
    \framebox[1\columnwidth]{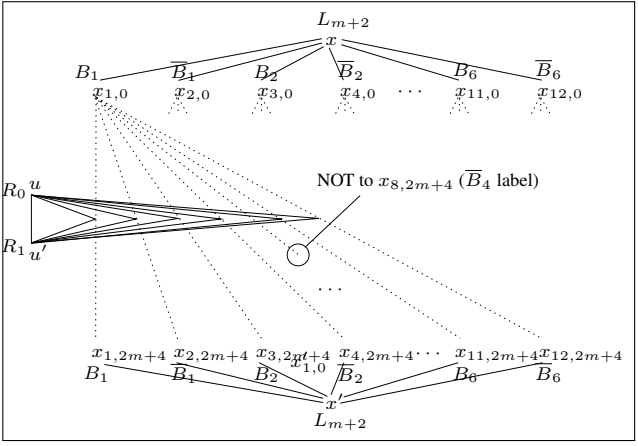}
    \caption{Additional component for CQ.}
    \label{fig:cq}
  \end{center}
\end{figure}
The dotted edges denote $S$-paths of length $2m+4$. For readability,
we show only some of the paths. The general scheme is that every
variable $x_{i,0}$ has a path to every variable $x_{j,2m+4}$ unless
the two variables are labeled with complementary concept names, that
is, with concept names $B_i$ and $\overline{B}_j$ such that $i \in
\{1,2,3\}$ and $j=i+1$ or with concept names $\overline{B}_i$ and
$B_j$ such that $i \in \{1,2,3\}$ and $j=i+1$.  In the figure, we only
show the paths outgoing from $x_{1,0}$.  The edges that connect $u$
and $u'$ with the dotted paths always end at the middle point of a
path, which has distance $m+2$ to the $x_{i,0}$ variable where the
path starts and also distance $m+2$ to the $x_{j,2m+4}$ variable where
it ends.

We have to argue that the CQ $q$ just constructed achieves (Q1).  As
before, $x$ and $x'$ must be mapped to leafs of counting trees because
of their $L_{m+2}$-label. Call these leaves $a_1$ and $a_2$. All
$x_{i,0}$ must then be mapped to $a$ or its predecessor, and all
$x_{i,2m+4}$ must be mapped to $a'$ or its predecessor. In fact, due
to the labeling of $a$ and $a'$ and their predecessors in their counting
tree with the concepts $B_i$, $\overline{B}_i$, exactly one variable
$x_{i,0}$ from $x_{0,0},\dots,x_{12,0}$ is mapped to $a$ while all
others are mapped to the predecessor of $a$; likewise, exactly one of
the $x_{j,2m+4}$ from $x_{0,2m+4},\dots,x_{12,2m+4}$ is mapped to $a'$
while all others are mapped to the predecessor of $a'$. To achieve
(Q1), we have to argue that $x_{i,0}$ and $x_{j,2m+4}$ are labeled
with complementary concept names, and that $a$ and $a'$ are in
neighboring computation trees.

Assume to the contrary that $x_{i,0}$ and $x_{j,2m+4}$ are not labeled with
complementary concept names. Then they are connected in $q$ by a path
of length $2m+4$ whose middle point $y$ is connected to the variables
$u$ and $u'$. In a homomorphism to the grid with counting trees, there
are four possible targets for $u$ and $u'$ and for the predecessor
$y_{-1}$ of $y$ on the connecting path and the successor $y_1$ of $y$
on that path:
\begin{enumerate}

\item $u,y_{-1}$ map to the same constant, and so do $u'$ and $y$;

\item $u,y$ map to the same constant, and so do $u'$ and $y_1$;

\item $u',y_{-1}$ map to the same constant, and so do $u$ and $y$;

\item $u',y$ map to the same constant, and so do $u$ and $y_1$.

\end{enumerate}
However, options~1 and~3 are impossible because there would have to be
a path of length $2m+1$ from a node labeled $R_0$ or $R_1$ to the leaf
$a$. Similarly, options~2 and~4 are impossible because there would
have to be a path of length $2m+1$ from a node labeled $R_0$ or $R_1$
to the leaf $a'$. Thus, we have shown that $x_{i,0}$ and $x_{j,2m+4}$
are labeled with complementary concept names. 

This together with the labeling scheme of Figure~\ref{fig:tree} also
means that $a$ and $a'$ (to which $x_{i,0}$ and $x_{j,2m+4}$ are
mapped) are not in the same counting tree. Moreover, they cannot be in
counting trees that are further apart than one step because under the
assumption that $x=x_{i,0}$ and $x'=x_{j,2m+4}$, there is a path of
length $2m+5$ in the query from $x$ to $x'$. Note that
we can identify $u$ with the $2m+2$nd variable on any such path and
$u'$ with the $2m+3$rd variable (or vice versa) to admit a match in
neighboring counting trees.

\bigskip
An OMQ $Q=(\Tmc,\Sbf_E,q)$ is \emph{FO-rewritable} iff there is 
an FO-query that is equivalent to $Q$. Rewritability into monadic
Datalog and into unrestricted Datalog are defined accordingly.

\THMrewriteOMQ*

\begin{proof}
  For $(\ALC,\text{UCQ})$, it suffices to note that every MDDLog program with
  only unary and binary EDB relations can be translated into an
  equivalent OMQ from $(\ALC,\text{UCQ})$ in polynomial time
  \cite{DBLP:journals/tods/BienvenuCLW14}. Thus, the lower bounds for
  $(\ALC,UCQ)$ are an immediate consequence of
  Theorem~\ref{thm:rewrite}. 

  For $(\ALCI,\text{CQ})$, we adapt the above hardness proof for
  containment, essentially in the same way as in the proof of
  Theorem~\ref{thm:rewrite}. Our aim is thus to show that, from a
  2-exp torus tiling problem $P$ and an input $w_0$ to~$P$, we can
  construct in polynomial time an $(\ALCI,\text{CQ})$-OMQ $Q$ such
  that
  \begin{enumerate}

  \item if there is a tiling for $P$ and $w_0$, then $Q$ is
    FO-rewritable;

  \item if there is no tiling for $P$ and $w_0$, then $Q$ is
    not Datalog-rewritable.

  \end{enumerate}
  We have shown how to construct from a 2-exp torus tiling problem $P$
  and an input $w_0$ to $P$, two $(\ALCI,\text{CQ})$-OMQs $Q_1$, $Q_2$
  such that $Q_1 \subseteq Q_2$ iff there is a tiling for $P$ and
  $w_0$. Moreover, $Q_2$ consists only of inclusions of the form
  $C \sqsubseteq \bot$ with $C$ an \ELI-concept (a concept built only
  from conjuntion and existential restrictions, possibly using inverse
  roles). As above, let $Q_i = (\Tmc_i,\Sbf_E,q_i)$ for
  $i \in \{1,2\}$. The desired OMQ $Q=(\Tmc,\Sbf'_E,q)$ is constructed
  by choosing $q=q_2$, $\Sbf'_E=\Sbf_E \cup \{ s, u \}$, and
  choosing for \Tmc the union of $\Tmc_1$ and $\Tmc_2$, extended with
  the following CIs:
  \begin{enumerate}

  \item $\exists u . D \sqsubseteq R \sqcup G \sqcup B$ (where $D$ is the concept name 
    used in $q_1$);

  \item $C_1 \sqcap C_2 \sqsubseteq C_{q_2}$ for all distinct
    $C_1,C_2 \in \{R,G,B\}$ where $C_{q_2}$ is an (easy to construct)
    \ELI-concept such that $C^\Imc_{q_2} \neq \emptyset$
    implies $\Imc \models q_2$ for all interpretations \Imc;
  
  \item $C \sqcap \exists s . C \sqsubseteq C_{q_2}$ for all $C \in
    \{R,G,B\}$.

\end{enumerate}
We now show that $Q$ satisfies Points~1 and~2 above.

For Point~1, assume that there is a tiling for $P$ and~$w_0$.  Then
$Q_1 \subseteq Q_2$. We claim that we obtain a UCQ-rewriting $\varphi$
of $Q$ by taking the disjunction of $q_2$ and of $\exists x \, C(x)$
for every CI $C \sqsubseteq \bot$ in $\Tmc_2$. To see this, let $I$ be
an $\Sbf'_E$-instance. Clearly, $I \models \varphi$ implies
$I \models Q$. Conversely, assume that $I \models Q$. If
$I \not\models Q_1$, then there is a model \Imc of $I$ and $\Tmc_1$
such that $D^\Imc = \emptyset$. Thus the additional CIs from the
construction of $Q$ are inactive and $I \models Q$ implies
$I \models Q_2$, thus $I \models \varphi$.  Now assume
$I \models Q_1$. Then $I \models Q_2$ since $Q_1 \subseteq Q_2$. Thus,
again, $I \models \varphi$.

For Point~2, assume there is no tiling for $P$ and $w_0$. Then
$Q_1 \not\subseteq Q_2$. Given an undirected graph $G=(V,E)$, let the
instance $I^+_G$ be defined as the disjoint union of the instance
$I_0$ which represents the $2^{2^n}$-grid plus counting gadgets, the
instance $I_G$ which contains a fact $s(v_1,v_2)$ for every
$\{v_1,v_2\} \in V$, extended with the fact $u(v,g)$ for every
$v \in V$ and element $g$ of $I_0$.

Since there is no tiling for $P$ and $w_0$, we have $I_0 \models Q_1$
and thus $I^+_G \models Q_1$. By construction of $Q$ and since
$Q_1 \not\subseteq Q_2$, this implies that $I^+_G \models Q$ iff $G$
is not 3-colorable. It remains to argue that, consequently, a
Datalog-rewriting of $Q$ gives rise to a Datalog-rewriting of
non-3-colorability (which doesn't exist). This can be done as in the
proof of Theorem~\ref{thm:rewrite}.
\end{proof}

\end{document}

%% file: tree.pdf_t
\begin{picture}(0,0)%
\includegraphics{tree.pdf}%
\end{picture}%
\setlength{\unitlength}{2403sp}%
\begingroup\makeatletter\ifx\SetFigFont\undefined%
\gdef\SetFigFont#1#2#3#4#5{%
  \reset@font\fontsize{#1}{#2pt}%
  \fontfamily{#3}\fontseries{#4}\fontshape{#5}%
  \selectfont}%
\fi\endgroup%
\begin{picture}(5825,1838)(4364,-3033)
\put(8912,-2651){\makebox(0,0)[lb]{\smash{{\SetFigFont{7}{8.4}{\familydefault}{\mddefault}{\updefault}{\color[rgb]{0,0,0}$+_h$}%
}}}}
\put(4636,-1366){\makebox(0,0)[lb]{\smash{{\SetFigFont{7}{8.4}{\familydefault}{\mddefault}{\updefault}{\color[rgb]{0,0,0}$\mn{gactive}^0$}%
}}}}
\put(4906,-2691){\makebox(0,0)[lb]{\smash{{\SetFigFont{7}{8.4}{\familydefault}{\mddefault}{\updefault}{\color[rgb]{0,0,0}$B_1$}%
}}}}
\put(4906,-2964){\makebox(0,0)[lb]{\smash{{\SetFigFont{7}{8.4}{\familydefault}{\mddefault}{\updefault}{\color[rgb]{0,0,0}$B_2$}%
}}}}
\put(6156,-1366){\makebox(0,0)[lb]{\smash{{\SetFigFont{7}{8.4}{\familydefault}{\mddefault}{\updefault}{\color[rgb]{0,0,0}$\mn{hactive}^1$}%
}}}}
\put(7682,-1366){\makebox(0,0)[lb]{\smash{{\SetFigFont{7}{8.4}{\familydefault}{\mddefault}{\updefault}{\color[rgb]{0,0,0}$\mn{gactive}^2$}%
}}}}
\put(9222,-1366){\makebox(0,0)[lb]{\smash{{\SetFigFont{7}{8.4}{\familydefault}{\mddefault}{\updefault}{\color[rgb]{0,0,0}$\mn{hactive}^0$}%
}}}}
\put(6432,-2691){\makebox(0,0)[lb]{\smash{{\SetFigFont{7}{8.4}{\familydefault}{\mddefault}{\updefault}{\color[rgb]{0,0,0}$B_3$}%
}}}}
\put(6433,-2964){\makebox(0,0)[lb]{\smash{{\SetFigFont{7}{8.4}{\familydefault}{\mddefault}{\updefault}{\color[rgb]{0,0,0}$B_4$}%
}}}}
\put(7972,-2691){\makebox(0,0)[lb]{\smash{{\SetFigFont{7}{8.4}{\familydefault}{\mddefault}{\updefault}{\color[rgb]{0,0,0}$B_5$}%
}}}}
\put(7973,-2964){\makebox(0,0)[lb]{\smash{{\SetFigFont{7}{8.4}{\familydefault}{\mddefault}{\updefault}{\color[rgb]{0,0,0}$B_6$}%
}}}}
\put(9492,-2691){\makebox(0,0)[lb]{\smash{{\SetFigFont{7}{8.4}{\familydefault}{\mddefault}{\updefault}{\color[rgb]{0,0,0}$B_1$}%
}}}}
\put(9492,-2964){\makebox(0,0)[lb]{\smash{{\SetFigFont{7}{8.4}{\familydefault}{\mddefault}{\updefault}{\color[rgb]{0,0,0}$B_2$}%
}}}}
\put(4379,-2851){\makebox(0,0)[lb]{\smash{{\SetFigFont{7}{8.4}{\familydefault}{\mddefault}{\updefault}{\color[rgb]{0,0,0}$=$}%
}}}}
\put(5566,-2945){\makebox(0,0)[lb]{\smash{{\SetFigFont{7}{8.4}{\familydefault}{\mddefault}{\updefault}{\color[rgb]{0,0,0}$=$}%
}}}}
\put(7479,-2705){\makebox(0,0)[lb]{\smash{{\SetFigFont{7}{8.4}{\familydefault}{\mddefault}{\updefault}{\color[rgb]{0,0,0}$=$}%
}}}}
\put(5719,-1671){\makebox(0,0)[lb]{\smash{{\SetFigFont{7}{8.4}{\familydefault}{\mddefault}{\updefault}{\color[rgb]{0,0,0}$S$}%
}}}}
\put(7259,-1671){\makebox(0,0)[lb]{\smash{{\SetFigFont{7}{8.4}{\familydefault}{\mddefault}{\updefault}{\color[rgb]{0,0,0}$S$}%
}}}}
\put(8766,-1671){\makebox(0,0)[lb]{\smash{{\SetFigFont{7}{8.4}{\familydefault}{\mddefault}{\updefault}{\color[rgb]{0,0,0}$S$}%
}}}}
\put(7093,-2945){\makebox(0,0)[lb]{\smash{{\SetFigFont{7}{8.4}{\familydefault}{\mddefault}{\updefault}{\color[rgb]{0,0,0}$=$}%
}}}}
\put(8673,-2945){\makebox(0,0)[lb]{\smash{{\SetFigFont{7}{8.4}{\familydefault}{\mddefault}{\updefault}{\color[rgb]{0,0,0}$=$}%
}}}}
\put(5879,-2665){\makebox(0,0)[lb]{\smash{{\SetFigFont{7}{8.4}{\familydefault}{\mddefault}{\updefault}{\color[rgb]{0,0,0}$+_h$}%
}}}}
\end{picture}%

%% file: q1.pdf_t
\begin{picture}(0,0)%
\includegraphics{q1.pdf}%
\end{picture}%
\setlength{\unitlength}{2486sp}%
\begingroup\makeatletter\ifx\SetFigFont\undefined%
\gdef\SetFigFont#1#2#3#4#5{%
  \reset@font\fontsize{#1}{#2pt}%
  \fontfamily{#3}\fontseries{#4}\fontshape{#5}%
  \selectfont}%
\fi\endgroup%
\begin{picture}(5796,4700)(5611,-7642)
\put(9871,-7094){\makebox(0,0)[lb]{\smash{{\SetFigFont{7}{8.4}{\familydefault}{\mddefault}{\updefault}{\color[rgb]{0,0,0}$\overline{B}_4$}%
}}}}
\put(6632,-6133){\makebox(0,0)[lb]{\smash{{\SetFigFont{7}{8.4}{\familydefault}{\mddefault}{\updefault}{\color[rgb]{0,0,0}$x'_{2m+3}$}%
}}}}
\put(5905,-6133){\makebox(0,0)[lb]{\smash{{\SetFigFont{7}{8.4}{\familydefault}{\mddefault}{\updefault}{\color[rgb]{0,0,0}$x_{2m+3}$}%
}}}}
\put(5899,-6786){\makebox(0,0)[lb]{\smash{{\SetFigFont{7}{8.4}{\familydefault}{\mddefault}{\updefault}{\color[rgb]{0,0,0}$x_{2m+4}$}%
}}}}
\put(6252,-7349){\makebox(0,0)[lb]{\smash{{\SetFigFont{7}{8.4}{\familydefault}{\mddefault}{\updefault}{\color[rgb]{0,0,0}$x'$}%
}}}}
\put(6861,-6662){\makebox(0,0)[lb]{\smash{{\SetFigFont{7}{8.4}{\familydefault}{\mddefault}{\updefault}{\color[rgb]{0,0,0}$A_i$}%
}}}}
\put(7698,-7355){\makebox(0,0)[lb]{\smash{{\SetFigFont{7}{8.4}{\familydefault}{\mddefault}{\updefault}{\color[rgb]{0,0,0}$A_i$}%
}}}}
\put(10060,-6133){\makebox(0,0)[lb]{\smash{{\SetFigFont{7}{8.4}{\familydefault}{\mddefault}{\updefault}{\color[rgb]{0,0,0}$x_{2m+3}=x'_{2m+4}$}%
}}}}
\put(9764,-6142){\makebox(0,0)[lb]{\smash{{\SetFigFont{7}{8.4}{\familydefault}{\mddefault}{\updefault}{\color[rgb]{0,0,0}$A_i$}%
}}}}
\put(7981,-7349){\makebox(0,0)[lb]{\smash{{\SetFigFont{7}{8.4}{\familydefault}{\mddefault}{\updefault}{\color[rgb]{0,0,0}$x'=x'_{2m+4}$}%
}}}}
\put(8011,-6133){\makebox(0,0)[lb]{\smash{{\SetFigFont{7}{8.4}{\familydefault}{\mddefault}{\updefault}{\color[rgb]{0,0,0}$x_{2m+3}=x'_{2m+2}$}%
}}}}
\put(7988,-3913){\makebox(0,0)[lb]{\smash{{\SetFigFont{7}{8.4}{\familydefault}{\mddefault}{\updefault}{\color[rgb]{0,0,0}$x_0=x$}%
}}}}
\put(7697,-3916){\makebox(0,0)[lb]{\smash{{\SetFigFont{7}{8.4}{\familydefault}{\mddefault}{\updefault}{\color[rgb]{0,0,0}$A_i$}%
}}}}
\put(10055,-4651){\makebox(0,0)[lb]{\smash{{\SetFigFont{7}{8.4}{\familydefault}{\mddefault}{\updefault}{\color[rgb]{0,0,0}$x_1=x'_2$}%
}}}}
\put(10054,-3360){\makebox(0,0)[lb]{\smash{{\SetFigFont{7}{8.4}{\familydefault}{\mddefault}{\updefault}{\color[rgb]{0,0,0}$x=x'_0$}%
}}}}
\put(10042,-3938){\makebox(0,0)[lb]{\smash{{\SetFigFont{7}{8.4}{\familydefault}{\mddefault}{\updefault}{\color[rgb]{0,0,0}$x_0=x'_1$}%
}}}}
\put(9763,-3943){\makebox(0,0)[lb]{\smash{{\SetFigFont{7}{8.4}{\familydefault}{\mddefault}{\updefault}{\color[rgb]{0,0,0}$A_i$}%
}}}}
\put(7989,-4651){\makebox(0,0)[lb]{\smash{{\SetFigFont{7}{8.4}{\familydefault}{\mddefault}{\updefault}{\color[rgb]{0,0,0}$x_1=x'_0$}%
}}}}
\put(5906,-3938){\makebox(0,0)[lb]{\smash{{\SetFigFont{7}{8.4}{\familydefault}{\mddefault}{\updefault}{\color[rgb]{0,0,0}$x_0$}%
}}}}
\put(5906,-4651){\makebox(0,0)[lb]{\smash{{\SetFigFont{7}{8.4}{\familydefault}{\mddefault}{\updefault}{\color[rgb]{0,0,0}$x_1$}%
}}}}
\put(6259,-3388){\makebox(0,0)[lb]{\smash{{\SetFigFont{7}{8.4}{\familydefault}{\mddefault}{\updefault}{\color[rgb]{0,0,0}$x$}%
}}}}
\put(5657,-3915){\makebox(0,0)[lb]{\smash{{\SetFigFont{7}{8.4}{\familydefault}{\mddefault}{\updefault}{\color[rgb]{0,0,0}$A_i$}%
}}}}
\put(6619,-4651){\makebox(0,0)[lb]{\smash{{\SetFigFont{7}{8.4}{\familydefault}{\mddefault}{\updefault}{\color[rgb]{0,0,0}$x'_1$}%
}}}}
\put(6838,-3915){\makebox(0,0)[lb]{\smash{{\SetFigFont{7}{8.4}{\familydefault}{\mddefault}{\updefault}{\color[rgb]{0,0,0}$\overline{A}_i$}%
}}}}
\put(6632,-6786){\makebox(0,0)[lb]{\smash{{\SetFigFont{7}{8.4}{\familydefault}{\mddefault}{\updefault}{\color[rgb]{0,0,0}$x'_{2m+4}$}%
}}}}
\put(6619,-3945){\makebox(0,0)[lb]{\smash{{\SetFigFont{7}{8.4}{\familydefault}{\mddefault}{\updefault}{\color[rgb]{0,0,0}$x'_0$}%
}}}}
\put(7696,-6811){\makebox(0,0)[lb]{\smash{{\SetFigFont{7}{8.4}{\familydefault}{\mddefault}{\updefault}{\color[rgb]{0,0,0}$\overline{A}_i$}%
}}}}
\put(5626,-6766){\makebox(0,0)[lb]{\smash{{\SetFigFont{7}{8.4}{\familydefault}{\mddefault}{\updefault}{\color[rgb]{0,0,0}$\overline{A}_i$}%
}}}}
\put(7985,-6786){\makebox(0,0)[lb]{\smash{{\SetFigFont{7}{8.4}{\familydefault}{\mddefault}{\updefault}{\color[rgb]{0,0,0}$x_{2m+4}=x'_{2m+3}$}%
}}}}
\put(10047,-6809){\makebox(0,0)[lb]{\smash{{\SetFigFont{7}{8.4}{\familydefault}{\mddefault}{\updefault}{\color[rgb]{0,0,0}$x_{2m+4}=x'$}%
}}}}
\put(9766,-6811){\makebox(0,0)[lb]{\smash{{\SetFigFont{7}{8.4}{\familydefault}{\mddefault}{\updefault}{\color[rgb]{0,0,0}$\overline{A}_i$}%
}}}}
\put(9766,-3346){\makebox(0,0)[lb]{\smash{{\SetFigFont{7}{8.4}{\familydefault}{\mddefault}{\updefault}{\color[rgb]{0,0,0}$\overline{A}_i$}%
}}}}
\put(7696,-4651){\makebox(0,0)[lb]{\smash{{\SetFigFont{7}{8.4}{\familydefault}{\mddefault}{\updefault}{\color[rgb]{0,0,0}$\overline{A}_i$}%
}}}}
\put(6436,-3166){\makebox(0,0)[lb]{\smash{{\SetFigFont{7}{8.4}{\familydefault}{\mddefault}{\updefault}{\color[rgb]{0,0,0}$L_{m+2}$}%
}}}}
\put(6436,-7531){\makebox(0,0)[lb]{\smash{{\SetFigFont{7}{8.4}{\familydefault}{\mddefault}{\updefault}{\color[rgb]{0,0,0}$L_{m+2}$}%
}}}}
\put(6031,-3166){\makebox(0,0)[lb]{\smash{{\SetFigFont{7}{8.4}{\familydefault}{\mddefault}{\updefault}{\color[rgb]{0,0,0}$B_1$}%
}}}}
\put(6031,-7531){\makebox(0,0)[lb]{\smash{{\SetFigFont{7}{8.4}{\familydefault}{\mddefault}{\updefault}{\color[rgb]{0,0,0}$\overline{B}_4$}%
}}}}
\put(8116,-3699){\makebox(0,0)[lb]{\smash{{\SetFigFont{7}{8.4}{\familydefault}{\mddefault}{\updefault}{\color[rgb]{0,0,0}$L_{m+2}$}%
}}}}
\put(7811,-3700){\makebox(0,0)[lb]{\smash{{\SetFigFont{7}{8.4}{\familydefault}{\mddefault}{\updefault}{\color[rgb]{0,0,0}$B_1$}%
}}}}
\put(8123,-7532){\makebox(0,0)[lb]{\smash{{\SetFigFont{7}{8.4}{\familydefault}{\mddefault}{\updefault}{\color[rgb]{0,0,0}$L_{m+2}$}%
}}}}
\put(10156,-3153){\makebox(0,0)[lb]{\smash{{\SetFigFont{7}{8.4}{\familydefault}{\mddefault}{\updefault}{\color[rgb]{0,0,0}$L_{m+2}$}%
}}}}
\put(9844,-3101){\makebox(0,0)[lb]{\smash{{\SetFigFont{7}{8.4}{\familydefault}{\mddefault}{\updefault}{\color[rgb]{0,0,0}$B_1$}%
}}}}
\put(10156,-7027){\makebox(0,0)[lb]{\smash{{\SetFigFont{7}{8.4}{\familydefault}{\mddefault}{\updefault}{\color[rgb]{0,0,0}$L_{m+2}$}%
}}}}
\put(7811,-7573){\makebox(0,0)[lb]{\smash{{\SetFigFont{7}{8.4}{\familydefault}{\mddefault}{\updefault}{\color[rgb]{0,0,0}$\overline{B}_4$}%
}}}}
\end{picture}%

%% file: cq.pdf_t
\begin{picture}(0,0)%
\includegraphics{cq.pdf}%
\end{picture}%
\setlength{\unitlength}{2279sp}%
\begingroup\makeatletter\ifx\SetFigFont\undefined%
\gdef\SetFigFont#1#2#3#4#5{%
  \reset@font\fontsize{#1}{#2pt}%
  \fontfamily{#3}\fontseries{#4}\fontshape{#5}%
  \selectfont}%
\fi\endgroup%
\begin{picture}(6946,4593)(4898,-7600)
\put(5216,-5728){\makebox(0,0)[lb]{\smash{{\SetFigFont{7}{8.4}{\familydefault}{\mddefault}{\updefault}{\color[rgb]{0,0,0}$u'$}%
}}}}
\put(5899,-6786){\makebox(0,0)[lb]{\smash{{\SetFigFont{7}{8.4}{\familydefault}{\mddefault}{\updefault}{\color[rgb]{0,0,0}$x_{1,2m+4}$}%
}}}}
\put(5906,-3938){\makebox(0,0)[lb]{\smash{{\SetFigFont{7}{8.4}{\familydefault}{\mddefault}{\updefault}{\color[rgb]{0,0,0}$x_{1,0}$}%
}}}}
\put(5710,-3741){\makebox(0,0)[lb]{\smash{{\SetFigFont{7}{8.4}{\familydefault}{\mddefault}{\updefault}{\color[rgb]{0,0,0}$B_1$}%
}}}}
\put(8059,-6891){\makebox(0,0)[lb]{\smash{{\SetFigFont{7}{8.4}{\familydefault}{\mddefault}{\updefault}{\color[rgb]{0,0,0}$x'_{1,0}$}%
}}}}
\put(8445,-3388){\makebox(0,0)[lb]{\smash{{\SetFigFont{7}{8.4}{\familydefault}{\mddefault}{\updefault}{\color[rgb]{0,0,0}$x$}%
}}}}
\put(8446,-7349){\makebox(0,0)[lb]{\smash{{\SetFigFont{7}{8.4}{\familydefault}{\mddefault}{\updefault}{\color[rgb]{0,0,0}$x'$}%
}}}}
\put(8343,-3166){\makebox(0,0)[lb]{\smash{{\SetFigFont{7}{8.4}{\familydefault}{\mddefault}{\updefault}{\color[rgb]{0,0,0}$L_{m+2}$}%
}}}}
\put(8309,-7531){\makebox(0,0)[lb]{\smash{{\SetFigFont{7}{8.4}{\familydefault}{\mddefault}{\updefault}{\color[rgb]{0,0,0}$L_{m+2}$}%
}}}}
\put(10765,-6786){\makebox(0,0)[lb]{\smash{{\SetFigFont{7}{8.4}{\familydefault}{\mddefault}{\updefault}{\color[rgb]{0,0,0}$x_{12,2m+4}$}%
}}}}
\put(9851,-6786){\makebox(0,0)[lb]{\smash{{\SetFigFont{7}{8.4}{\familydefault}{\mddefault}{\updefault}{\color[rgb]{0,0,0}$x_{11,2m+4}$}%
}}}}
\put(7685,-6786){\makebox(0,0)[lb]{\smash{{\SetFigFont{7}{8.4}{\familydefault}{\mddefault}{\updefault}{\color[rgb]{0,0,0}$x_{3,2m+4}$}%
}}}}
\put(6792,-6786){\makebox(0,0)[lb]{\smash{{\SetFigFont{7}{8.4}{\familydefault}{\mddefault}{\updefault}{\color[rgb]{0,0,0}$x_{2,2m+4}$}%
}}}}
\put(10747,-3938){\makebox(0,0)[lb]{\smash{{\SetFigFont{7}{8.4}{\familydefault}{\mddefault}{\updefault}{\color[rgb]{0,0,0}$x_{12,0}$}%
}}}}
\put(9852,-3938){\makebox(0,0)[lb]{\smash{{\SetFigFont{7}{8.4}{\familydefault}{\mddefault}{\updefault}{\color[rgb]{0,0,0}$x_{11,0}$}%
}}}}
\put(8599,-3938){\makebox(0,0)[lb]{\smash{{\SetFigFont{7}{8.4}{\familydefault}{\mddefault}{\updefault}{\color[rgb]{0,0,0}$x_{4,0}$}%
}}}}
\put(7700,-3938){\makebox(0,0)[lb]{\smash{{\SetFigFont{7}{8.4}{\familydefault}{\mddefault}{\updefault}{\color[rgb]{0,0,0}$x_{3,0}$}%
}}}}
\put(6799,-3945){\makebox(0,0)[lb]{\smash{{\SetFigFont{7}{8.4}{\familydefault}{\mddefault}{\updefault}{\color[rgb]{0,0,0}$x_{2,0}$}%
}}}}
\put(9825,-3729){\makebox(0,0)[lb]{\smash{{\SetFigFont{7}{8.4}{\familydefault}{\mddefault}{\updefault}{\color[rgb]{0,0,0}$B_6$}%
}}}}
\put(10725,-3729){\makebox(0,0)[lb]{\smash{{\SetFigFont{7}{8.4}{\familydefault}{\mddefault}{\updefault}{\color[rgb]{0,0,0}$\overline{B}_6$}%
}}}}
\put(6758,-3744){\makebox(0,0)[lb]{\smash{{\SetFigFont{7}{8.4}{\familydefault}{\mddefault}{\updefault}{\color[rgb]{0,0,0}$\overline{B}_1$}%
}}}}
\put(8573,-3737){\makebox(0,0)[lb]{\smash{{\SetFigFont{7}{8.4}{\familydefault}{\mddefault}{\updefault}{\color[rgb]{0,0,0}$\overline{B}_2$}%
}}}}
\put(7658,-3745){\makebox(0,0)[lb]{\smash{{\SetFigFont{7}{8.4}{\familydefault}{\mddefault}{\updefault}{\color[rgb]{0,0,0}$B_2$}%
}}}}
\put(10725,-7028){\makebox(0,0)[lb]{\smash{{\SetFigFont{7}{8.4}{\familydefault}{\mddefault}{\updefault}{\color[rgb]{0,0,0}$\overline{B}_6$}%
}}}}
\put(9224,-3941){\makebox(0,0)[lb]{\smash{{\SetFigFont{7}{8.4}{\familydefault}{\mddefault}{\updefault}{\color[rgb]{0,0,0}$\cdots$}%
}}}}
\put(8354,-6093){\makebox(0,0)[lb]{\smash{{\SetFigFont{7}{8.4}{\familydefault}{\mddefault}{\updefault}{\color[rgb]{0,0,0}$\cdots$}%
}}}}
\put(8354,-4924){\makebox(0,0)[lb]{\smash{{\SetFigFont{7}{8.4}{\familydefault}{\mddefault}{\updefault}{\color[rgb]{0,0,0}NOT to $x_{8,2m+4}$ ($\overline{B}_4$ label)}%
}}}}
\put(4913,-5644){\makebox(0,0)[lb]{\smash{{\SetFigFont{7}{8.4}{\familydefault}{\mddefault}{\updefault}{\color[rgb]{0,0,0}$R_1$}%
}}}}
\put(4913,-5026){\makebox(0,0)[lb]{\smash{{\SetFigFont{7}{8.4}{\familydefault}{\mddefault}{\updefault}{\color[rgb]{0,0,0}$R_0$}%
}}}}
\put(9825,-7028){\makebox(0,0)[lb]{\smash{{\SetFigFont{7}{8.4}{\familydefault}{\mddefault}{\updefault}{\color[rgb]{0,0,0}$B_6$}%
}}}}
\put(8573,-7039){\makebox(0,0)[lb]{\smash{{\SetFigFont{7}{8.4}{\familydefault}{\mddefault}{\updefault}{\color[rgb]{0,0,0}$\overline{B}_2$}%
}}}}
\put(7658,-7029){\makebox(0,0)[lb]{\smash{{\SetFigFont{7}{8.4}{\familydefault}{\mddefault}{\updefault}{\color[rgb]{0,0,0}$B_2$}%
}}}}
\put(6758,-7035){\makebox(0,0)[lb]{\smash{{\SetFigFont{7}{8.4}{\familydefault}{\mddefault}{\updefault}{\color[rgb]{0,0,0}$\overline{B}_1$}%
}}}}
\put(5811,-7026){\makebox(0,0)[lb]{\smash{{\SetFigFont{7}{8.4}{\familydefault}{\mddefault}{\updefault}{\color[rgb]{0,0,0}$B_1$}%
}}}}
\put(9419,-6791){\makebox(0,0)[lb]{\smash{{\SetFigFont{7}{8.4}{\familydefault}{\mddefault}{\updefault}{\color[rgb]{0,0,0}$\cdots$}%
}}}}
\put(8591,-6786){\makebox(0,0)[lb]{\smash{{\SetFigFont{7}{8.4}{\familydefault}{\mddefault}{\updefault}{\color[rgb]{0,0,0}$x_{4,2m+4}$}%
}}}}
\put(5216,-4949){\makebox(0,0)[lb]{\smash{{\SetFigFont{7}{8.4}{\familydefault}{\mddefault}{\updefault}{\color[rgb]{0,0,0}$u$}%
}}}}
\end{picture}%